\documentclass[11pt,a4paper]{amsart}

\usepackage{macros-mtheory}

\DeclareUnicodeCharacter{0306}{\u{\i}}

\begin{document}

\title{Twisted eleven-dimensional supergravity}

\author{Surya Raghavendran}
\address{Perimeter Institute for Theoretical Physics \\ 31 Caroline Street North \\ 
Waterloo, Ontario N2L 2Y5\\ Canada}
\email{sraghavendran@perimeterinstitute.ca}

\author{Ingmar Saberi}
\address{Ludwig-Maximilians-Universit\"at M\"unchen \\ Fakult\"at f\"ur Physik \\ Theresienstra\ss{}e 37 \\ 80333 M\"unchen \\ Deutschland}
\email{i.saberi@physik.uni-muenchen.de}

\author{Brian R. Williams}
\address{School of Mathematics \\ University of Edinburgh \\ Edinburgh EH9 3FD \\ Scotland}
\email{brian.williams@ed.ac.uk}

\begin{abstract}
We construct a fully interacting holomorphic/topological theory in eleven dimensions that is defined on products of Calabi--Yau fivefolds with real one-manifolds. 
The theory describes a particular deformation of the cotangent bundle to the moduli space of Calabi--Yau structures on the fivefold.
Its field content matches the holomorphic (or minimal) twist of the 
eleven-dimensional supergravity multiplet recently computed by the second two authors, and we offer numerous consistency checks showing that the interactions correctly describe interacting twisted eleven-dimensional supergravity at the perturbative level. 
We prove that the global symmetry algebra of our model on flat space is an $L_\infty$ central extension of the infinite-dimensional simple exceptional super Lie algebra $E(5,10)$, following a recent suggestion of Cederwall in the context of the relevant pure spinor model. 
Twists of superconformal algebras map to the fields of our model on the complement of a stack of M2 or M5 branes, laying the groundwork for a fully holomorphic version of twisted holography in this context.
\end{abstract}
\maketitle

\vfill\eject

\setcounter{tocdepth}{1}
\tableofcontents

%
%
%
%
%

Twists of supersymmetric field theories have been a subject of active and fruitful study for many years. 
The notion of twisting was introduced by Witten and dates back to the first example~\cite{WittenTwist}, in which a topological field theory related to Donaldson invariants of four-manifolds was constructed as a twist of four-dimensional $\N=2$ supersymmetric gauge theory.
Shortly afterwards (thirty years ago), the field had expanded to include a diverse set of examples related to topological field theory~\cite{BlauThompson}. 

For many years, twisting constructions were motivated by the goal of producing topological field theories, and correspondingly emphasized the role of particular modifications of  Lorentz symmetry via ``twisting homomorphisms'' designed to produce scalar supercharges.
From a more modern perspective, though, twisting a supersymmetric field theory is just a systematic way of taking invariants with respect to a fixed supercharge. This requires that the supercharge square to zero, so that the supercharge defines an odd abelian subalgebra of the supersymmetry algebra.\footnote{More generally, one could consider ``omega backgrounds'' by taking the derived invariants of a more general subalgebra in which the odd part is not necessarily abelian.}
The variety of appropriate square-zero elements was studied systematically in~\cite{NV}, and has deep relations not only to the  classification of  twists, but to the construction  of supersymmetric field theories via the pure spinor formalism~\cite{Cederwall,EHSW}. The name ``pure spinor'' originates because the  variety of  square-zero elements always contains a minimal stratum that is related to the space of Cartan pure spinors; in turn, these correspond to choices of complex structure on the  spacetime, at least when the dimension is even. 
As such, any $d$-dimensional  supersymmetric theory that admits a twist admits a \emph{minimal} or~\emph{holomorphic} twist, in which $n=\lfloor d/2\rfloor$ of the translations act nontrivially. Other twists can be understood as further twists of this holomorphic theory.
The importance and interest of holomorphic twists of supersymmetric gauge theories was emphasized in the work of Costello~\cite{CostelloHol}, who developed the notion mathematically, building on earlier work on holomorphic field theories in the physics literature~\cite[for example]{NekThesis,BCOV}.
Si Li pushed applications of holomorphic QFT further in his work on the higher genus $B$-model (see~\cite{SiThesis}, and later \cite{CLbcov1} together with Costello). 
In the physics literature, further constructions related to holomorphic twists had appeared in the context of supersymmetric indices; see~\cite{Romelsberger}, for example.

In a supergravity theory, local supersymmetry is already gauged, so that the standard notion of twisting is {\em a priori} not a sensible operation. Twisted supergravity thus needs to be  understood differently.
In their seminal paper~\cite{CLsugra}, Costello and Li gave a definition of twisted supergravity: it is supergravity in a background where the bosonic ghost field representing a twisting supercharge takes a nonzero value.
Twisting a supersymmetric field theory, in the classical sense, can then be  understood as coupling the theory to supergravity and placing it in such a background.

Using ideas related to string field theory and the topological B-model, Costello and Li provided conjectural descriptions of many examples of ten-dimensional theories of twisted supergravity. They used these descriptions to rigorously construct quantizations of open-closed string field theories~\cite{CLbcov1,CLtypeI}. 
A very active and promising direction of current research uses these twisted ten-dimensional models to formulate twisted versions of holography; see~\cite{CostelloM2, Costello_2021, costello2021twisted, Ishtiaque_2020, budzik2021giant, Gaiotto:2020vqj}, for example.
Supergravity in the $\Omega$-background was introduced in~\cite{CostelloM5}, and has been further studied in~\cite{Gaiotto:2019wcc, Oh:2020hph, Oh:2021bwi}. 

In this paper, we provide perturbative descriptions of all twists of eleven-dimensional supergravity (before turning on any $\Omega$-background).
Thanks to a straightforward classification, which we recall below, there are essentially two types of twists which deform the flat background $\RR^{11}$: 
\begin{itemize}[leftmargin=\parindent,itemsep=\parskip]
\item The minimal (or holomorphic) twist. 
The resulting twisted theory is defined on $\CC^5 \times \RR$ and is holomorphic along $\CC^5$ and topological along the `time' direction~$\RR$. 
The twist is $SU(5)$ invariant and involves a fixed choice of a Calabi--Yau form on $\CC^5$. 
\item The non-minimal twist. 
The resulting twisted theory is defined on $\CC^2 \times \RR^7$ and is holomorphic along $\CC^2$ and topological along $\RR^7$.
The twist is $SU(2)$ invariant and involves a fixed choice of a hyperk\"ahler structure on $\CC^2$.  
\end{itemize}

In~\cite{SWspinor}, the last two authors used the pure spinor formalism to describe the {\em free field limit} of the minimal twist of eleven-dimensional supergravity on flat space, building on work of Cederwall~\cite{Ced-towards,Ced-11d} that constructed the eleven-dimensional theory in the pure spinor formalism. In fact, the result of~\cite{SWspinor} shows that the pure spinor formalism is compatible with twisting, so that the twist of any multiplet can be recovered from the algebraic geometry of a neighborhood of the corresponding point in the scheme of square-zero elements.
(In separate work~\cite{EagerHahner}, the free-field limit of the non-minimal twist was derived from the component-field multiplet as used in the physics literature.)
The resulting theory is $\ZZ/2$ graded; the grading is inherited from the totalization of ghost number and intrinsic parity in the original, untwisted, eleven-dimensional theory.
In \S\ref{s:dfn} we introduce a full interacting theory in the Batalin--Vilkovisky (BV) formalism, the free limit of which is the twist computed in \cite{SWspinor}.\footnote{One subtlety is that, unlike the usual BV formalism which involves a $\ZZ$-grading by ghost number, the eleven-dimensional theory is only $\ZZ/2$ graded.
As such, a division of the odd fields in the twisted theory into ``ghosts'' and ``antifields'' is not really meaningful.}
Our main conjecture is that this interacting BV theory is equivalent to the minimal twist of eleven-dimensional supergravity on flat space.

\begin{conj}
The eleven-dimensional  holomorphic-topological theory on $\CC^5 \times \RR$ that we will define in \S \ref{s:dfn} is equivalent to twisted supergravity on $\RR^{11}$, where the bosonic ghost takes a value corresponding to a holomorphic supercharge. 
\end{conj}

In this paper we will provide evidence for this conjecture on a variety of fronts:
\begin{itemize}[leftmargin=\parindent,itemsep=\parskip]
\item 
In \S\ref{sec:locchar}, we show that the character of the local operators in our 11-dimensional theory agrees with the index of 11-dimensional supergravity. 
\item
In \S\ref{sec:susy}, we will show that this theory has all of the expected residual supersymmetries present after performing the holomorphic twist.

\item
In \S\ref{s:nonmin}, we describe the non-minimal twist of eleven-dimensional supergravity as a background of our theory on $\CC^5 \times \RR$. 
This non-minimal twist is invariant for the group $G_2 \times SU(2)$. 
We find a match with a conjectural description of this $G_2 \times SU(2)$ invariant twist formulated by Costello in \cite{CostelloM5} and further developed in~\cite{RY}.  
\item In \S\ref{sec:dimred}, we compute dimensional reductions and show that our model is compatible with descriptions of twists of lower dimensional twists of supergravity. 
For instance, reducing along a circle in a complex plane agrees with the conjectural description of the $SU(4)$ twist of type IIA supergravity in~\cite{CLsugra}.
\item
Finally, in \S\ref{sec:ads}, we propose a description of supergravity in twisted versions of both the ${\rm AdS}_7$ and ${\rm AdS}_4$ backgrounds. 
We show that the residual symmetries of supergravity in these backgrounds are present in our twisted background.
\end{itemize}

\subsection*{A geometric description of the model} 

Our eleven-dimensional theory is defined more generally on the product of manifolds 
\[
X \times S
\]
where $X$ is a Calabi--Yau five-fold and $S$ is a smooth oriented real one-dimensional manifold. 
The theory is of a holomorphic gravitational flavor as it describes ``partial'' deformations of complex structures along $X$. 
We will explain what we mean by this momentarily. 

In our theory there is a field which encodes this partial deformation of complex structure on the Calabi--Yau manifold $X$.
It is an even field 
\[
\mu^{1,1} \in \Omega^{0,1} (X , \T_X) \otimes C^\infty(S) 
\]
where $\T_X$ denotes the holomorphic tangent bundle on $X$.
Locally, $\mu^{1,1}$ can be decomposed as Beltrami-like differential
\[
\mu^{1,1} = \mu^i_j (z,\zbar, t) \d \zbar_i \frac{\partial}{\partial z_j} .
\]
Physically speaking, $\mu^{1,1}$ is a component of the metric which survives in the twisted theory; see~\S\ref{s:components}. 

Next, there are two fields 
\[
\gamma^{1,0} \in \Omega^{1,0} (X) \otimes C^\infty(S), \quad \gamma^{1,2} \in \Omega^{1,2}(X) \otimes C^\infty(S) .
\]
The field $\gamma^{1,2}$ is a component of the supergravity $3$-form $C$-field that survives in the twisted theory. 
The field $\gamma^{1,0}$ can be interpreted as a component of the one-form which is a ghost-for-a-ghost of the $C$-field.\footnote{The $C$-field has a gauge symmetry of the form $\delta C = \d B$ where $B$ is a two-form.
This ghost $B$ field has an additional gauge symmetry $\delta B = \d A$ for $A$ a one-form.
The field $\gamma^{1,0}$ is a component of $A$.}

There is a background where the equation of motion involving the fields $\mu^{1,1}$, $\gamma^{1,0}$,  and~$\gamma^{1,2}$ reads
\begin{align*}
\dbar \mu^{1,1} + \frac12 [\mu^{1,1}, \mu^{1,1}] + \Omega^{-1} \vee \left(\del \gamma^{1,0} \wedge \del \gamma^{1,2} \right) & = 0 \\
\end{align*}
Additionally, there are the conditions that $\mu^{1,1}$ preserve the holomorphic volume form on $X$ and that all fields are locally constant along the topological direction $S$.  
Notice that this would be precisely the integrability equation for $\mu^{1,1}$ to determine a complex structure on $X$, were it not for the presence of the last term in the first equation.

If we work in a background where one of $\gamma^{1,0}$ or $\gamma^{1,2}$ is zero, then we see that $\mu^{1,1}$ is exactly a deformation of complex structure along $X$. 
In terms of the eleven-dimensional geometry, these field configurations describe deformations of the natural transverse holomorphic foliation (THF structure) on $X \times S$ which further preserve the holomorphic volume form along the leaves.
We further unpack the equations of motion in more general backgrounds in \S \ref{s:components}, but leave a complete study for future work.

\subsection*{Appearance of exceptional Lie superalgebras}

The gauge symmetries of a field configuration in any theory form a Lie algebra. 
In the Batalin--Vilkovisky formalism, one combines fields of all ghost number into a single object which, together with the linear BRST operator, has the structure of a dg Lie algebra.\footnote{In some models one actually obtains an $L_\infty$ algebra.} 
From the point of view of deformation theory, this dg Lie algebra describes the formal moduli space of deformations of the particular field configuration. 

The simplest field configuration in the twisted theory is the flat background; this corresponds to considering to our eleven-dimensional theory on $\CC^5 \times \RR$, where we equip $\CC^5$ with the flat holomorphic volume form. 
In this case, we find a striking relationship to a certain infinite-dimensional simple super Lie algebra studied by Kac and collaborators \cite{KacBible,KacE510}.\footnote{The twist of the eleven-dimensional gravity multiplet was computed in~\cite{SWspinor} by reducing it to a particular pure spinor model based on $\Gr(2,5)$; a possible relationship between this pure spinor model and $E(5,10)$ was first pointed out by Martin Cederwall in~\cite{martinSL5}.}

\begin{thm}
The global symmetry algebra of the eleven-dimensional theory on $\CC^5 \times \RR$ is equivalent to a central extension of the exceptional simple super Lie algebra $E(5,10)$. 
\end{thm}

In particular, correlation functions involving observables of the eleven-dimensional theory will be constrained by the infinite-dimensional symmetry algebra $E(5,10)$. 
Given our main conjecture that the interacting eleven-dimensional BV theory on $\CC^5 \times \RR$ is the twist of supergravity, we obtain the following.

\begin{conj} 
A central extension of the super Lie algebra $E(5,10)$ is a symmetry of supergravity on $\RR^{11}$ which preserves the background where the bosonic ghost takes value equal to a holomorphic supercharge $Q$.
\end{conj}

\subsection*{Relationship to other twists of supergravity} 

Motivated by the topological string, Costello and Li formulated conjectural descriptions of certain twists of 10-dimensional theories of supergravity. 
All of the descriptions center around the theory of Kodaira--Spencer gravity, otherwise known as BCOV theory. 
This theory was introduced in \cite{BCOV} as the closed-string field theory of the topological $B$-model on Calabi--Yau three-folds.
It was further extended to all Calabi--Yau manifolds in \cite{CLbcov1}. 
It is a sort of holomorphic version of gravity which at the genus zero  describes fluctuations of the complex structure of the Calabi--Yau manifold. 
We recall relevant aspects of Kodaira--Spencer gravity in \S\ref{s:BCOV}. 

Through the dimensional reduction of our eleven-dimensional theory we will find a match with Costello and Li's descriptions of twists of 10-dimensional supergravity in terms of BCOV theory.
In Table \ref{diagram}, we provide a summary of the comparisons with twists of supergravity in lower dimensions we perform. The squiggly arrows denote further twists, and the solid arrows denote various kinds of reductions. 

\begin{table}
\begin{center}
\[
\begin{tikzcd}[row sep=large,column sep = 3ex]
                                                                         &  \parbox{3cm}{\centering minimal twist of 11d} \arrow[ld, "\text{slab }" '] \arrow[d, "\text{topological }S^1" description] \arrow[rdd, "\text{CY3}" description, pos=.7, bend left = 5] \arrow[rr, rightsquigarrow] \arrow[rd, "\text{holomorphic }S^1"] &                                                             &  \parbox{3cm}{\centering non-minimal twist of 11d} \arrow[d, "\text{topological }S^1" description] \\
 \parbox{3cm}{\centering minimal twist of type I}  &  \parbox{3cm}{\centering minimal twist of type IIA}\arrow[l, "\text{orbifold}", start anchor={[xshift=1.5ex]west}, end anchor={[xshift=-1.5ex]east}]\arrow[r, rightsquigarrow, crossing over, end anchor = {[xshift=2ex]west}]                                                                                                                                                                                                &  \parbox{3cm}{\centering SU(4) twist of type IIA} \arrow[r, rightsquigarrow, start anchor={[xshift=-2ex]east}, end anchor={[xshift=2ex]west}] &  \parbox{3cm}{\centering SU(2) twist of type IIA}                                                 \\
                                                                         &                                                                                                                                                                                                                                                          & \parbox{3cm}{\centering minimal twist of 5d $\cN=1$}                   &                                                                                                
\end{tikzcd}
\]
\end{center}
\caption{Relations between various twists}
\label{diagram}
\end{table}

In \S\ref{s:su4red} we will show that the reduction along $\{0\} \times \RR \times \{0\} \subset \CC^4 \times \CC \times \RR$ is equivalent to the $SU(4)$ twist of type IIA supergravity. 
The topological string approach does not lead to a description of the minimal, or holomorphic, twist of type IIA supergravity. 
In \S \ref{s:su5IIA}, we describe the reduction along the line $\{0\} \times \RR \subset \CC^5 \times \RR$ to obtain a conjectural description of the holomorphic $SU(5)$ twist of type IIA supergravity. 
In \S \ref{s:Ired}, we show that the reduction along the interval $\{0\}\times [0,1]\subset \CC^5\times [0,1]$ with certain boundary conditions on the endpoints of the interval is equivalent to the twist of type I supergravity. We further observe in \S \ref{s:orbifold} that the twist of type I arises as the fixed points of a natural $\ZZ/2$ action on the minimal twist of IIA. Next, in \S \ref{s:CY3} we study the reduction along a Calabi-Yau 3 fold to obtain a conjectural description of the minimal twist of 5d $\cN=1$ supergravity. Finally, we show in \S \ref{s:nonmin} how the non-minimal twist of eleven-dimensional supergravity arises as a further twist of our holomorphic twist.


\subsection*{Acknowledgements}
We warmly thank I.~Brunner, M.~Cederwall, K.~Costello, R.~Eager, C.~Elliott, D.~Gaiotto, O.~Gwil\-liam, F.~Hahner, J.~Huerta, H.~Loosen, S.~Li, N.~Paquette, S.~R\"osch, J.~Walcher, P.~Yoo for conversation and inspiration of all kinds. I.S. and B.W. thank the University of Heidelberg, B.W. thanks the University of Edinburgh and the Mathematisches Forschungsinstitut Oberwolfach, and I.S. thanks the Krameterhof for hospitality while portions of this work were being completed. The work of S.R. is supported by the Perimeter Institute for Theoretical Physics. Research at Perimeter Institute is supported in part by the Government of Canada, through the Department of Innovation, Science and Economic Development Canada, and by the Province of Ontario, through the Ministry of Colleges and Universities.
The work of I.S. is supported by the Free State of Bavaria. The work of B.W. is supported by the University of Edinburgh.


\section{The minimal twist of eleven-dimensional supergravity} 
\label{s:dfn}

In this section we define the central theory of study within the Batalin--Vilkovisky (BV) formalism. 
The theory will be defined on any eleven-dimensional manifold of the form $X \times L$, where $X$ is a Calabi--Yau five-fold and $L$ is a smooth oriented one-manifold.

In \cite{SWspinor}, we showed that the underlying free limit of the theory we consider here is the free limit of the minimal twist of 11-dimensional supergravity on $\CC^5 \times \RR$. 
The goal of this section is to introduce interactions at the level of the minimal twist. 
The main result is Theorem \ref{thm:dfn} where we show that the theory is consistent within the BV formalism. 

In the remainder of the paper we discuss further consistency with supersymmetry and string theory, but in this section we focus mostly on the theory as a partially holomorphic, partially topological, theory of gravity. 
Nevertheless, we do provide some preliminary justification towards the relationship to physical supergravity including a matching of indices in \S \ref{sec:locchar}. 

\subsection{Divergence-free vector fields} 

\subsubsection{}
\label{sec:divfree}
We set up some notations and conventions in the context of complex geometry. 
Let $V$ be a holomorphic vector bundle on a $d$ dimensional complex manifold $X$. 
If $j$ is an integer, we let $\Omega^{0,j}(X, V)$ denote the space of anti-holomorphic Dolbeault forms of type $j$ on~$X$ with values in $V$.
The $\dbar$ operator $\dbar \colon \Omega^{0,j}(X, V)\to \Omega^{0,j+1}(X,V)$ defines the Dolbeault complex of $V$:
\[
  \Omega^{0,\bu}(X, V) = \left(\Omega^{0,j}(X, V)[-j] , \; \dbar\right) .
\]
This is a free (smooth) resolution for the sheaf of holomorphic sections of $V$.

Suppose $X$ is a Calabi--Yau manifold with holomorphic volume form $\Omega$.
The divergence $\div(\mu)$ of a holomorphic vector field $\mu$ is defined by the formula
\[
\div (\mu) \wedge \Omega = L_\mu (\Omega)
\]
where, on the right hand side, we mean the Lie derivative of $\Omega$ with respect to $\mu$.

Let $\T_X$ denote the holomorphic tangent bundle and consider its Dolbeault complex $\Omega^{0,\bu}(X , \T_X)$ resolving the sheaf of holomorphic vector fields. 
The divergence operator extends to the Dolbeault complex to yield a map of cochain complexes 
\[
\div \colon \Omega^{0,\bu}(X , \T_X) \to \Omega^{0,\bu}(X) .
\]
The resulting complex of sheaves
\beqn\label{eqn:cplx1}
\begin{tikzcd}
\ul{0} & \ul{1} \\
\Omega^{0,\bu}(X , \T_X) \ar[r, "\div"] & \Omega^{0,\bu}(X) ,
\end{tikzcd}
\eeqn
resolves the sheaf of holomorphic divergence-free vector fields $\Vect_0 (X)$.
The anti-holomorphic Dolbeault degrees and the $\dbar$ operator are left implicit. 

There is a direct way to extend the Lie bracket of vector fields to the complex \eqref{eqn:cplx1}. 
Denote by $\mu$ an element of $\Omega^{0,\bu}(X , \T_X)$ and $\nu$ an element of $\Omega^{0,\bu}(X)$ (for simplicity in notation, we will not expand the anti-holomorphic dependence). 
The Lie bracket defined by the formulas
\begin{align*}
[\mu, \mu'] & = L_\mu \mu' \\
[\mu, \nu] & = L_\mu \nu 
\end{align*}
is compatible with $\div$ and endows \eqref{eqn:cplx1} with the structure of a sheaf of dg Lie algebras.
We will refer to this sheaf by the symbol $\cL_0(X)$, or just $\cL_0$ if $X$ is understood. 

The sheaf $\cL_0$ has the structure of a {\em local} dg Lie algebra~\cite[\S 3.1.3]{CG2}.
This means that, as a graded sheaf, $\cL_0$ is the smooth sections of a graded vector bundle, and its differential and Lie bracket are given by differential and bidifferential operators respectively.

\parsec[sec:Linfty]

Recall that an $L_\infty$ algebra is a $\ZZ$-graded vector space $\cL$ together with the datum of a square-zero, degree $+1$ derivation $\delta_\cL$ of the free commutative graded algebra $\Sym\left(\cL^\vee [-1] \right)$. 
The Chevalley--Eilenberg cochain complex is 
\[
\left(\Sym\left(\cL^\vee [-1] \right), \delta_\cL\right) .
\]
The Taylor components of $\delta_\cL$ define higher brackets $\{[-]_k\}_{k=1,2,\ldots}$ where $[-]_k \colon \cL^{\otimes k} \to \cL[2-k]$. 
The condition that the differential $\delta_\cL$ is square-zero is equivalent to the higher Jacobi relations.

An $L_\infty$ morphism $\Phi: \cL \rightsquigarrow \cL'$ is the same datum as a map of commutative dg algebras 
\deq{
  \Phi^*: \clie^\bu(\cL') \to \clie^\bu(\cL)
}
between their respective Lie algebra cochains. It follows from this that \emph{any} automorphism $\Phi$ of the free commutative algebra on $\cL^\vee[-1]$ defines a new model of the $L_\infty$ algebra $\cL$, for which the Chevalley--Eilenberg differential is obtained by conjugating $\delta_\cL$ by~$\Phi$, and where $\Phi$ itself defines the $L_\infty$ isomorphism.

$\cL_0$ is the sheaf~\eqref{eqn:cplx1} resolving divergence-free vector fields equipped with the dg Lie algebra structure constructed in the previous section.
We consider the following automorphism of~$\Sym(\cL_0^\vee[-1])$, defined by its action on generators:
\deq[eq:newbase]{
    \Psi_\infty: \nu \mapsto 1 - e^{-\nu}, \quad \mu \mapsto e^{-\nu} \mu.
}
This map defines a new model of the $L_\infty$ algebra with underlying graded vector space the same as \eqref{eqn:cplx1}, which we will call $\cL_\infty$.\footnote{We are being slightly abusive and using the symbols $\nu,\mu$ dually as coordinates, or operators, on the graded linear space $\cL[1]$.}
The formulas for the automorphism above clearly arise from maps of vector bundles and hence endow $\cL_\infty$ with the structure of a local $L_\infty$ algebra, meaning all operations are given by polydifferential operators.  

The notation refers to the fact that this new model has nonvanishing $L_\infty$ brackets of every order. 
It is this new model that we will use to define the eleven-dimensional theory of twisted supergravity.


We can describe the $L_\infty$ structure on our new model $\cL_\infty$ explicitly.
Recall that we have two types of elements: $\mu \in \PV^{1,\bu}$ and $\nu \in \PV^{0,\bu}[-1]$. (Here, and in what follows, we will use the symbol $\PV^{i,\bu}$ for the Dolbeault resolution of \emph{holomorphic polyvector fields;} by definition, this is the complex $\Omega^{0,\bu}(X,\wedge^i \T_X)$.)
The first few nonzero brackets are
\begin{align*}
[\mu]_1 & = \dbar \mu + \div \mu \\
[\mu_1,\mu_2]_2 & = \div (\mu_1 \wedge \mu_2) \\
[\nu, \mu_1,\mu_2]_3 & = \div(\nu \mu_1 \wedge \mu_2) 
\end{align*}
For $k \geq 2$ the general formula for the $k$-ary brackets are 
\begin{align*}
[\nu_1, \ldots, \nu_{k-2}, \mu_1,\mu_2]_{k} & = \div(\nu_1 \cdots \nu_k \mu_1 \wedge \mu_2) \\
[\nu_1,\ldots, \nu_{k-3}, \mu_1,\mu_2,\gamma]_k & = \nu_1 \cdots \nu_{k-3} (\mu \wedge \mu') \vee \del \gamma .\\
[\nu_1,\ldots,\nu_{k-2}, \mu, \gamma]_k & = \nu_1 \cdots \nu_{k-2} \mu \vee \del \gamma .
\end{align*}



\subsection{Theories of BF type}

\parsec
Suppose that $\cL$ is an $L_\infty$ algebra with $L_\infty$ operations $\{[-]^\cL_k\}_{k=1,2,\ldots}$ and that $(\cA, \d_\cA)$ is a commutative dg algebra. 
The graded vector space $\cL \otimes \cA$ is equipped with the natural structure of an $L_\infty$ algebra with operations $\{[-]_k\}_{k=1,2,\ldots}$ defined by
\begin{align*}
[x \otimes a]_1 & = [x]^\cL_1 \otimes a + (-1)^{|x|} x \otimes \d_\cA a \\
[x_1 \otimes a_1, \ldots , x_k \otimes a_k]_k & = [x_1,\ldots,x_k]^\cL_k \otimes (a_1 \cdots a_k), \qquad k \geq 2 .
\end{align*}

We apply this construction, taking $\cL$ to be the sheaf resolving divergence-free holomorphic vector fields on a Calabi--Yau manifold $X$ equipped with either the strict dg Lie algebra structure $\cL_0(X)$ or its non-strict $L_\infty$ structure $\cL_\infty (X)$. 
The algebra $\cA$ will be the smooth de Rham complex $(\Omega^\bu(S) , \d_S)$ where $S$ is a smooth manifold. 

We thus obtain the structure of an dg Lie algebra on $\cL_0(X) \otimes \Omega^{\bu}(S)$ or an $L_\infty$ algebra $\cL_\infty(X) \otimes \Omega^\bu(S)$.
These define equivalent local $L_\infty$ algebras on the product manifold~$X \times S$. 

\parsec[s:bf]

Associated to any local $L_\infty$ algebra is a classical field theory in the BV formalism.
Let $\cL$ be a local $L_\infty$ algebra on some manifold $M$, it is the sheaf of sections of some graded vector bundle $L$. 
For a section $A \in \cL$, introduce the `higher curvature map' defined by the formula
\[
\mathsf{F}_A = [A]_1 + \frac12 [A,A]_2 + \frac{1}{3!} [A,A,A]_3 + \cdots .
\]

The fields of the associated BV theory are pairs
\[
  (A, B) \in \cL[1] \oplus \cL^{!}[-2] .
\]
Here $\cL^!$ denotes the sheaf of sections of the bundle $L^* \otimes {\rm Dens}$, where ${\rm Dens}$ is the bundle of densities. 
The shifted symplectic BV pairing is the obvious integration pairing between $A$ and $B$. 

The action functional reads $S_{\rm BF} = \int_M B \, \mathsf{F}_{A}$ which leads to the equations of motion $\mathsf{F}_{A} = 0$ and $\mathsf{D}_A B= 0$ where $\mathsf{D}_A$ is the higher covariant derivative along $A$. 
We refer to this as the ``BF theory'' associated to $\cL$.

We thus obtain a theory in the BV formalism on the product manifold $X \times S$ associated to both local $L_\infty$ algebras $\cL_0(X) \otimes \Omega^{\bu}(S)$ and $\cL_\infty(X) \otimes \Omega^\bu(S)$.

\parsec
For concreteness, we spell out the fields of the theories we have constructed on $X \times S$.
In both cases, the space of fields equipped with the linear BRST operator is
\begin{equation}
  \label{eq:sympfields} 
  \begin{tikzcd}[row sep = 1 ex]
    -n & -n + 1 & -1 & 0 \\ \hline
    \Omega^{0}(X;S) \ar[r, "\del"] & \Omega^{1}(X;S) & 
     \PV^{1}(X; S) \ar[r, "\div"] & \PV^{0}(X; S).
\end{tikzcd}
\end{equation}
We denote the fields $(\beta,\gamma,\mu,\nu)$ respectively.
We are using the shorthand notation
\begin{align*}
\Omega^{i}(X;S) & = \Omega^{i , \bu;\bu}(X;S) \\
 & = \oplus_{j,k} \PV^{i,j}(X) \otimes \Omega^k(S) [-j-k] .
\end{align*}
which is equipped with the $\dbar + \d_S$ operator and similarly for $\PV^{i}(X;S)$. 

The natural pairing between $\PV^i(X;S)$ and~$\Omega^i(X;S)$ is of degree $-\dim_\C(X) -\dim_\R(S)$. 
As such, the $\Z$-grading indicated in~\eqref{eq:sympfields} equips the sheaf of fields with a degree $(-1)$ pairing, provided that we choose the shift to be given by
\deq{
  n = \dim_\C(X) + \dim_\R(S) - 1.
}
The pairing is defined by the formula 
\[
\int^\Omega_{X \times S} \mu \vee \gamma + \int^\Omega_{X \times S} \nu \beta 
\]
where $\int^\Omega_{X \times S} \alpha = \int_{X \times S} \alpha \wedge \Omega$. 

We have constructed two equivalent descriptions of the BF theory which share the linear BRST complex \eqref{eq:sympfields}.
Explicitly, the action functional for BF theory associated to the local dg Lie algebra $\cL_0(X) \otimes \Omega^{\bu}(S)$ is
\deq{
  S_{BF,0} =  \int^\Omega \bigg[\beta \wedge (\dbar + \d_S) \nu +  \gamma \wedge (\dbar + \d_S) \mu +  \beta \wedge \partial_\Omega \mu +  \frac{1}{2} [\mu,\mu] \vee \gamma +  [\mu,\nu] \beta \bigg] .
}
As in the Lie algebra structure of this strict model, notice that the Schouten--Nijenhuis bracket appears explicitly. 

The action functional of BF theory associated to $\cL_\infty(X) \otimes \Omega^{\bu}(S)$ is non polynomial. 
In fact, it is related to the BCOV action functional via dimensional reduction (see \S \ref{sec:dimred}).
Explicitly, this action functional is
\deq{
  S_{BF,\infty} =  \int^\Omega \bigg[ \beta \wedge (\dbar + \d_S) \nu +  \gamma \wedge (\dbar + \d_S) \mu +   \beta \wedge \partial_\Omega \mu + \frac12 \frac{1}{1-\nu} \mu^2 \vee \del \gamma \bigg] .
}

We demonstrated above that the two local $L_\infty$ algebras on which these BF theories are based are equivalent. As such, the BF theories are also equivalent; the map~\eqref{eq:newbase} extends uniquely to an automorphism of BV theories.
Explicitly, the automorphism is
\begin{equation}\label{eqn:auto1}
  \mu \mapsto e^{-\nu} \mu, \qquad \nu \mapsto 1-e^{-\nu} , \qquad
  \beta \mapsto (\beta - \mu \vee \gamma) e^{\nu},\qquad \gamma \mapsto e^{\nu} \gamma .
\end{equation}

\parsec[]

In what follows, we specialize to the case that $X$ is a Calabi--Yau five-fold and that $S$ is a one-dimensional smooth orientable manifold. 
In this case, with $n = 5 + 1 - 1 = 5$ the theories described in this section are $\ZZ$-graded in the BV formalism.
Momentarily, we consider a new term in the action which will break this grading; as such, this integer shift will not play an essential role.

\subsection{A deformation of BF theory} 

Let $X$ be a Calabi--Yau five-fold and $S$ be a smooth oriented one-dimensional real manifold. 
We will break the $\ZZ$-grading present in BF theory discussed in the previous section to a $\ZZ/2$ grading. 
For reference, this means that the linear cochain complex of fields of the model now takes the following form. 

\begin{equation}
  \label{eq:sympfields} 
  \begin{tikzcd}[row sep = 1 ex]
    {\rm odd} & {\rm even} & {\rm odd} & {\rm even} \\ \hline
    \Omega^{0}(X;S)_\beta \ar[r, "\del"] & \Omega^{1}(X;S)_\gamma & 
     \PV^{1}(X; S)_\mu \ar[r, "\div"] & \PV^{0}(X; S)_\nu.
\end{tikzcd}
\end{equation}

\parsec

To define our classical field theory on $X \times S$, we consider  a deformation of BF theory $S_{BF}$ (this refers to either the presentation as $S_{BF,0}$ or $S_{BF,\infty}$). 
Such deformations are governed by the classical master equation: the parameterized family of actions 
\beqn\label{eqn:defaction}
S_{BF} + g J
\eeqn
defines a consistent theory in the BV formalism if and only if
\deq{
  \{S_{BF} + g J, S_{BF} + g J \} = 0.
}
Since this must hold for all $g$, and since the undeformed action $S$ is already a solution to the classical master equation, this reduces to the pair of conditions
\deq[eq:2cond]{
  \{S_{BF},J\} =  \{J,J\} = 0.
}

The form of $J$ depends on which presentation we use for BF theory.
To begin, we will use the presentation of BF theory $S_{BF, \infty}$ which uses the the non-strict $L_\infty$ structure on divergence-free holomorphic vector fields.
The deformation $J$ does not make reference to the Calabi--Yau structure explicitly, but it does involve the holomorphic de Rham operator $\del$ on $X$. 

The main result of this section is the following. 

\begin{thm}
\label{thm:dfn}
Let $X$ be a Calabi--Yau five-fold and $S$ a smooth one-dimensional manifold, and consider the BV theory $(\cE, S_{BF,\infty})$ on $X \times S$ defined above. The local functional 
  \deq{
    J = \frac16 \gamma \wedge \del \gamma \wedge \del \gamma ,
  }
  where $\gamma \in \Omega^{1,\bu}(X;S)$, defines a deformation of~$(\cE,S_{BF,\infty})$ as a $\Z/2$-graded BV theory.
\end{thm}

\parsec[]

Before proceeding to the proof, we remark on grading issues. 
In the original $\Z$-grading on the BF theory given in \eqref{eq:sympfields} with $n=5$, the component 
\[
\gamma^{1,i;j} \in \Omega^{1,i}(X) \otimes \Omega^j(S) 
\]
sits in degree $-4+i+j$. 
Thus, we see that in the original $\Z$-grading on BF theory one has
  \deq{
    \deg(J) = 6.
  }
Thus $S_{BF} + J$ is not of homogenous $\ZZ$ grading (although it is even).

This is completely reasonable from the point of view of twisting supersymmetry in eleven dimensions. 
Indeed, the $R$-symmetry group is trivial, and there is not a way to regrade the fields of the twisted theory using twisting data~\label{CosHol,ESW}. 
Nevertheless, if we break to the obvious $\ZZ/2$ grading, the functional $S_{BF} + g J$ defines an even action functional.
Unless otherwise stated, we will work with this $\ZZ/2$ grading for the remainder of this section.

\parsec[]
We proceed to show that $S_{BF,\infty} + g J$ solves the classical master equation.
For notational simplicity we will omit the integral symbol $\int^\Omega$.

\begin{proof}

It is immediate from the form of the BV bracket that $\{J,J\} = 0$, since $J$ depends only on the $\gamma$ field. 
It remains to check that $\{S_{BF,\infty},J\} = 0$. 
For the quadratic term in the BF action, we note that 
  \deq{
    \{\beta \wedge \div\mu, J\} = \frac12 \del\beta \wedge \del \gamma \wedge \del \gamma = 0,
  }
  because total derivatives are equivalent to zero as local functionals. 
  
The contribution from the remaining BF action takes the form
\[
    \left\{ \frac12 \frac{1}{1-\nu} \del\gamma \vee \mu^2, \frac16 \gamma \wedge \del\gamma \wedge \del \gamma \right\} =\frac12 (\mu \vee \del \gamma) \wedge \del \gamma \wedge \del \gamma .
\]
This expression is zero for symmetry reasons. 
Recall that $\del\gamma$ is a two-form, and that the expression must be a totally symmetric local functional which is cubic in this two-form. We can ask whether such a  contraction exists just at the level of $\lie{sl}(5)$ representation theory. Let $\ydiagram{1}$ denote the fundamental representation of~$\lie{sl}(5)$, which we identify with constant one-forms. Since the term must be a scalar, the contraction $(\partial\gamma^3$ must sit in the fundamental representation again, since it is dual to a vector field. Computing the decomposition of the tensor cube of the two-form, we find
\deq{
  \Sym^3 \left( \ydiagram{1,1} \right) \cong \ydiagram{3,3} \oplus \ydiagram{2,2,1,1}\,.
}
(In fact, the absence  of the relevant irreducible representation does not  even depend on the parity of the  field $\gamma$, since 
  \deq{
    \wedge^3\left( \ydiagram{1,1} \right) \cong \ydiagram{3,1,1,1}  \oplus \ydiagram{2,2,2}\, ;
  }
the  fundamental representation has symmetry type $\ydiagram{2,1}$.) 
\end{proof}

\parsec[s:coupling]

We make note of the dependence on the coupling constant $g$ in the definition of the deformed action $S_{BF,\infty} + g J$. 

When $g = 0$ we recover BF theory for the $L_\infty$ algebra $\cL_\infty(\CC^5) \otimes \Omega^\bu(\RR)$. 
For any $g \ne 0$ the theories are essentially equivalent in perturbation theory. 
Indeed, if $g \ne 0$ we can make the following field redefinition 
\[
\gamma \mapsto \sqrt{g} \gamma, \quad \beta \mapsto \sqrt{g} \beta 
\]
to write the action as 
\[
\frac{1}{\sqrt{g}} \left(S_{BF,\infty} + J \right)  .
\]

In perturbation theory, this has the affect of modifying the quantization parameter $\hbar$ to $\hbar / \sqrt{g}$.
Thus, after modifying $\hbar$ and making the above field redefinition, the perturbative expansion of any theory is equivalent to the one with $g = 1$. 

\parsec[s:altdfn]

We remark on an alternative, equivalent, description of the deformed theory which involves the strict dg Lie algebra structure on divergence-free holomorphic vector fields.

We can replace $S_{BF,\infty}$ by $S_{BF,0}$ via applying the field automorphism \eqref{eqn:auto1}.
Doing this we see that $J$ becomes 
\[
\til{J} = \frac16 e^\nu \gamma \wedge \del (e^\nu \gamma) \wedge \del(e^\nu \gamma) .
\]
Since this automorphism preserves the odd BV bracket, the actions $S_{BF,\infty} + g J$ and $S_{BF, 0} + g \til{J}$ are both solutions to the classical master equation, and are equivalent as~$\ZZ/2$ graded BV theories.

\subsection{Equations of motion of the component fields} \label{s:components}

Soon, we will provide a series of justifications for the assertion that the deformed theory $S_{BF, \infty} + g J$ is the minimal twist of eleven-dimensional supergravity on flat space $X \times S = \CC^5 \times \RR$ where $\CC^5$ is equipped with its flat Calabi--Yau form. 
For the moment, we briefly read off the equations of motion of the general theory on $X \times S$.
Let $\Omega$ denote the Calabi--Yau form on $X$. 

We consider the action $S_{BF, \infty} + gJ$.
The equation of motion obtained by varying $\beta$ is especially simple---in fact linear---since $\beta$ only appears in the action via a quadratic term. 
It is
\beqn\label{eqn:eombeta}
\dbar \nu + \d_S \nu + \div \mu = 0 .
\eeqn
Varying $\gamma$ we obtain the equation of motion
\beqn\label{eqn:eomgamma}
\dbar \mu + \d_S \mu + \frac12 \frac{1}{1-\nu} \div (\mu^2) + \frac12 (\del \gamma \wedge \del \gamma) \vee (g \Omega^{-1}) = 0 .
\eeqn
The last term represents the contraction of an element of $\Omega^{4,\bu}(X;S)$ with the nonvanishing section $\Omega^{-1} \in \PV^{5,\bu}(X;S)$ to yield an element of $\PV^{1,\bu}(X;S)$. 
If we vary the $\mu$ we obtain 
\beqn\label{eqn:eommu}
(\dbar + \d_S) \gamma + \del \beta + \frac{1}{1-\nu} (\mu \vee \del \gamma) = 0 .
\eeqn
Finally, if we vary $\nu$ we obtain
\beqn\label{eqn:eomnu}
(\dbar + \d_S) \beta + \frac12 \frac{1}{(1-\nu)^2} \mu^2 \vee \del \gamma = 0 .
\eeqn

The equation of motion must hold for any inhomogenous superfields.
We can get a better sense of the equations if we expand in components of these fields. 
The component fields of the eleven-dimensional theory on $X \times S$ have the following form: 
\begin{itemize}
\item $\mu = \sum_{i,j} \mu^{i;j}$ is a superfield where
\[
\mu^{i;j} \in \PV^{1,i}(X) \otimes \Omega^j(\RR) ,\quad i=0,\ldots, 5, \quad j=0,1.
\]
The component $\mu^{i;j}$ has parity $i+j+1 \pmod 2$. 
\item $\nu = \sum_{i,j} \nu^{i;j}$ is a superfield where
\[
\nu^{i;j} \in \PV^{0,i}(X, \T_X) \otimes \Omega^j(\RR) ,\quad i=0,\ldots, 5, \quad j=0,1.
\]
The component $\nu^{i;j}$ has parity $i+j \pmod 2$. 
\item 
$\gamma = \sum_{i,j} \gamma^{i;j}$ is a superfield where
\[
\gamma^{i;j} \in \Omega^{1,i}(X) \otimes \Omega^j(\RR) ,\quad i=0,\ldots, 5, \quad j=0,1.
\]
The component $\gamma^{i;j}$ has parity $i+j\pmod 2$. 
\item 
\item $\beta = \sum_{i,j} \beta^{i;j}$ is a superfield where
\[
\beta^{i;j} \in \Omega^{0,i}(X) \otimes \Omega^j(\RR) ,\quad i=0,\ldots, 5, \quad j=0,1.
\]
The component $\beta^{i;j}$ has parity $i+j+1 \pmod 2$. 
\end{itemize}

We look closely at the geometric meaning of \eqref{eqn:eombeta}. 
Let's make the simplifying assumption that all components of $\mu$ are divergence-free, and further that all fields are locally constant along $S$: that is, $\div \mu = 0$ and $\d_S \mu = \d_S \gamma = 0$.
Then $\nu = 0$ is a solution to \eqref{eqn:eombeta} and we can assume that all fields are functions, or zero-forms, along $S$. 
Then, there is a component of \eqref{eqn:eomgamma} which can be written as 
\beqn\label{eqn:eomgamma1}
\dbar \mu^{1;0} + \frac12 [\mu^{1;0},\mu^{1;0}] + \left(\frac12 \del \gamma^{1;0} \wedge \del \gamma^{1;0} + \del \gamma^{2;0} \wedge \del \gamma^{0;0}\right) \vee (g \Omega^{-1}) = 0 
\eeqn
where now $[-,-]$ stands for the Schouten bracket.

To further simplify \eqref{eqn:eomgamma1}, we can look for a solutions where  $\gamma^{1;0}$, the $(0,1)$ Dolbeault part of $\gamma$, is zero. 
Then, up to the term involving 
\[
\alpha \define \del \gamma^{0;0},
\]
we find precisely the integrability equation for the complex structure determined by Beltrami differential $\mu^{1;0} \in \PV^{1,1} \otimes \Omega^0$. 
If $\dbar \alpha = 0$, the holomorphic two-form $\alpha \in \Omega^{2,hol}(X)$ defines a map of sheaves
\[
\Omega^{2,hol}_X \xto{\wedge \alpha} \Omega^{4,hol}_X \cong_\Omega \cT^{hol}_X
\]
where $\cT^{hol}_X$ denotes the sheaf of holomorphic vector fields and the last isomorphism uses the Calabi--Yau form $\Omega$ on $X$.  
The image of $\Omega^{2,hol}_X$ defines a subsheaf $\cF_{\alpha} \subset \cT^{hol}_X$. 
Since $\del \alpha = 0$, this subsheaf is automatically integrable and hence determines a foliation. 

Summarizing, see that there is a field configuration where the Beltrami diffrential $\xi = \mu^{1;0} \in \Omega^{0,1}(X, \T_X)$ satisfies the modified integrability condition
\[
\dbar \xi + \frac12  [\xi , \xi] = \alpha \vee \rho
\]
for some $\rho \in \PV^{2,2} (X)$. 
In other words, $\xi$ defines an integrable complex structure deformation along the leaf space associated to the foliation $\cF_\alpha$. 
We leave a more complete exploration of the moduli space of solutions of the equations of motion for future work. 

In~\cite{SWspinor}, the second two authors showed that the free limit of the minimal twist of eleven-dimensional supergravity agrees with the free limit of the eleven-dimensional theory that we have introduced here. 
Given this result, we can recognize many fields in the twisted theory as components of the physical fields of supergravity which remain after we twist. 

\begin{itemize}
\item 
The components 
\begin{align*}
\mu^{1;0} & = \mu^j_i(z,\zbar,t) \d \zbar_j \partial_{z_i} \\
\mu^{0;1} & = \mu^t_i (z,\zbar,t) \d t \partial_{z_i}
\end{align*}
of $\mu$ comprise components of the metric which remain after the twist. 
The components 
\[
\mu^{0;0} = \mu_i (z,\zbar,t) \partial_{z_i} 
\]
comprise the ghosts for infinitesimal (holomorphic) diffeomorphisms. 
\item 
The three-form fields
\begin{multline}
\beta^{3;0} = \beta^{ijk} (z,\zbar,t) \d \zbar_i \d \zbar_j \d \zbar_k , \quad \beta^{2;1} = \beta^{ij}_t (z,\zbar,t) \d \zbar_i \d \zbar_j \d t \\
\gamma^{2;0} = \gamma^{ijk} (z,\zbar,t) \d z_i \d \zbar_j \d \zbar_k , \quad \gamma^{1;1} = \gamma^{ij}_t (z,\zbar,t) \d z_i \d \zbar_j \d t .
\end{multline} 
comprise components of the supergravity $C$-field which remain after the twist. 
The two-form fields $\beta^{2;0}, \beta^{1;1}, \gamma^{1;0}, \gamma^{0;1}$, the one-form fields $\beta^{1;0}, \beta^{0;1}$, and the zero-form field $\beta^{0;0}$ is what remains of the ghost system (ghosts, ghosts for ghosts, etc.) for the supergravity $C$-field. 
\end{itemize}

\subsection{Local character}\label{sec:locchar}

We consider the eleven-dimensional theory on the manifold $\CC^5 \times \RR$, where $\CC^5$ is equipped with its standard Calabi--Yau structure. 
On this background, the theory is manifestly $SU(5)$ invariant. 
In this section, we compute the corresponding character of the local operators at the origin. 

The local character is only sensitive to the free limit of the theory.
Furthermore, the linear BRST operator is an $SU(5)$-invariant deformation of the $(\dbar + \d_{\RR})$ operator. 
Therefore, to compute the character it suffices to compute the $SU(5)$-equivariant character of the $\dbar$ cohomology. 

The solutions to the $(\dbar + \d_{\RR})$-equations of motion simply say that all fields are holomorphic along $\CC^5$ and constant along $\RR$. 
Thus, the solutions can be identified with 
\begin{align*}
\mu^{i}\partial_{z_i} & \in \Vect(\CC^5) \cong \cO(\C^5)\partial_{z_i},\quad 
\nu \in \cO (\C^5) \\
\beta & \in \cO (\C^5), \quad \gamma^{i} \d z_i \in \Omega^{1}(\CC^5) \cong \cO (\C^5)\d z_i 
\end{align*}
where $z_i$ is a holomorphic coordinate on $\CC^5$. 

Corresponding to each of the above, we have a tower of linear local operators labeled by $(m_j) = (m_1, m_2, m_3, m_4, m_5)\in \Z^5_{\geq 0}$; these are given by
\begin{align*}
 \boldsymbol{\mu}^{i}_{(m_j)} &: \mu^{i}\mapsto \partial_{z_1}^{m_1}\partial_{z_2}^{m_2}\partial_{z_3}^{m_3}\partial_{z_4}^{m_4}\partial_{z_5}^{m_5}\mu^{i} (0) \\
\boldsymbol{\nu}_{(m_j)} &: \nu\mapsto \partial_{z_1}^{m_1}\partial_{z_2}^{m_2}\partial_{z_3}^{m_3}\partial_{z_4}^{m_4}\partial_{z_5}^{m_5}\nu (0) \\
\boldsymbol{\gamma}^{i}_{(m_j)} &: \gamma^{i}\mapsto \partial_{z_1}^{m_1}\partial_{z_2}^{m_2}\partial_{z_3}^{m_3}\partial_{z_4}^{m_4}\partial_{z_5}^{m_5}\gamma^{i} (0) \\
 \boldsymbol{\beta}_{(m_j)} &: \beta\mapsto \partial_{z_1}^{m_1}\partial_{z_2}^{m_2}\partial_{z_3}^{m_3}\partial_{z_4}^{m_4}\partial_{z_5}^{m_5}\beta (0)
\end{align*}

It is easiest to label the Cartan subgroup of $SU(5)$ by fugacities $q_1,\ldots, q_5$ subject to the constraint that $\prod_{i=1}^5 q_i = 1$. 
We first compute the single particle index.
This is the $SU(5)$ character of the space of linear local operators.

\begin{lem}
The single particle index is 
\[
i(q_1,\ldots,q_5) = \frac{\sum_{i=1}^5 q_i}{\prod_{i=1}^5 (1-q_i)} + \frac{\sum_{i=1}^5 q_i^{-1}}{\prod_{i=1}^5 (1-q_i^{-1})}
\]
where the fugacities satisfy the constraint $\prod_{i=1} q_i = 1$. 
\end{lem}
\begin{proof}
The linear local operators $ \boldsymbol{\nu}_{(m_j)}$ and $\boldsymbol{\beta}_{(m_j)}$ are of the same $q$-weight but opposite parity.
Thus, they do not contribute to the single particle index.

The $q$-weight of the odd local operator $\boldsymbol{\mu}_{(m_j)}^i$ is 
\[
q_1^{m_1+1} \cdots q_i^{m_i} \cdots q_5^{m_5+1} .
\]
The $q$-weight of the even local operator $\boldsymbol{\gamma}_{(m_j)}^i$ is 
\[
q_1^{m_1} \cdots q_i^{m_i + 1} \cdots q_5^{m_5} .
\]

Thus we find that the single particle index is given by the infinite series
\beqn\label{infseriesindex}
\sum_{i=1}^5\left ( \sum_{(m_i)\in \Z^5_{\geq 0}} q_1^{m_1} \cdots q_i^{m_i + 1} \cdots q_5^{m_5} - \sum_{(m_i)\in \Z^5_{\geq 0}} q_1^{m_1+1} \cdots q_i^{m_i} \cdots q_5^{m_5+1} \right)
\eeqn

which sums to the expression
\beqn\label{singleparticleindex}
- \frac{\sum_{i=1}^5 q_1 \cdots \Hat{q_i} \cdots q_5}{\prod_{i=1}^5 (1-q_i)} + \frac{\sum_{i=1}^5 q_i}{\prod_{i=1}^5 (1-q_i)} .
\eeqn

This simplifies to the stated expression.
\end{proof}

This single particle index for our space of local operators agrees with the one computed in \cite{NekrasovInstanton}. 
To obtain the full index of local operators we apply the plethystic exponential ${\rm PE}[f(x)] = \exp\left(\sum_n \frac1n f(x^n)\right)$. 

\begin{prop}\label{prop:locchar}
The character of local operators of the eleven-dimensional theory on $\CC^5 \times \RR$ is 
\[
\prod_{i=1}^{5} \prod_{(m_i)\in \Z^5_{\geq 0}} \frac{1-q_1^{m_1+1}\cdots q_i^{m_i}\cdots q_5^{m_5+1}}{1-q_1^{m_1}\cdots q_i^{m_i+1}\cdots q_5^{m_5}}
\]
\end{prop}
\begin{proof}
Recall that the plethystic exponential takes sums to products and monomials to geometric series. Apply this to the infinite series \eqref{infseriesindex}.
\end{proof}


\subsection{One-loop quantization}

In \cite{GRWthf} an existence result for one-loop quantizations of mixed topological-holomorphic theories was established. 
We apply this to the eleven-dimensional model at hand. 

The eleven-dimensional theory is a mixed topological-holomorphic theory.
On flat space $\CC^5_z \times \RR_t$, this means that the theory is translation invariant and that the following act homotopically trivially:
\begin{itemize}
\item the vector fields $\del_{\zbar_1}, \ldots, \del_{\zbar_{5}}$ corresponding to infinitesimal anti-holomorphic translations,
\item the vector field $\partial_t$ corresponding to infinitesimal translations in the $\RR_t$ direction. 
\end{itemize}

Recall that the action functional of the eleven-dimensional theory is $S_{BF, \infty} + c J$. 
Since the cubic and higher interactions only involve holomorphic derivatives, we obtain the following directly from the main result of \cite{GRWthf}. 

\begin{thm}
There exists a gauge fixing condition for the eleven-dimensional theory on $\CC^5 \times \RR$ which renders its one-loop quantization finite and anomaly-free. 
\end{thm} 

When $g=0$, this result is actually exact,
since there are no Feynman diagrams present past one-loop order in this case. 
When $g \ne 0$, on the other hand, this result does not immediately imply the existence of a gauge-invariant perturbative quantization to higher orders in $\hbar$. 
The presence of the functional $J = \frac16 \int \gamma \del \gamma \del \gamma$ allows one to construct Feynman graphs at arbitrary loop order.

In \cite{CostelloM5}, Costello argues that, upon performing the $\Omega$-background, the theory localizes to a five-dimensional theory on $\CC^2 \times \RR$. 
Via a cohomological argument, it is shown that this effective five-dimensional theory exhibits an essentially unique quantization in perturbation theory. 
We will return to the existence and uniqueness of a higher order quantization of the eleven-dimensional theory in future work.

\section{Infinite-dimensional symmetry in flat backgrounds}
\label{s:E(5,10)}

\subsection{Global symmetry algebra}
\label{sec:global}

In any field theory, the cohomology classes of states of odd ghost number have the structure of a Lie algebra. 
More generally, after shifting the cohomological degree by one, the full cohomology of states with respect to the linear BRST operator is naturally a graded Lie algebra. 
If we forget the grading to a $\ZZ/2$ grading, then this global symmetry algebra has the structure of a super Lie algebra. 

In general, taking cohomology loses information. 
If the dg Lie (or $L_\infty$) algebra we start with is not formal, then there exist higher-order operations on the linearized BRST cohomology. 
We will refer to this $L_\infty$ algebra as the global symmetry algebra of the theory.

Before taking cohomology with respect to the linear BRST operator, we described the super $L_\infty$ structure on the parity shift of the eleven-dimensional fields in the previous section. 
This is encoded by the full BV action of the eleven-dimensional theory.
The cubic component of the full BV action induces the super Lie algebra structure present in the linearized BRST cohomology. 

Our main result is to relate the global symmetry algebra of the minimal twist of eleven-dimensional supergravity on $\CC^5 \times \RR$ to a certain infinite-dimensional exceptional super Lie algebra studied by Kac \cite{KacBible,KacE510} called $E(5,10)$.
We recall the definition below. 

\begin{thm}\label{thm:global}
Let $\Pi\cE(\CC^5 \times \RR)$ be the parity shift of the fields of eleven-dimensional supergravity on $\CC^5 \times \RR$ and denote by $\delta^{(1)}$ the linearized BRST operator. 
\begin{enumerate}
\item 
As a super Lie algebra, the $\delta^{(1)}$-cohomology of $\Pi\cE(\CC^5 \times \RR)$ is isomorphic to the trivial one-dimensional central extension of the super Lie algebra $E(5,10)$.
\item 
The global symmetry algebra is equivalent, as a super $L_\infty$ algebra, to the non-trivial central extension of $E(5,10)$ determined by the even cocycle defined in~\eqref{eqn:cocycle}. 
\end{enumerate}
\end{thm}

This result implies that the action functional $S_{BF, \infty} + J$ of the eleven-dimensional theory is invariant for the infinite-dimensional Lie algebra $E(5,10)$. 




\subsection{Linearized BRST cohomology} 

We compute the linearized BRST cohomology of eleven-dimensional supergravity.
Then we will describe the induced structure of a super Lie algebra present in the parity shift of the cohomology, proving part (1) of Theorem~\ref{thm:global}.

\parsec[]

First we recall the definition of the exceptional simple super Lie algebra $E(5,10)$. 
Recall that $\Vect_0 (\CC^5)$ is the Lie algebra of divergence-free holomorphic vector fields on $\CC^5$.
Let $\Omega^{2}_{cl} (\CC^5)$ be the module of holomorphic $2$-forms that are closed for the holomorphic de Rham operator $\del$.

The even part of the super Lie algebra $E(5,10)$ is the Lie algebra
\[
E(5,10)_+ = \Vect_0(\CC^5)
\]
of divergence-free vector fields on $\CC^5$,
whose elements we continue to denote by $\mu$. 
The odd piece is the module 
\[
E(5,10)_- = \Omega^{2}_{cl} (\CC^5),
\]
whose elements we denote by $\alpha$. 
Besides the natural module structure, there is odd bracket $ E(5,10)_-\otimes E(5,10)_\to E(5,10)_+$
The bracket uses the isomorphism $\Omega^{-1} \vee (-) \colon \Omega^{4} \cong \Vect (\CC^5)$ induced by the standard Calabi--Yau form $\d^5 z$, and is defined by
\beqn\label{eqn:e510}
[\alpha, \alpha'] = \Omega^{-1} \vee (\alpha \wedge \alpha') .
\eeqn
Since both $\alpha, \alpha'$ are closed two-forms,  the resulting vector field on the right hand side is divergence free. 
In coordinates, if $f^{ij} \d z_i \wedge \d z_j$ and $g^{kl} \d z_k \wedge \d z_l$ are two closed two-forms, their bracket is the vector field $\ep_{ijklm} f^{ij}g^{kl} \partial_{z_m}$. 

To be precise, Kac studied a more algebraic version of the algebra we have just introduced, where holomorphic functions are replaced by holomorphic polynomials.
As such, the simple super Lie algebra that appears in the classification in~\cite{KacBible} is a dense sub Lie algebra of what we call $E(5,10)$, consisting of those vector fields and two-forms that have polynomial coefficients.

\parsec[]

If $\cE$ is the space of fields of any theory in the BV or BRST formalism, the shift $\cL = \cE[-1]$ has the structure of a Lie, possibly $L_\infty$ algebra. 
In the $\ZZ/2$ graded world, the parity shifted object $\cL = \Pi \cE$ has the structure of a super $L_\infty$ algebra. 

In this section, we use the description of the eleven-dimensional theory as the deformation of the BF action $S_{BF,\infty}$ by the functional $J$ of Theorem \ref{thm:dfn}. 
We set the coupling $g = 1$. For any other nonzero value of $g$, we will obtain an isomorphic super $L_\infty$ algebra as explained above.
We would also obtain equivalent  results if we used the other model of the eleven-dimensional theory explained in~\S\ref{s:altdfn}. 

The full differential on the cochain complex of observables of the theory is given by the BV bracket with the BV action. 
For us, this~is 
\[
\delta = \{S_{BF,\infty} + J, -\} .
\]
The linear BRST operator (dual to the differential on the cochain complex of fields) comes only from the quadratic summands in $S_{BF,\infty}$, and is of the form
\beqn\label{eqn:linearBRST}
\delta^{(1)} = \dbar + \d_{\RR} + \div |_{\mu \to \nu} + \del |_{\beta \to \gamma} .
\eeqn

To compute the cohomology with respect to $\delta^{(1)}$ we can use a spectral sequence, first taking the cohomology with respect to $\dbar + \d_{\RR}$ and then with respect to $\div$. 
By the $\dbar$ and de Rham Poincar\'e lemmas, the cohomology of the space of fields of the eleven-dimensional theory on $\CC^5 \times \RR$ with respect $\dbar + \d_{\RR}$ results in the cochain complex
\begin{equation}
  \label{eq:lin1} 
  \begin{tikzcd}[row sep = 1 ex]
    - & + \\ \hline
    \Vect(\CC^5) \ar[r, "\div"] & \cO(\CC^5) \\ 
     \cO(\CC^5) \ar[r, "\del"] & \Omega^{1}(\CC^5).
\end{tikzcd}
\end{equation}
Recall that $\Vect(\CC^5), \cO(\CC^5)$, and $\Omega^1(\CC^5)$ denote the space of holomorphic vector fields, functions, and one-forms, respectively.

The cohomology with respect to the remaining linearized BRST operator consists of the space of triples $(\mu, [\gamma], b)$ where:
\begin{itemize}
\item $\mu$ is a divergence-free holomorphic vector field on $\CC^5$, which is constant along $\RR$
\[
\mu = \mu \otimes 1 \in \Pi \Vect_0(\CC^5) \otimes \Omega^0(\RR) .
\]
Note that $\mu$ is a ghost in the $\ZZ/2$ graded theory. 
\item $[\gamma]$ is an equivalence class of a holomorphic one-form modulo exact holomorphic one-forms along $\CC^5$, which are also constant along $\RR$
\[
[\gamma] = [\gamma] \otimes 1 \in \left(\Omega^{1}(\CC^5) / \d \cO(\CC^5) \right) \otimes \Omega^0(\RR) .
\]
\item A constant function $b \in \Pi \CC$ on $\CC^5 \times \RR$.
This is a $\beta$-type field in the eleven-dimensional theory, any constant function is closed for the de Rham differential. 
This element is also a ghost in the $\ZZ/2$-graded theory. 
\end{itemize}

\parsec[]

After parity shifting, we've identified the solutions to the linear equations of motion with triples
\[
(\mu, [\gamma], b) \in \Vect_0(\CC^5) \oplus \Pi \Omega^{1}(\CC^5) / \del \cO(\CC^5) \oplus \CC .
\]
The bracket induced by the cubic component of $S_{BF, \infty}$ in the classical BV action is the usual bracket on divergence-free vector fields together with the module structure on holomorphic one-forms by Lie derivative.
Notice that the Lie derivative commutes with the $\del$ operator, so this action descends to equivalence classes as above. 
The elements $b$ are central. 

The final term in the BV action $J = \frac16\int \gamma \wedge \del \gamma \wedge \del \gamma$ induces the following Lie bracket on the solutions to the linearized equations of motion
\beqn\label{eqn:eqb}
\big[[\gamma], [\gamma'] \big] = \Omega^{-1} \vee (\del \gamma \wedge \del \gamma') \in \Vect_0(\CC^5) .
\eeqn
where $\Omega^{-1}$ denotes the section of $\PV^{5,hol}(\CC^5)$ which is inverse to the Calabi--Yau form $\Omega$ on $\CC^5$. 
Notice that this bracket is well-defined as it does not depend on the particular equivalence classes and that the resulting vector field is automatically divergence-free.

\parsec[]

Having described the linearized BRST cohomology as a super vector space, we turn to the proof of Theorem \ref{thm:global}.

\begin{proof}[Proof of Theorem \ref{thm:global}]
For the first part, we write down an explicit map between the cohomology computed above and the algebra $E(5,10)$. 

The relationship of the $\mu$-elements in $E(5,10)$ and the eleven-dimensional theory is apparent.

Next, we need to relate the equivalence classes $[\gamma]$ with the closed two-forms $\alpha$ in $E(5,10)$. 
On flat space, any closed differential form is exact (this is a holomorphic version of the Poincar\'e lemma). 
In other words, there is an isomorphism
\[
\del \colon \Omega^1 (\CC^5) / \d \cO(\CC^5) \xto{\cong} \Omega^{2}_{cl}(\CC^5)
\]
induced by the holomorphic de Rham differential.
This gives the relationship between the equivalence class $[\gamma]$ in the eleven-dimensional theory and a closed two-form in $E(5,10)$ by $\alpha = \del \gamma$. 
It is clear from Equations \eqref{eqn:e510} and \eqref{eqn:eqb} that this assignment intertwines the Lie brackets in $E(5,10)$ and the twist of eleven-dimensional supergravity. 
This completes the proof of part (1).

For part (2), we first produce the following homotopy data:
\begin{equation}
\begin{tikzcd}
\arrow[loop left]{l}{K}(\Pi \cE , \delta^{(1)})\arrow[r, shift left, "q"] &(E(5,10) \oplus \CC_b \, , \, 0)\arrow[l, shift left, "i"] \: ,
\end{tikzcd}
\end{equation}

\begin{itemize}
\item On the $\nu$'s we take $K$ to be any operator $K \colon \cO \to \Vect$ such that $\div K \nu = \nu$. 
On the $\gamma$'s we take $K$ to be any operator $K \colon \Omega^1 \to \Omega^0$ such that $\del K(\gamma) = \gamma$. 
Also, introduce the auxiliary operator $\til{K} \colon \Omega^2_{cl} \to \Omega^1$ which satisfies the homotopy relation
\beqn\label{eqn:htpy1}
\til{K} \del \gamma + \del K \gamma = \gamma . 
\eeqn
The precise form of each of these operators will not be needed.
The existence of such operators is guaranteed by the holomorphic Poincar\'e lemma.
The operator $K$ annihilates fields $\beta$ and $\mu$. 
\item 
The map $q$ is described as follows. 
First $q(\mu) = \mu - K \div (\mu)$.
Notice that $q(\mu)$ is automatically divergence-free.
Next, $q(\gamma) = [\gamma]$, the equivalence class in $\Omega^1 / \d$. 
If $\beta$ is a holomorphic function, then $q(\beta) = \beta (z=0)$.
\item 
The map $i$ embeds $\mu$ and $b$ in the obvious way.
On the equivalence class $[\gamma] \in \Omega^1 / \d$ we define $i([\gamma]) = \gamma - \til{K} \del \gamma$. 
Notice that this is independent of the choice of representative $\gamma$. 
\end{itemize}

It is straightforward to check that this comprises well-defined homotopy data, the only nontrivial thing to check is the relation $\id - i \circ q = \delta^{(1)} K - K \delta^{(1)}$. 
Plugging in the field $\gamma$ we see that we must check that
\[
\gamma - \til{K} \del \gamma = \del K \gamma 
\]
which is precisely \eqref{eqn:htpy1}. 

Given this homotopy data, we can compute the homotopy transferred $L_\infty$ structure on the linearized BRST cohomology. 
Since $\nu$ does not survive to cohomology and the fact that there are no nontrivial Lie brackets involving the field $\beta$, this transferred structure is easy to compute. 

There is a single diagram which contributes to the transferred structure, it is given by
\begin{equation}
\begin{tikzpicture}
\begin{feynman}
\vertex(a) at (-1,1) {$i(\mu)$};
\vertex(b) at (-1,0) {$i([\gamma])$};
\vertex(c) at (-1,-1) {$i(\mu')$};
\vertex(d) at (0,0.5);
\vertex(e) at (1,0);
\vertex(f) at (2,0) {$q$};
\diagram* {(a)--(d), (b)--(d), (d)--[edge label = $K$](e), (c)--(e), (f)--(e)};
\end{feynman}
\end{tikzpicture}
\end{equation}
together with a similar diagram with the $\mu$ and $\mu'$ flipped. 

This diagram leads to a new $3$-ary bracket on $E(5,10) \oplus \CC_b$
\[
\big[\mu,\mu',[\gamma]\big]_3 = \varphi(\mu,\mu',[\gamma])
\]
where $\varphi \in \clie^\text{even} (E(5,10))$ is the even Lie algebra cocycle defined by the formula
\beqn
\begin{array}{rclr}
\varphi \colon E(5,10) \times E(5,10) \times E(5,10) & \to & \CC_b \\
\varphi(\mu,\mu',\alpha) & = & \<\mu \wedge \mu', \alpha\>|_{z=0} .
\label{eqn:cocycle}
\end{array}
\eeqn
Since $b$ is central, this cocycle defines a central extension of $E(5,10)$.
\end{proof}

\parsec[]
We briefly remark on Lie algebra cohomology for super Lie algebras.
The Lie algebra cohomology $\clie^{\bu,\bu}(\cL)$ of any super Lie algebra $\cL$ is graded by $\ZZ \times \ZZ/2$. 
The first grading is by the symmetric degree in the Chevalley--Eilenberg complex.
The second grading is the internal parity of the super Lie algebra $\cL$. 
The Chevalley--Eilenberg differential is degree $(1,+)$. 

The cocycle $\varphi$ has homogenous bigrading $(3,-)$.
In the above discussion we forgot the bigrading to a totalized $\ZZ/2$ grading where 
\begin{align*}
\clie^\text{even} (\cL) & = \clie^{2\bu , +} (\cL) \oplus \clie^{2\bu+1, -}(\cL) \\
\clie^\text{odd} (\cL) & = \clie^{2\bu , -} (\cL) \oplus \clie^{2\bu+1, +}(\cL) .
\end{align*}
With this totalization, $\varphi$ is an even cocycle and hence determines a super $L_\infty$ central extension by the one-dimensional even vector space $\CC$. 




\section{Residual supersymmetry} 
\label{sec:susy}

In this section we consider the minimal twist of eleven-dimensional supersymmetry explicitly. 
We compute the residual supersymmetry algebra given by taking the cohomology of the eleven-dimensional supersymmetry algebra with respect to the minimal twisting supercharge. 
In order for this to map to the gauge symmetries of the eleven-dimensional theory, it is necessary to consider an extension of the eleven-dimensional supersymmetry algebra corresponding to the M2 brane.
We will see how this extension is compatible, upon twisting by the minimal supercharge, with the central extension of $E(5,10)$ we found as the global symmetry algebra in the previous section.

\subsection{Supersymmetry in eleven dimensions}
\label{sec:11dsusy}

The (complexified) eleven-dimensional supertranslation algebra is a complex super Lie algebra of the form
\[
  \ft_{11d} = V \oplus \Pi S
\]
where $S$ is the (unique) spin representation and $V \cong \CC^{11}$ the complex vector representation, of~$\lie{so}(11, \CC)$. 
The bracket is the unique surjective $\lie{so}(11,\CC)$-equivariant map from the symmetric square of~$S$ to~$V$;
this decomposes into three irreducibles, 
\beqn\label{eqn:decomp}
  \Sym^2(S) \cong V \oplus \wedge^2 V \oplus \wedge^5 V.
\eeqn
Denote by $\Gamma_{\wedge^1}, \Gamma_{\wedge^2}, \Gamma_{\wedge^5}$ the projections onto each of the summands above. 
The bracket in $\ft_{11d}$ is defined using the first projection
\[
[\psi, \psi'] = \Gamma_{\wedge^1} (\psi, \psi') .
\]
The super Poincar\'{e} algebra is
\[
  \lie{siso}_{11d} = \lie{so}(11 , \CC) \ltimes \ft_{11d} .
\]
The $R$-symmetry is trivial in eleven-dimensional supersymmetry. 

\subsection{Extensions of the supersymmetry algebra} 
\label{sec:m2brane}

Extensions of the supersymmetry algebra correspond to the existence of extended objects, such as branes, in the supergravity theory.
In eleven-dimensional supersymmetry, there are two such extensions corresponding to the M2 brane and the M5 brane.
We begin by describing a less standard dg Lie algebra model for the M2 brane algebra.
In the next section we will explain the relationship to other descriptions in terms of $L_\infty$ algebras \cite{Basu_2005,Bagger_2007,FSS}. 

Our model for the M2 brane algebra is a dg Lie algebra extension of the super Poincar\'e algebra $\lie{siso}_{11d}$.
 
Introduce the cochain complex $\Omega^{\bu}(\RR^{11})$ of (complex valued) differential forms on $\RR^{11}$ equipped with the de Rham differential $\d$.
The M2 brane algebra arises as an extension of $\lie{siso}_{11d}$ by the cochain complex $\Omega^\bu(\RR^{11})[2]$ and is defined by a cocycle
\[
    c_{M2} \in \clie^{2,+} \left(\lie{siso}_{11d} \; ; \; \Omega^\bu (\RR^{11})[2]\right) .
\]
The formula is
  \[c_{M2} (\psi, \psi') = \Gamma_{\wedge^2}(\psi, \psi') \in \Omega^2(\RR^{11})\]
  where $\Gamma_{\wedge^2}$ is the projection onto $\wedge^2 V$, thought of as the space of constant coefficient two-forms, as in the decomposition \eqref{eqn:decomp}.

\begin{dfn}
The algebra $\m2$ is the $\ZZ \times \ZZ/2$-graded dg Lie algebra defined by the extension of $\lie{siso}_{11d}$ by the cocycle $c_{M2}$.  
\end{dfn}
  
Here, we are using a bigrading by $\ZZ \times \ZZ/2$. 
The super Poincar\'e algebra is concentrated in zero integer grading and carries its natural $\ZZ/2$ grading as a super Lie algebra.
The complex $\Omega^{\bu}(\RR^{11})[2]$ is concentrated in integer degrees $[-2,9]$ and has even parity.
The bracket in $\m2$ is bidegree $(0,+)$ and the differential is bidegree $(1,+)$.

\subsection{The minimal twist}
\label{sec:mintwist}

Fix a supercharge $Q \in S$ satisfying $Q^2 = 0$ that is in the lowest stratum of the nilpotence variety.
Such a supercharge has a six-dimensional image in the space of (complexified) translations on $\RR^{11}$ and defines the minimal twist of eleven-dimensional supersymmetry \cite{SWspinor}. 
We characterize the cohomology of the algebra $\m2$ with respect to this supercharge. 

$Q$ defines a maximal isotropic subspace $L \subset V$. 
In turn, we will decompose the super Poincar\'e algebra into $\lie{sl}(L) = \lie{sl}(5)$ representations.
First, the defining and spinor representations decompose as
\deq{
  V = L \oplus L^\vee \oplus \CC_t, \qquad S = \wedge^\bu L.
}
In the expression for $S$, we are omitting factors of $\det(L)^{\frac12}$ for simplicity. 
Also, $\lie{so}(11, \CC)$ decomposes as
\[
\lie{sl}(L) \oplus \wedge^2 L \oplus \wedge^2 L^\vee \oplus L \oplus L^\vee \oplus \C .
\]
Furthermore, the spin representation can be identified with
\[
S = \wedge^\bu (L) = \CC \oplus L \oplus \wedge^2 L \oplus \wedge^3 L \oplus \wedge^4 L \oplus \wedge^5 L .
\]
The element $Q$ lives in the first summand.
Let 
\[
{\rm Stab}(Q) =  \lie{sl}(L) \oplus \wedge^2 L^\vee \oplus L^\vee \subset \lie{so}(11,\CC)
\]
 be the stabilizer of $Q$. 
This is a parabolic subalgebra whose Levi factor is $\lie{sl}(5)$.

\subsection{$Q$-cohomology of $\m2$}
\label{sec:m2branetwist}

Any element $Q \in S$ satisfying $Q^2 = 0$ determines a deformation of the dg Lie algebra $\m2$.
To deform $\d$ by $Q$ we must break the $\ZZ \times \ZZ/2$ bigrading.
The supercharge $Q$ is odd and of cohomological degree zero.
Recall, the original differential on $\m2$ is the de Rham differential $\d$ which just acts on the summand $\Omega^\bu(\R^{11})[2]$ and is even of cohomological degree $+1$.
Thus, only the totalized $\ZZ/2$ grading makes the differential $\d + [Q,-]$ homogenous. 

\begin{dfn}
The $Q$-twist $\m2^Q$ of $\m2$ is the super dg Lie algebra whose differential is $\d + [Q,-]$.
The bracket is unchanged.
\end{dfn}

Let $Q$ be a minimal supercharge satisfying $Q^2 = 0$. We first determine $H^\bullet(\m2^q)$ as a super vector space. 

\begin{lem}
As a $\ZZ/2$ graded space, the cohomology of the $Q$-twist $\m2^Q$ is
\beqn\label{eqn:susycoh}
L \oplus {\rm Stab}(Q) \oplus \Pi \left(\wedge^2 L^\vee\right) \oplus \CC
\eeqn
whose elements we denote by $(v, m, \psi, c)$.
\end{lem}
\begin{proof}
The cohomology of the non-centrally extended algebra was computed in \cite{SWspinor}, we briefly recall the result. 
The element $Q$ only acts nontrivially on the summands $\wedge^4 L$ and $\wedge^5 L$ in $S$. 
The image of $\wedge^4 L \cong L^\vee$ trivializes the antiholomorphic translations while the image of $\wedge^5 L$ trivializes the time translation.
So, of the translations, only the holomorphic ones, which live in $L$, survive.
The map 
\[
[Q,-] \colon \lie{so}(11,\CC) \to S 
\] 
is the projection onto $\wedge^0 L \oplus \wedge^1 L \oplus \wedge^2 L$. 
The kernel of $[Q,-]$ is the stabilizer~${\rm Stab}(Q)$.

In summary, the space of odd translations which survive cohomology is $\wedge^3 L \cong \wedge^2 L^\vee$; two such elements bracket to a holomorphic translation by taking the wedge product to get an element of $\wedge^4 L^\vee \cong L$.
This completes the calculation of the cohomology. 
\end{proof}

The main result of this subsection is the following:
\begin{prop}\label{prop:susycoh}
The cohomology of the $Q$-twist $H^\bu(\m2^Q)$ has the following structures:
\begin{enumerate}
\item As a super Lie algebra, $H^\bu(\m2^Q)$ is the natural extension of ${\rm Stab}(Q)$ together with the bracket
\beqn\label{eqn:susy2bra}
[\psi, \psi']_2 = \psi \wedge \psi' \in \wedge^4 L^\vee \cong L_v \\
\eeqn
\item 
$\m2^Q$ is not formal as a super dg Lie algebra.
As a super $L_\infty$ algebra, the $Q$-twist is equivalent to \eqref{eqn:susycoh} with $2$-brackets described in (1) where we additionally introduce the $3$-ary bracket 
\beqn\label{eqn:susy3bra}
[v, v', \psi]_3 = 4 \<v \wedge v', \psi\> \in \CC_b .
\eeqn
\end{enumerate}
\end{prop}

It will be useful to list the formulas for the brackets in terms of coordinates. 
Let $\{z_i\}$ denote a basis for $L$, which we will also think of as a linear coordinate on $\CC^5$. 
Let $\{\partial_{z_i}\}$ be a dual basis for $L^\vee$, which we will also think of as translation invariant vector fields.
The $2$-ary bracket above is 
\[
[z_i \wedge z_j, z_k \wedge z_l]_2 = \ep_{ijklm} \partial_{z_m} 
\]
and the $3$-ary bracket is
\[
[\partial_{z_i}, \partial_{z_j}, z_{k} \wedge z_{\ell}]_3 = 4 (\delta^i_k \delta^j_\ell - \delta^i_\ell \delta^j_k) .
\] 

\parsec[]
 
One way to prove the proposition above is to use homotopy transfer directly to $\m2^Q$, just as we did in the proof of Theorem \ref{thm:global} to deduce the form of the $3$-ary bracket. 
Instead, we will use the following minimal model for $\m2^Q$ to prove Proposition~\ref{prop:susycoh}.
This minimal model also has the advantage of being more directly related to the eleven-dimensional supergravity theory.

\begin{lem}
\label{lem:gmodel}
Let $\fg$ denote the following $\ZZ/2$ graded dg Lie algebra which as a cochain complex is
\[
H^\bu(\m2^Q) \oplus (L^\vee \xto{\id} \Pi L^\vee)  .
\]
Denote the elements of the second summand by $(\lambda, \til\lambda)$. 
The Lie structure extends the one on $H^\bu(\m2^Q)$ described in (1) of Proposition \ref{prop:susycoh} together with the brackets
\begin{align*}
[v,\lambda] & = \<v, \lambda\> \in \CC_b \\ 
[v,\psi] & = \<v, \psi\> \in \Pi L^\vee_{\Tilde{\lambda}}.
\end{align*}

There is an $L_\infty$ map 
\[
\fg \rightsquigarrow \m2^Q
\] 
which is a quasi-isomorphism of cochain complexes.  
\end{lem}
\begin{proof}
We embed $\fg$ into $\m2^Q$ in the following way: ${\rm Stab}(Q)$ and $L$ sit inside in the evident way.
The central element maps to $c \mapsto - 1 \in \Omega^0(\RR^{11})$.
The summand $L_\lambda$ is mapped to the linear functions in $\Omega^0(\RR^{11})$ and $\Pi L_{\Tilde{\lambda}}$ is sent to the constant coefficient one-forms in $\Pi \Omega^1(\RR^{11})$. 
It remains to define where $\psi \in \wedge^2 L$ is mapped.

Notice that, at least naively, $\psi \in \wedge^2 L$ is not $Q$-closed due to the presence of the central extension. 
To embed $\wedge^2 L$ we introduce the operator
\[
H \colon \Omega^2 (\RR^{11}) \to \Omega^1(\RR^{11}),
\]
which sends a two-form $\alpha$ to the one-form $H \alpha$ defined by the formula $(H \alpha) (x) = \int_0^x \alpha$
where we integrate over a straight line path from $0$ to $x$.

Notice that if $\alpha$ is $\d$-closed then $\d (H \alpha) = \alpha$. 
It follows that any element $\psi \in \wedge^2 L \subset S$ can be lifted to a closed element at the cochain level in $\m2^Q$ by the formula
\[
\Tilde{\psi} = \psi - H \Gamma_{\wedge^2} (Q, \psi) \in \Pi S \oplus \Pi \Omega^1 .
\]
Thus, sending $\psi \mapsto \Tilde{\psi}$ defines a cochain map $\fg \to \m2^Q$. 

The Lie bracket $[\Tilde{\psi}, \Tilde{\psi}']$ agrees with $[\psi, \psi']$. 
On the other hand, in $\m2^Q$ there is the Lie bracket 
\[
[v,\Tilde{\psi}] = - L_v (H \Gamma_{\wedge^2} (Q, \psi)) = -\<v, \Gamma_{\wedge^2}(Q, \psi)\> - \d \<v, H \Gamma_{\wedge^2}(Q, \psi)\> .
\]
The first term agrees with the bracket $[v, \psi]_{\fg}$ in $\fg$. 
The other term is exact in $\m2^Q$ and can hence be corrected by the following bilinear  
\[
v \otimes \psi \mapsto \<v, H \Gamma_{\wedge^2} (Q,\psi) \> \in L_\lambda .
\] 
Together with the cochain map described above, this bilinear term prescribes the desired $L_\infty$ map. 

\end{proof}

\parsec[]
We now proceed to the proof of proposition \ref{prop:susycoh}.

\begin{proof}[Proof of Proposition \ref{prop:susycoh}]
Using the model $\fg$, the first part of Proposition \ref{prop:susycoh} follows immediately. 
We deduce the second part using homotopy transfer. 

Recall that we described the cohomology of $\m2^Q$ in \eqref{eqn:susycoh}.
Let $\delta$ denote the differential on $\fg$ which simply maps $\Pi L$ to $L$ by the identity map. 
We produce the homotopy data
\begin{equation}
\begin{tikzcd}
\arrow[loop left]{l}{K}(\fg , \delta)\arrow[r, shift left, "q"] &(H^\bu(\m2^Q) \, , \, 0)\arrow[l, shift left, "i"] \: ,
\end{tikzcd}
\end{equation}
as follows.
\begin{itemize}[leftmargin=\parindent]
\item The operator $K$ annihilates $H^\bu(\m2^Q)$ and is the identity map~$K \colon \Pi L_{\til \lambda} \to L_\lambda$. 
\item The map $q$ is the identity on $H^\bu(\m2^Q)$ and annihilates the summand~$L \to \Pi L$. 
\item The map $i$ embeds $H^\bu(\m2^Q)$ in the obvious way. 
\end{itemize}

It is immediate to verify this data prescribes valid homotopy data.
There is only a single term in the $L_\infty$ structure generated by homotopy transfer. 
It is determined by the following tree diagram
\begin{equation}
\begin{tikzpicture}
\begin{feynman}
\vertex(a) at (-1,1) {$i(v)$};
\vertex(b) at (-1,0) {$i(\psi)$};
\vertex(c) at (-1,-1) {$i(v)$};
\vertex(d) at (0,0.5);
\vertex(e) at (1,0);
\vertex(f) at (2,0) {$q$};
\diagram* {(a)--(d), (b)--(d), (d)--[edge label = $K$](e), (c)--(e), (f)--(e)};
\end{feynman}
\end{tikzpicture}
\end{equation}
together with a similar diagram with the $v$ and $v'$ reversed. 
It is an immediate calculation to show that these trees recover the formula in (2) of Proposition \ref{prop:susycoh}.
\end{proof}

\subsection{Embedding supersymmetry into the eleven-dimensional theory} \label{s:residual}

Consider now the super $L_\infty$ algebra $\cL$ underlying the eleven-dimensional theory on $\CC^5 \times \RR$. 

\begin{prop}
Endow the cohomology of $\m2^Q$ with the $L_\infty$ structure of Proposition \ref{prop:susycoh} and let $\cL(\CC^5 \times \RR)$ be the super $L_\infty$ algebra underlying eleven-dimensional supergravity on $\CC^5 \times \RR$. 
There is a map of super $L_\infty$ algebras 
\[
H^\bu(\m2^Q) \rightsquigarrow \cL (\CC^5 \times \RR)
\]
In particular, the $Q$-twisted algebra $\m2^Q$ is a symmetry of eleven-dimensional theory on $\CC^5 \times \RR$. 
\end{prop}
\begin{proof}
Recall that the cohomology of $\m2^Q$ takes the following form
\beqn 
\begin{tikzcd}
\ul{even} & \ul{odd} & \ul{even} \\
 L^\vee & (\wedge^2 L^\vee)_2 & L \\
(\wedge^2 L^\vee)_1 & & \\
\lie{sl}(5) && \CC_b 
\end{tikzcd}
\eeqn
The lefthand column is ${\rm Stab}(Q)$. 
The subscripts are used to distinguish between the two copies of $\wedge^2 L^\vee$.

The $L_\infty$ map from the dg Lie model $\fg$ to the fields of the twisted eleven-dimensional supergravity theory has a linear piece $\Phi^{(1)}$ and a quadratic piece $\Phi^{(2)}$.
Define the linear map $\Phi^{(1)} \colon \fg \to \cL$ as follows:
\begin{align*}
 L^\vee & \mapsto 0 \\
 \wedge^2 L_1^\vee  & \mapsto 0 \\
z_i \wedge z_j \in \wedge^2 L^\vee_2 & \mapsto \frac12 (z_i \d z_j - z_j \d z_i) \in \Omega^{1,0} (\CC^5) \hotimes \Omega^0 (\RR) \\
A_{ij} \in \lie{sl}(5) & \mapsto \sum_{ij} A_{ij} z_i \partial_{z_j} \in \PV^{1,0}(\CC^5) \hotimes \Omega^0(\RR) \\ \partial_{z_j} \in L & \mapsto
\partial_{z_i} \in \PV^{1,0} (\CC^5) \hotimes \Omega^0 (\RR^5) \\ 
1 \in \CC_b & \mapsto 1 \in \Omega^{0,0}(\CC^5) \hotimes \Omega^0 (\RR) .
\end{align*}

It is immediate to check that this is a map of cochain complexes, since all elements in the image of this map lie in the kernel of the linearized BRST operator~\eqref{eqn:linearBRST}. 

This map also preserves the bracket between odd elements in $\wedge^2 L_2^\vee$. 
In the cohomology of $\m2^Q$ we have the bracket
\[
[z_i\wedge z_j , z_k \wedge z_l] = \ep_{ijklm} \partial_{z_m}
\]
which is precisely the bracket induced by the cubic term in the action $J = \frac16 \in \gamma \del \gamma \del \gamma$. 

This map does not preserve all of the brackets, however. 
Indeed, in the eleven-dimensional theory $\cL(\CC^5 \times \RR)$ there is the bracket 
\[
\left[\partial_{z_i}, z_j \d z_k - z_k \d z_j\right] = \delta^i_j \d z_k - \delta^i_k \d z_j 
\]
arising from the cubic term in $\frac12 \int \frac{1}{1-\nu} \mu^2 \del \gamma$. 
To remedy the failure for $\Phi^{(1)}$ to preserve the brackets, we introduce the odd bilinear map $\Phi^{(2)} \colon \fg \otimes \fg \to \Pi \cL$ defined by 
\beqn\label{eqn:phi2}
\Phi^{(2)} \left(\partial_{z_i} , z_j \wedge z_k\right) = \frac12 (\delta^i_j z_k - \delta^i_k z_j) .
\eeqn
Notice that the field on the right hand side is of type $\beta$. 

The bilinear map $\Phi^{(2)}$ provides a homotopy trivialization for the failure for $\Phi^{(1)}$ to preserve the $2$-ary bracket: 
\[
[\Phi^{(1)} (\partial_{z_i}) , \Phi^{(1)}(z_j \wedge z_k)] = \del \Phi^{(2)}\left(\partial_{z_i} , z_j \wedge z_k\right).
\]
The lefthand side is $\frac12 (\delta_j^i \d z_k - \delta_k^i \d z_j)$ which is precisely the de Rham differential applied to \eqref{eqn:phi2}.

To define an $L_\infty$ morphism $\Phi^{(1)} + \Phi^{(2)}$ must satisfy additional higher relations. 
There is a single nontrivial cubic relation to verify:
\begin{multline} \label{eqn:cubicrln}
\Phi^{(1)}\left[\partial_{z_i}, \partial_{z_j}, z_k \wedge z_l\right]_3 = [\Phi^{(1)}(\partial_{z_i}), \Phi^{(1)}(\partial_{z_i}), \Phi^{(1)}(z_k \wedge z_l)]_3 \\ + [\del_{z_i}, \Phi^{(2)}(\partial_{z_j}, z_k \wedge z_l)] + [\del_{z_j}, \Phi^{(2)}(\del_{z_i}, z_k \wedge z_l)]
\end{multline}
where $[-]_3$ on the left hand side is the $3$-ary bracket defined in Proposition \ref{prop:susycoh} and $[-]_3$ on the right hand side is the $3$-ary bracket defined by the quartic part of the action $\frac12 \int \frac{1}{1-\nu} \mu^2 \vee \del \gamma$. 
The two terms in the second line of \eqref{eqn:cubicrln} cancel for symmetry reasons and the quartic term in the BV action induces precisely the correct $3$-ary bracket. 

\end{proof}

\parsec[] From the previous proposition we can readily compare the super $L_\infty$ algebra $H^\bu(\m2^Q)$ with the global symmetry algebra of our theory.
\begin{cor}
There is a map of super $L_\infty$ algebras 
\[
H^\bu(\m2^Q)\to \Hat{E(5,10)},
\]
where $\Hat{E(5,10)}$ is a central extension of~$E(5,10)$ by the cocycle~\eqref{eqn:cocycle}.
\end{cor}
\begin{proof}
Because this map preserves differentials, it descends to a map in cohomology. 
We have already computed the cohomology of $\cL$ on $\CC^5 \times \RR$; it is the trivial one-dimensional central extension of $E (5,10)$. 
The Lie algebra structure present in the cohomology of $\m2^Q$ is described in part (1) of Proposition~\ref{prop:susycoh}. 
The map
\[
L \oplus {\rm Stab}(Q) \oplus \Pi \left(\wedge^3 L \right) \oplus \CC_b \to E (5,10) \oplus \CC_{b'}
\]
is defined by very similar formulas as above
\begin{align*}
 L^\vee_1 & \mapsto 0 \\
 \wedge^2 L^\vee_1 & \mapsto 0 \\
z_i \wedge z_j \in \wedge^2 L_2 & \mapsto \d z_i \wedge \d z_j \in \Omega^{2}_{cl} (\CC^5) \\
A_{ij} \in \lie{sl}(5) & \mapsto \sum_{ij} A_{ij} z_i \partial_{z_j} \in \Vect_0(\CC^5) \\ \partial_{z_i} \in L & \mapsto
\partial_{z_i} \in \Vect_0(\CC^5) \\
b \in \CC_b & \mapsto b \in \CC_{b'} .
\end{align*}

The relationship between the transferred $L_\infty$ structures can be described as follows. 
Recall that the linear BRST cohomology of the parity shift of the fields of the eleven-dimensional theory is equivalent to the super $L_\infty$ 
algebra $\Hat{E(5,10)}$.
Also, we described the $L_\infty$ structure present in the cohomology of $\m2^Q$ in part (2) of Proposition~\ref{prop:susycoh}. 
Each of these $L_\infty$ structure involved introducing a single new $3$-ary bracket, which are easily seen to be compatible. 
\end{proof}
\parsec[]

In this short section we compare to another description of the M2 brane algebra given as a one-dimensional $L_\infty$ central extension of the super Poincar\'e algebra.
Such central extensions were studied in~\cite{BHsusyII, SSS, FSS}, following~\cite{CDF}. 
In these references, the algebra $\m2$ is defined as an $L_\infty$  
central extension of $\lie{siso}_{11d}$. 

Recall that given two spinors $\psi, \psi' \in S$ we can form the constant coefficient two-form $\Gamma_{\wedge^2} (\psi, \psi')$. 
Using this two-form we can define the following four-linear expression
\[
\mu_2 (\psi, \psi',v,v') = \<v \wedge v', \Gamma(\psi, \psi')\> .
\]
This expression is symmetric on the spinors and antisymmetric on the vectors, therefore it defines an element in $\clie^4(\lie{siso}_{11d})$. 
This expression defines a nontrivial class in $H^4(\lie{siso}_{11d})$ so defines a one dimensional central extension of $\lie{siso}_{11d}$ as a Lie 3-algebra. 
Instead of working with a one-dimensional central extension by $\CC[2]$, we work with a central extension by the resolution $\Omega^\bullet(\RR^{11})[2]$ determined by a cocycle $c_{M2}$, see \S \ref{sec:m2brane}. 
There is a quasi-isomorphism $\clie^\bu(\lie{siso}_{11d})\to \clie^\bu(\lie{siso}_{11d}, \Omega^\bullet (\R^{11}))$ induced by the embedding of constant functions into the full de Rham complex.
The cocycles $\mu_2$ and $c_{M2}$ are cohomologous via a two-step zig-zag in the double complex $\clie^\bu(\lie{siso}_{11d}, \Omega^\bullet (\R^{11}))$. 

\section{The non-minimal twist}\label{s:nonmin}

We have provided numerous consistency checks that the eleven-dimensional theory defined on a manifold with $SU(5)$ holonomy is a twist of supergravity. 
We have referred to this theory as ``minimal,'' since it renders the minimal number of translations homotopically trivial, or (slightly improperly) as ``holomorphic.''
In this section we characterize the unique further twist of eleven-dimensional supergravity on flat space, as seen through the lens of the holomorphic theory.  
This further twist is invariant for the group $G_2 \times SU(2)$ and is fully topological along seven directions, as opposed to just a single direction as in the minimal twist. This is easiest to see by decomposing the eleven-dimensional spinor as a representation of $\Spin(4) \times \Spin(7)$; from this perspective, a square-zero element is a rank-one element in the tensor product of a chiral spin representation of~$\Spin(4)$ and the spin representation of~$\Spin(7)$. Elements of the latter fall into two distinct orbits under the $\Spin(7)$ action, the minimal orbit---``Cartan pure spinors''---and the generic orbit~\cite{Igusa}. The stabilizer of an element of the generic orbit is~$G_2$, almost by definition.

We will show that the non-minimal twist is equivalent to an interacting theory on $\CC^2 \times \RR^7$ that we call ``Poisson'' Chern--Simons theory, using a direct description of the further twist together with an indirect cohomological argument. This completes the confirmation of a conjecture in the literature~(\cite{CostelloM5}; see also~\cite{RY}); the result was checked at the level of the free theory in~\cite{EagerHahner} by computing the nonminimal twist of the eleven-dimensional multiplet directly at the component-field level.

In the BV formalism, the theory is $\ZZ/2$ graded, with fields given by
\[
A \in \Pi \Omega^{0,\bu}(\CC^2) \hotimes \Omega^\bu(\RR^7) ,
\]
where $\Pi$, as always, denotes parity shift.
The equations of motion are of the form
\[
\dbar A + \d_{\RR^7} A + \partial_{z_1} A \wedge \partial_{z_2} A = 0 .
\]
The action functional depends on the holomorphic symplectic structure on $\CC^2$ through the Poisson bracket on the algebra of holomorphic functions.
We give a precise definition below. 

The main result of this section is the following.

\begin{thm}
\label{thm:nonmin}
The non-minimal twist of the eleven-dimensional theory is equivalent to Poisson Chern--Simons theory on 
\[
\CC^2 \times \RR^7 .
\]
\end{thm}

From the point of view of the untwisted theory, the non-minimal twist is defined by working in a background where the fermionic ghost in the physical theory is equal to a supertranslation of the form
\[
Q + Q_{nm} 
\]
where $Q$ is the supertranslation which defines the minimal twist, see \S \ref{sec:mintwist}.
The minimal twist of supergravity is obtained by setting a fermionic ghost equal to $Q$. 

In the language of the minimal twist, the supercharge $Q_{nm}$ determines a square-zero element in the $Q$-cohomology of the original supersymmetry algebra (which we will denote by the same letter). 
The characterization of this cohomology in Proposition \ref{prop:susycoh} implies that $Q_{nm}$ is an element 
\[
Q_{nm} \in \wedge^2 \left(L^\vee\right)
\]
where $L \cong \CC^5$ is the defining $SU(5)$ representation. 
In other words, $Q$ is a translation invariant holomorphic two-form on $\CC^5$. 
The condition that $[Q_{nm}, Q_{nm}] = 0$ simply says that $Q_{nm}\wedge Q_{nm} = 0$ as a translation invariant four-form on $\CC^5$. 
By a linear change of coordinates, all such two-forms $Q$ are of the form $Q_{nm} = \d z_i \wedge \d z_j$ where $i,j=1,\ldots, 5$.

From hereon in this section we will rename coordinates by
\[
\CC^5 \times \RR = \CC^2_{z_i} \times \CC^3_{w_a} \times \RR
\]
which is most natural from the point of view of the non-minimal twist. 
We will fix the non-minimal supercharge 
\[
Q_{nm} = \d z_1 \wedge \d z_2 .
\]
Notice that this choice of supercharge breaks the holonomy of the eleven-dimensional theory from $SU(5)$ to $SU(2) \times SU(3)$. 

\subsection{Index matching}
\label{sec:indexcheck}

As a first consistency check, we can compare deformation invariants attached to the holomorphic twist and the nonminimal twist. We will find that the local character of the latter agrees with a specialization of the local character computed in~\S\ref{sec:locchar}

\begin{prop}
The  local character of the nonminimal twist of eleven-dimensional supergravity on flat space is given by
\[
\prod _{(n_1,n_2)\in \Z^2_{\geq 0}} \frac{1}{1-q^{-n_1+n_2}}.
\] 
This agrees with the specializaiton of the local character computed in proposition \ref{prop:locchar}.
\end{prop}
\begin{proof}
The space of solutions to linearized equations of motion is parametrized by a holomorphic function $A$ on $\C^2_{w_j}$. The corresponding linear local operators are labeled by $(n_1,n_2)\in \Z^2_{\geq 0}$  and are given by 
\[
\boldsymbol{A}_{(n_1,n_2)} : A \mapsto \partial_{w_1}^{n_1}\partial_{w_2}^{n_2} A (0).
\]

The character of the linear span of these is given by the geometric series
\beqn\label{nonmin:singleparticle}
\sum _{(n_1,n_2)\in \Z^2_{\geq 0}} q^{-n_1+n_2}
\eeqn
with plethystic exponential given by 
\[
\prod _{(n_1,n_2)\in \Z^2_{\geq 0}} \frac{1}{1-q^{-n_1+n_2}}.
\]

For the last part, it suffices to observe the specialization at the level of single particle indices. A natural choice of fugacities for $SU(2)\times SU(3)$ is given in terms of the fugacities $q_i$ for $SU(5)$ chosen in~\S\ref{sec:locchar} by requiring the additional constraints \[q_1q_2 = 1, \ \ \ q_3q_4q_5=1.\]
After imposing the above constraints, the single particle index~\eqref{singleparticleindex} is
\[
i(q) = \frac{1}{(1-q)(1-q^{-1})}
\]
where $q = q_1=q_1^{-1}$. This is exactly the sum of the geometric series~\eqref{nonmin:singleparticle}.
\end{proof}


%

\subsection{The non-minimal global symmetry algebra}

We constructed an embedding of the $Q$-cohomology of the supersymmetry algebra into the fields of our eleven-dimensional theory on $\CC^5 \times \RR$. 
The further twist is obtained by working in a background where a certain field on $\CC^5 \times \RR$ takes nonzero value $Q_{nm}$. 
Explicitly, the element $Q_{nm} \in \wedge^2 L$ corresponds to the image under $\del$ of a $\gamma$-field of type $\Omega^{1,0}(\CC^5) \otimes \Omega^0(\RR)$. 
According to the embedding in \S \ref{s:residual} this is the $\gamma$-field 
\beqn\label{eqn:gammanm}
\gamma_{nm} = \frac12 (z_1 \d z_2 - z_2 \d z_1) \in \Omega^{1,0}(\CC^5) \otimes \Omega^0(\RR) 
\eeqn
Notice that $\del \gamma_{nm} = \d z_1 \wedge \d z_2$ as desired.

\parsec[sec:nmsymmetry]
Before proceeding to the proof of the theorem above, we perform a simple calculation of the global symmetry algebra present in the $Q_{nm}$-twisted theory. 

Recall that up to a copy of constant functions, the global symmetry algebra of the holomorphic twist of the eleven-dimensional theory is the super Lie algebra $E(5,10)$.
From this point of view, the global symmetry algebra of the $Q_{nm}$-twisted theory is given by deformation of this super Lie algebra by the Maurer--Cartan element 
\[
\d z_1 \wedge \d z_2 \in \Omega^{2}_{cl}(\CC^5) .
\]
We recall that the space of closed two-forms on $\CC^5$ is precisely the odd part of the super Lie algebra $E(5,10)$. 

We compute the cohomology of $E(5,10)$ with respect to the differential which is bracketing with this Maurer--Cartan element. 
Recall that we are using the holomorphic coordinates $(z_1,z_2,w_1,w_2,w_3)$ on $\CC^5$. 

There are the following brackets in the super Lie algebra $E(5,10)$ 
\begin{align*}
[f_l \partial_{z_l} , \d z_1 \wedge \d z_2] & = \del f_i \wedge \d z_j - \del f_j \wedge \d z_i \\
[g_a \partial_{w_a} , \d z_1 \wedge \d z_2] & = 0 \\
[h^{ab} \d w_a \wedge \d w_b , \d z_1 \wedge \d z_2 ] & = \ep_{abc} h^{ab} \partial_{w_c} .
\end{align*}
where $f_l \partial_{z_l}$, $g_a \partial_{w_a}$ are divergence-free vector fields on $\CC^5$ and $h^{ab} \d w_a \wedge \d w_b$ is a closed two-form. 

From these relations, we see that the following elements are in the kernel of $[\d z_1 \wedge \d z_2, -]$:
\begin{itemize}
\item $f(z_i, w_a)dz_1\wedge dz_2$ for $f$ a holomorphic function on $\CC^5$.
\item $f(z_i) \partial_{z_1} + g(z_i) \partial_{z_2}$ for holomorphic functions $f,g$ on $\CC_{z_1} \times \CC_{z_2}$ which satisfy 
\[
\del_{z_1} f + \del_{z_2} g = 0 .
\]
In other words, this is a divergence-free vector field on $\CC_{z_1} \times \CC_{z_2}$. 
\item $f_b(z_i, w_a) \partial_{w_b}$ for $f_b$ a holomorphic function on $\CC^5$ where $b=1,2,3$. 
\end{itemize}
It is immediate to check that these are the only nonzero elements in the kernel. 
Further, any element of the first type is clearly exact and any element of the last type is clearly exact by the closed two-form $\ep^{ijklm} f \d w_l \d w_m$. 

Thus, the cohomology is the (purely bosonic) Lie algebra of divergence-free vector fields on $\CC^2 = \CC_i \times \CC_j$
\[
H^\bu\big(E(5,10), [\d z_1 \wedge \d z_2, -] \big) \simeq \Vect_0(\CC^2) .
\]

We proved in Theorem \ref{thm:global} that the global symmetry algebra of the eleven-dimensional theory on $\CC^5 \times \RR$ is equivalent to a central extension $\Hat{E(5,10)}$ of the super Lie algebra~$E(5,10)$. 

The Lie algebra of divergence-free vector fields on $\CC^2$ also admits a central extension:
\beqn\label{eqn:centralvect}
0 \to \CC \to \cO (\CC^2) \to \Vect_0 (\CC^2) \to 0
\eeqn
where $\cO(\CC^2)$ is equipped with the Poisson bracket with respect to the symplectic form~$\d z_1 \wedge \d z_2$.
These two central extension are compatible. 


\begin{prop}
\label{prop:ham}
Let $\Hat{E(5,10)}$ be the central extension of $E(5,10)$ which is equivalent to the global symmetry algebra of the eleven-dimensional theory on $\CC^5 \times \RR$. 
Then there is an isomorphism of Lie algebras 
\[
H^\bu \big(\Hat{E(5,10)} , [\d z_1 \wedge \d z_2, -] \big) \simeq \cO(\CC^2) .
\]
\end{prop}
\begin{proof}
The only thing to check is that, in cohomology, the cocycle defining the central extension of $E(5,10)$ is the cocycle exhibiting $\cO(\CC^2)$ as the central extension of divergence-free vector fields. 
Recall that the formula \eqref{eqn:cocycle} for the cocycle is 
\[
\varphi(\mu, \mu', \alpha) = \<\mu \wedge \mu', \alpha\>|_{z=0}.
\]

In cohomology, we obtain the cocycle for divergence-free vector fields by plugging in $\alpha = \d z_1 \wedge \d z_2$. 
This gives the cocycle on $\Vect_0(\CC^2)$ 
\[
(f_i \del_{z_i}, g_j \del_{z_j}) \mapsto (f_1 g_2 - f_2 g_1)(z_1=z_2=0) .
\]
This is the cocycle defining \eqref{eqn:centralvect}, as desired. 
\end{proof}.

This proposition implies that the global symmetry algebra of the non-minimal twist of eleven-dimensional supergravity is the Lie algebra $\cO(\CC^2)$. 
We will see that this is compatible with the calculation of the non-minimal twist of the full BV theory. 

\subsection{The non-minimal twist of the eleven-dimensional theory}

Now, we turn to deducing the action functional of the non-minimal twist and hence the proof of Theorem \ref{thm:nonmin}. 
We will show that the eleven-dimensional theory on $\CC^5 \times \RR$ placed in the background where the $(1,0)$ component of $\gamma$ takes the value $\gamma_{nm}$ \eqref{eqn:gammanm} is equivalent to a theory with a purely Chern--Simons-like action functional that we referred to in the introduction to this section. 

Poisson Chern--Simons theory is defined on any manifold of the form
\[
Z \times M
\]
where $Z$ is a hyper K\"ahler surface and $M$ is a smooth manifold of real dimension seven. 
The fundamental field of the theory is  
\[
\alpha \in \Pi \Omega^{0,\bu}(Z) \; \Hat{\otimes} \; \Omega^\bu(M)  .
\]
Just in our original eleven-dimensional theory, this theory is also only $\ZZ/2$ graded. 

The holomorphic symplectic form $\omega_Z^{2,0}$ on $Z$ induces a Poisson bracket define on all Dolbeault forms $\Omega^{0,\bu}(Z)$ which we denote by $\{-,-\}_{pb}$. 
In local Darboux coordinates $(z_1,z_2)$, this bracket reads
\[
\{\alpha^I (z,\zbar) \d \zbar_I , \alpha'^J (z,\zbar) \d \zbar_J\}_{pb} = (\partial_{z_1} \alpha^I \partial_{z_2} \alpha^J \pm \partial_{z_2} \alpha^I \partial_{z_1} \alpha^J) \d \zbar_I \wedge \d \zbar_J . 
\]
The action functional of Poisson Chern--Simons theory is 
\beqn\label{eqn:pcsaction}
    \frac12 \int_{Z \times M} (\alpha \wedge \d\alpha) \wedge \omega^{2,0}_Z  + \frac16 \int_{Z\times M} \alpha \wedge \{\alpha, \alpha\}_{pb} \wedge \omega^{2,0}_Z
\eeqn
where $\{-,-\}$ is the Poisson bracket induced from the symplectic form $\omega_Z$ on $Z$. 

For simplicity, we will work only on flat space $\CC^5 \times \RR = \CC^2_z \times (\CC^3_w \times \RR)$, where we view $Z = \CC^2_z$ as a hyper K\"ahler manifold with its standard holomorphic symplectic form $\omega^{2,0} = \d^2 z$.

We will decompose the fields according to these coordinates. 
\begin{proof}[Proof of Theorem \ref{thm:nonmin}]
Decompose the $\mu$-field as $\mu = \mu_z + \mu_w$ where
\begin{align*}
\mu_z  &\in \PV^{1,\bu}(\CC^2_z) \otimes \PV^{0,\bu}(\CC^3_w) \otimes \Omega^\bu(\RR) \\
\mu_w & \in \PV^{0,\bu}(\CC^2_z) \otimes \PV^{1,\bu}(\CC^3_w) \otimes \Omega^\bu(\RR)  .
\end{align*}
and similarly $\gamma = \gamma_z + \gamma_w$. 
We will also use the notation $\del^z$ for the holomorphic de Rham differential along $\CC_z^2$ and similarly $\del^w$ for the holomorphic de Rham differential along $\CC^3_w$. 

To twist, we expand near the background where the field $\gamma_z$ takes value $\gamma_{nm}$ as in \eqref{eqn:gammanm}. 
This will generate new kinetic and interacting terms. 

There are two types of interactions in the original theory.
The first is
\begin{equation}\label{eqn:int1}
  \frac12 \int_{\CC^2 \times \CC^3 \times \RR} \frac{1}{1-\nu} \left(\del \gamma \vee \mu^2 \right) \wedge (\d^2 z \wedge \d^3 w)
\end{equation}
and the second is
\begin{equation} \label{eqn:int2}
  \frac16\int_{\CC^2 \times \CC^3 \times \RR} \gamma \partial \gamma \partial \gamma .
\end{equation}

Expanding \eqref{eqn:int1} around the background where $\gamma$ takes value $\gamma_{nm}$, we obtain,
\begin{multline}
 \int \frac{1}{1-\nu} \left(\frac12 \del^w \gamma_w \vee \mu_w^2  + \del^z \gamma_w \vee \mu_w \mu_z + \del^w \gamma_z \vee \mu_w\mu_z + \frac12 \del^z \gamma_z \vee \mu_z^2 \right) \wedge (\d^2 z \wedge \d^3 w) 
 \\
  + \frac{1}{2} \int \frac{1}{1-\nu} \left(\d^2 z \vee \mu_z^2 \right) \wedge (\d^2 z \wedge \d^3 w) .
  \label{eqn:delta1}
\end{multline}

We similarly expand (\ref{eqn:int2}),
\beqn
\frac16 \int \left(\gamma_w \partial^z \gamma_w \partial^z \gamma_w +\gamma_w \partial^w \gamma_w \partial^z \gamma_z +  \gamma_w \partial^w \gamma_z \partial^w \gamma_z \right) + \frac{1}{2} \int \left(\gamma_w \partial^w \gamma_w \right) \wedge \d^2 z
\label{eqn:delta2}
\eeqn

The new terms in the non-minimally twisted linearized BRST differential arise from the quadratic terms in the action in Equations \eqref{eqn:delta1} and \eqref{eqn:delta2}:
\begin{equation}\label{eqn:newterms}
  \frac{1}{2} \int (\d^2 z \vee \mu_z^2) \wedge (\d^2 z \wedge \d^3 w) + \frac{1}{2} \int \left(\gamma_w \wedge \partial^w \gamma_w \right) \wedge \d^2 z .
\end{equation}
The non-minimally twisted linear BRST complex thus takes the form
\beqn\label{eqn:twisteddiagram}
  \begin{tikzcd}
  & \PV^{1,\bu}_Z \hotimes \PV^{0,\bu}_W \ar[dr, "\div^z"] \ar[dashed, rounded corners, to path={ -- ([yshift=-2ex]\tikztostart.west) |- ([xshift=-1.5ex]\tikztotarget.west) -- (\tikztotarget)}, dddddr]\\
  & & \PV^{0,\bu}_Z \hotimes \PV^{0,\bu}_W \\
 & \PV^{0,\bu}_Z \hotimes \PV^{1,\bu}_W \ar[ur, "\div^w"'] & \\
\;_{\cong}  & & \Omega^{0,\bu}_Z \hotimes \Omega^{1,\bu}_W \ar[ul, dashed, bend left = 10, "\Omega^{-1}_W \partial^w"]\\
 & \Omega^{0,\bu}_Z \hotimes \Omega^{0,\bu}_W \ar[ur, "\partial^w"] \ar[dr,"\partial^z"'] \\
  & & \Omega^{1,\bu}_Z \hotimes \Omega^{0,\bu}_W
  \end{tikzcd}
\eeqn
Here, we write $Z = \CC^2_z$ and $X = \CC^3_w$ for notational simplicity. 

Here, the dashed arrow along the outside of the diagram corresponds to the BV antibracket with the first term in (\ref{eqn:newterms}).
It is given by the isomorphism 
\[
\Omega^{1,\bu}_Z \hotimes \Omega^{0,\bu}_W \xto{\omega^{2,0}_Z \otimes \id} \PV^{1,\bu}_Z \hotimes \PV^{0,\bu}_W
\]
induced holomorphic symplectic form on $Z$. 
The other dashed arrow corresponds to the BV antibracket with the second term in (\ref{eqn:newterms}).
It is given by the composition
\[
\Omega^{0,\bu}_Z \hotimes \Omega^{1,\bu}_W \xto{\id \otimes \del^w} \Omega^{0,\bu}_Z \hotimes \Omega^{2,\bu}_W \xto{\id \otimes \Omega_W} \PV^{0,\bu}_Z \hotimes \PV^{1,\bu}_W
\]
given by applying the holomorphic de Rham operator along $X$ followed by contracting with the inverse holomorphic volume form along $X$. 

We replace this linear BRST complex, up to quasi-isomorphism, with a smaller BRST complex. 
Consider the complex
\beqn\label{eqn:zw}
\Omega^{0,\bu}_Z \hotimes \Omega^{\bu,\bu}_W \hotimes \Omega^\bu_L = \oplus_{k =0}^3 \Omega^{0,\bu}_Z \hotimes \Omega^{k,\bu}_W \hotimes \Omega^\bu_L 
\eeqn
which is equipped with the differential $\dbar^z + \dbar^w + \del^w + \d_{\RR}$. 
Write $\alpha = \alpha^0 + \cdots + \alpha^3$ for a field in this complex, using the decomposition on the right hand side. 
The full Poisson Chern--Simons action $S_{pCS}$ equips this complex with the structure of a dg Lie algebra. 

Define the following non-linear map of BRST complexes from \eqref{eqn:zw} to the twisted theory \eqref{eqn:twisteddiagram}.
It is defined by the equations
\begin{multline}
\mu_z = (1-\til \alpha^3) (\del_{z_1} \wedge \del_{z_2}) \vee \del^z \alpha^0 , \quad \mu_w = (\del_{w_1} \wedge \del_{w_2} \wedge \del_{w_3}) \vee \alpha^2, \quad \nu = \til{\alpha}^3 \\
\beta = \alpha^0 , \quad \gamma_w = \alpha^1 , \quad \gamma_z = 0 .
\label{eqn:g2map}
\end{multline}
In the above equation we have introduced the notation $\til{\alpha}^3 = \Omega_X^{-1} \vee \alpha^3$. 
Notice the only non-linearity appears in the definition of $\mu_z$. 

The restriction of the kinetic terms $\int \gamma (\dbar + \d_{\RR}) \mu + \beta (\dbar + \d_{\RR}) \nu$ along \eqref{eqn:g2map} is
\beqn\label{eqn:kin1}
\int \sum_{k=0}^3 \alpha^k (\dbar + \d_{\RR}) \alpha^{3-k} 
\eeqn
The restriction of the kinetic term $\int \beta \div \mu$ along \eqref{eqn:g2map} is
\beqn\label{eqn:kin2}
\int \alpha^0 \del^w \alpha^2 - \int \alpha^0 \del^z \alpha^0 \del^z \alpha^3 . 
\eeqn
Finally, the restriction of the kinetic term $\frac12 \int \gamma \del^w \gamma$ in \eqref{eqn:delta2} along \eqref{eqn:g2map} is 
\beqn\label{eqn:kin3} 
\int \frac12 \alpha^1 \del^w \alpha^1 . 
\eeqn
Together, \eqref{eqn:kin1}--\eqref{eqn:kin3} give the kinetic term in Poisson Chern--Simons theory. 

The formulas above show that the linear terms in \eqref{eqn:g2map} define a map of linear BRST complexes.
Applying the apparent contracting homotopy, we see that this map is a quasi-isomorphism.
We will show that the full non-linear map intertwines the action functionals up to cohomologically exact terms, and hence defines an equivalence of BV theories.

We substitute the values for the fields in \eqref{eqn:g2map} into the original eleven-dimensional action. 
Notice that any terms involving $\gamma_z$ can be discarded. 
Restricting \eqref{eqn:delta1} along this map, we obtain the action functional
\beqn\label{eqn:rest1}
\frac12 \int \frac{1}{1-\til\alpha^3}  \del^w \alpha^1 (\til \alpha^2)^2 \d^2 z + \int \alpha^2 \del^z \alpha^0 \del^z \alpha^1
+ \frac12 \int (1-\alpha^3) \del^z \alpha^0 \del^z \alpha^0 .
\eeqn
Here, $\til \alpha^2$ denotes the element of $\Omega^{0,\bu}_Z \hotimes \PV^{1,\bu}_W \hotimes \Omega^\bu_L$ corresponding to $\alpha^2$ determined by the Calabi--Yau form $\d^3 w$. 
Notice that the very last term is equivalent to the functional $-\frac12\int \alpha^3 \del^z \alpha^0 \del^z \alpha^0$ since the quadratic part is a total derivative. 

There is only one cubic term left in \eqref{eqn:delta2} when we substitute the fields according to \eqref{eqn:g2map}.
It is
\beqn
\frac16 \int \alpha^1 \del^z \alpha^1 \del^z \alpha^1. 
\eeqn

Combining all of these terms, we see that the total action restricted along the map \eqref{eqn:g2map} is 
\beqn\label{eqn:pcs1}
S_{pCS} (\alpha) + \int \frac12 \frac{1}{1-\til\alpha^3} \del^w \alpha^1 (\til \alpha^2)^2 \d^2 z
\eeqn
where $S_{pCS}$ is the Poisson Chern--Simons action in \eqref{eqn:pcsaction}.

We will show that the term not appearing in $S_{pCS}(\alpha)$ is cohomologically trivial. 
Consider the odd local functional
\beqn\label{eqn:triv1}
\frac16 \int \frac{1}{1 - \til \alpha^3} \alpha^2 (\til \alpha^2)^2 .
\eeqn
Applying the linearized BRST operator (in the non-minimal twist) this becomes 
\[
\frac12 \int \frac{1}{1 - \til \alpha^3} \del^w \alpha^1 (\til \alpha^2)^2 + \frac16 \int \frac{1}{1 - \til \alpha^3} \del^w (\alpha^2) \alpha^2 (\til \alpha^2)^2 .
\]
The first term in this expression agrees the term in \eqref{eqn:pcs1} which is not in $S_{pCS}(\alpha)$.

The latter term $(1 - \til \alpha^3)^{-1} \del^w (\alpha^2) \alpha^2 (\til \alpha^2)^2$ is of polynomial degree $\geq 4$. 
Since this term has trivial self BV bracket, it determines a cocycle for the dg Lie algebra underlying Poisson Chern--Simons theory. 
Since it is manifestly translation invariant, it arises via descent from a cocycle in the dg Lie algebra of $\infty$-jets of fields at $0 \in\CC^2 \times \RR^7$.
This dg Lie algebra is quasi-isomorphic to the Lie algebra $\CC[[z_1,z_2]]$ equipped with the holomorphic Poisson bracket.
(This is the formal power series version of the Lie algebra from Proposition~\ref{prop:ham}.)

There is a weight grading on this Lie algebra, given by declaring $|z_1^{n+1} z_2^{m+1}| = n+m$; in turn, this induces a grading on the Gelfand--Fuks Lie algebra cohomology.  
The weight of the Gelfand--Fuks class corresponding to our deformation is $+2$, since no derivatives of $z_1$ or $z_2$ appear. 
Results from~\cite{Fuks} show that there is no cohomology in this weight. 
As such, the cocycle must define a trivial deformation of the dg Lie algebra up to equivalence.
This completes the proof that the non-minimal twist is equivalent to Poisson Chern--Simons theory. 
\end{proof}

\def\im{{\rm i}}

\section{Dimensional reduction and ten-dimensional supergravity}
\label{sec:dimred}

In this section we demonstrate that our proposal for the action of minimally twisted eleven-dimensional supergravity agrees with conjectural descriptions of twisted type IIA and type I supergravities due to Costello and Li. 

The original motivation for M-theory was as the strong coupling limit for type IIA string theory.
Roughly, the radius of the M-theory circle plays the role of this coupling constant. 
Additionally, at low energies M-theory is expected to be approximated by eleven-dimensional supergravity in the same way that the low energy limit of type IIA/IIB string theory is type IIA/IIB supergravity. 
Combining these two pictures, various checks have been made that the dimensional reduction of eleven-dimensional supergravity along the M-theory circle is type IIA supergravity. 

Motivated by the topological string, Costello and Li have laid out a series of conjectures for twists of type IIA/IIB supergravity \cite{CLsugra} and type I supergravity \cite{CLtypeI}. 
Their description was inspired by the model of the open and closed $B$-model topological string on a Calabi--Yau manifold. 
The open sector is holomorphic Chern--Simons theory \cite{WittenOpen} and the closed sector is called Kodaira--Spencer theory \cite{BCOV}. 
There are a few different versions of Kodaira--Spencer theory, but the shared characteristic is that they are all `gravitational' in nature; they describe fluctuations of the Calabi--Yau structure. 
From this point of view, Kodaira--Spencer theory is at the heart of the formulation of the various flavors of twisted ten-dimensional supergravity.

We begin by introducing certain variants of Kodaira--Spencer theory which will feature in the descriptions of twists of type IIA and type I supergravity.

\subsection{Kodaira--Spencer theory}
\label{s:BCOV}

Let $X$ be a Calabi--Yau manifold; for now it can be of arbitrary complex dimension $d$. 
Define
\deq{
  \PV^{i,j}(X) = \Omega^{0,j}(X, \wedge^i \T_X).
}
We will consider the graded space $\PV^{\bu,\bu}(X) = \oplus_{i,j} \PV^{i,j}(X)[-i-j]$ where the piece of type $(i,j)$ sits in degree $i+j$. 

For each fixed $i$, while we let $j$ vary, the $\dbar$ operator defines a cochain complex $\PV^{i,\bu}(X) = (\oplus_j\PV^{i,j}(X) [-j], \dbar)$ which provides a resolution for the sheaf of holomorphic polyvector fields of type $i$. 
The divergence operator extends to an operator of the form
\[
\div \colon \PV^{i,\bu}(X) \to \PV^{i-1,\bu}(X) .
\]

Motivated by the states of the topological $B$-model, one defines the fields of Kodaira--Spencer gravity on $X$ to be the cochain complex
\beqn\label{eqn:ks1}
\left(\PV^{\bu,\bu} (X)[[u]] [2] \, , \, \dbar + u \div\right) .
\eeqn 
Here, $u$ is a parameter of cohomological degree $+2$, which turns $\delta_{KS}^{(1)} = \dbar + u \div$ into an operator of homogenous degree $+1$. 
We also have performed an overall cohomological shift by $2$ so that $u^k \PV^{i,j}$ sits in degree $i+j+2k-2$. 
More precisely, this is a model for the $S^1$-equivariant cohomology of the states of the $B$-model on a closed disk. 
We refer to \cite{CLtypeI, CLsugra} for detailed justification for this ansatz. 

\parsec[s:poisson]
The original action for Kodaira--Spencer theory posited by \cite{BCOV} has a nonlocal kinetic term. In the BV formalism, this is codified by stipulating that the BV pairing is a degenerate odd Poisson tensor rather than an odd symplectic form. 
The Poisson kernel is given by the expression 
\[
(\div\otimes 1)\delta_{\Delta \subset X\times X} \in \left[\PV^{\bu,\bu}(X)\right]^{\hotimes 2} ,
\]
see \cite[{\S 1.4}]{CLbcov1}. 
Here, we view the $\delta$-distribution as a polyvector field using the Calabi--Yau form. 
Notice that the shifted Poisson tensor does not involve the parameter $u$ at all. 
For this reason, only the duals of a small number of fields pair nontrivially under the resulting odd BV bracket.

\parsec[s:ksaction] 

There is a natural local interaction which equips the complex \eqref{eqn:ks1} with the structure of $\Z/2$ graded Poisson BV theory. Explicitly, it is given by 
\beqn
I_{BCOV}(\Sigma) = {\rm Tr}_X \, \langle \exp \Sigma\rangle_0 = \sum_{n\geq 0} {\rm Tr}_X \, \langle\Sigma^{\otimes n}\rangle_0
\eeqn
where ${\rm Tr}_X \, \Phi = \int_X (\Phi \vee \Omega) \wedge \Omega$ and where $\langle - \rangle_0$ denotes the genus zero Gromov-Witten invariant with marked points
\beqn
\langle u^{k_1}\mu_1 \otimes \cdots \otimes u^{k_m}\mu_m\rangle_0 := \left (\int _{\overline {\cM}_{0,m}} \psi_1^{k_1}\cdots \psi_m^{k_m}\right ) \mu_1\cdots \mu_m = \binom{m-3}{ k_1,\cdots, k_m}  \mu_1\cdots \mu_m.
\eeqn

This interaction is extremely natural from the point of view of string field theory. Indeed, the B-model localizes to the space of constant maps into $X$, which factors as a product of $\overline{\cM}_{0,m}\times X$. This is in keeping with finding an interaction that factors as an integral over $X$ times an integral over $\overline{\cM}_{0,m}$. 

In \cite{BCOV} the authors show that the above interaction satisfies the classical master equation. Moreover, they show that the $L_\infty$ structure determined by the above action is equivalent to a natural dgla structure on the complex of fields with Lie bracket given by the Schouten bracket. Explicitly, the equivalence is given by the transcendental automorphism 
\[
\Sigma \mapsto [u(\exp (\Sigma/u)-1)]_+
\]
where $[-]_+$ denotes projection onto positive powers of $u$.

\parsec[s:minimalks]

We pointed out in \S\ref{s:poisson} that the majority of fields pair to zero under the Poisson tensor. Physically these correspond to closed string fields that do not propogate. In the supergravity approximation, the fields that survive are those closed string fields that propogate. In terms of our description of closed string field theory in terms of Kodaira--Spencer theory, this motivates us to consider the smallest cochain complex containing those fields thathave nonzero pairing under the Poisson tensor. This is referred to as minimal Kodaira--Spencer theory.

The fields of minimal Kodaira--Spencer theory are given by the subcomplex of \label{eqn:ks1}
\beqn
\left (\bigoplus_{i+j\leq d -1}u^i\PV^{j,\bu}(X)[2], \dbar + u \div\right).
\eeqn
We observe that the original odd Poisson tensor lives in this subcomplex. 
There is a natural action functional given by restricting $I_{BCOV}$ to this space.

\subsection{The $SU(4)$ twist of type IIA supergravity}\label{sec:SU(4)twist}

We recall the description of the $SU(4)$ twist of type IIA supergravity conjectured in \cite{CLsugra}. 
In principal, there is also a minimal, $SU(5)$ invariant, twist of type IIA supergravity but so far no description, even conjecturally, exists.
We turn to this in \S \ref{s:su5IIA}. 

Let $X$ be a Calabi--Yau manifold of complex dimension four. 
The $\ZZ/2$ graded complex of fields of minimal Kodaira--Spencer theory on $X$ takes the form
\beqn
\begin{tikzcd}
- & + & - & + \\ \hline
                        &                          &                                     & {\PV^{0,\bu}}    \\
                        &                          & {\PV^{1,\bu}} \arrow[r, "u\div"]    & {u\PV^{0,\bu}}   \\
                        & {\PV^{2,\bu}} \arrow[r, "u\div"]  & {u\PV^{1,\bu}} \arrow[r, "u\div"]   & {u^2\PV^{0,\bu}} \\
{\PV^{3,\bu}} \arrow[r, "u\div"] & {u\PV^{2,\bu}} \arrow[r, "u\div"] & {u^2\PV^{1,\bu}} \arrow[r, "u\div"] & {u^3\PV^{0,\bu}}
\end{tikzcd}.
\eeqn
Denote this complex by $\cE_{mKS}(X)$. 
Here, $u^\ell\PV^{k,i}$ is placed in parity $k + i -1 \mod 2$. 
The classical BCOV action $I_{BCOV}$ follows from the general formula we gave above. 

With this in hand the conjecture of \cite{CLsugra} takes the following form.

\begin{conj}
The $SU(4)$-invariant twist of type IIA supergravity on $\RR^2\times \CC^4$ is the $\Z/2$-graded Poisson BV theory with fields 
\beqn\label{eqn:IIAfields}
\alpha = \sum_n \alpha_n u^n \in \cE_{mKS}(\CC^4) \otimes \Omega^{\bu}(\RR^2).
\eeqn
The classical interaction takes the form \[I_{IIA} = \int_{\CC^4 \times \RR^2} \alpha_0^3 + \cdots\]
\end{conj}

We will need a more detailed description of the classical action. 
For the moment, let us introduce some notations for the fields of this IIA model. As always, we leave the internal Dolbeault degree implicit:
\begin{multline}
\eta \in \PV^{0,\bu}(\CC^4) \otimes \Omega^\bu (\RR^2), \quad \mu + u \nu \in \PV^{1, \bu}(\CC^4) \otimes \Omega^\bu (\RR^2) \oplus u \PV^{0,\bu} (\CC^4) \otimes \Omega^\bu (\RR^2) \\
\Pi \in \PV^{3,\bu}(\CC^4) \otimes \Omega^\bu(\RR^2), \quad \sigma \in \PV^{3,\bu}(\CC^4) \otimes \Omega^\bu (\RR^2) .
\end{multline}
We will not need an explicit notation for the remaining descendant fields. 

With this notation in hand, we have the more precise form of the action appearing in the conjecture:
\beqn\label{eqn:IIAaction}
I_{IIA} = \frac12 {\rm Tr}_{\CC^4 \times \RR^2} \frac{1}{1-\nu} \mu^2 \wedge \Pi + {\rm Tr}_{\CC^4 \times \RR^2} \frac{1}{1-\nu} \eta \wedge \mu \wedge \sigma + \frac12 {\rm Tr}_{\CC^4 \times \RR^2} \frac{1}{1-\nu} \eta \wedge \Pi^2 + \cdots 
\eeqn
where the $\cdots$ denotes terms involving higher-order descendants. 

\subsection{Reduction to IIA supergravity}
\label{s:su4red}

We now turn back to our eleven-dimensional theory. 
The first goal is to compare the dimensional reduction of our eleven-dimensional theory on $\CC^5 \times \RR$
with the $SU(4)$ invariant twist of type IIA on $\R^{2}\times \CC^4$. 
Doing so will require a slight modification to the description of the $SU(4)$ twist of IIA supergravity recollected in \S \ref{sec:SU(4)twist}. 

\parsec[sec:IIApot]

Recall that in the physical theory, the components of the $C$-field in eleven dimensions that are not supported along the M-theory circle become the components of the Ramond--Ramond 2-form of type IIA. However, as noted in \cite{CLsugra} components of Ramond--Ramond fields do not appear as fields in Kodaira--Spencer theory; rather it is components of their field strengths that appear. 
We recalled in \S \ref{s:components} that components of the $C$-field become components of $\gamma_{11d}$ in $\cE$.
This suggests that we must modify our description of the twist of type IIA to include potentials for certain fields.

The fundamental fields of the $SU(4)$ twist of IIA supergravity were given in \eqref{eqn:IIAfields}. 
We modify the space of fields by introducing potentials for both the $\Pi$ and $\sigma$ fields. 
First, we introduce a field $\gamma \in \Omega^{1,\bu}(\CC^4) \otimes \Omega^\bu(\RR^2)$ (not to be confused, yet, with the $\gamma$ field in our eleven-dimensional theory) which satisfies $\Pi \vee \Omega = \del \gamma$ where $\Omega$ is the Calabi--Yau form on $\CC^4$. 
This condition does not uniquely fix $\gamma$. 
There is a new linear gauge symmetry determined by $\gamma \to \gamma + \div \beta$ where $\beta$ is a ghost that we must also introduce. 
Similarly, we introduce a field $\theta \in \Omega^{0,\bu}(\CC^4) \otimes \Omega^\bu(\RR^2)$ which satisfies $\sigma \vee \Omega = \del \theta$, there is no extra gauge symmetry present in this condition.\footnote{Using the Calabi--Yau form we have normalized the potential fields $\gamma, \beta,\theta$ to be written as differential forms instead of polyvector fields.}

In diagrammatic detail, the potential theory we are considering has underlying cochain complex of fields
\beqn\label{eqn:IIApot}
\begin{tikzcd}
- & + \\ \hline
& {\PV^{0,\bu} (\CC^4) \otimes \Omega^\bu (\RR^2) }_\eta  \\
{\PV^{1,\bu} (\CC^4) \otimes \Omega^\bu (\RR^2)}_\mu \arrow[r, "u\div"] & u{\PV^{0,\bu} (\CC^4) \otimes \Omega^\bu (\RR^2)}_\nu \\
u^{-1}{\Omega^{0,\bu} (\CC^4) \otimes \Omega^\bu (\RR^2)}_\beta \arrow[r, "u\del"] & {\Omega^{1,\bu} (\CC^4) \otimes \Omega^\bu (\RR^2)}_\gamma  \\
{\Omega^{0,\bu} (\CC^4) \otimes \Omega^\bu (\RR^2)}_\theta &
\end{tikzcd}
\eeqn.

The original space of fields of the twist of IIA supergravity on $\CC^4 \times \RR^2$ was equipped with an odd Poisson bivector which was degenerate.
In other words, it did not define a theory in the conventional BV formalism. 
One of the key features of this new complex of fields, after we have taken these potentials, is that it is equipped with an odd nondegenerate pairing, thus equipping it with the structure of a theory in the conventional BV formalism. 

The pairing is $\Res_u \frac{\d u}{u} \int^\Omega_{\CC^4 \times \RR^2} \alpha \vee \alpha'$ where $\alpha, \alpha'$ are two general fields in this potential theory on $\CC^4 \times \RR^2$. 
Explicitly, in the description of the fields in \eqref{eqn:IIApot} the pairing is 
\[
\int^\Omega_{\CC^4 \times \RR^2} \eta \theta + \int^\Omega_{\CC^4 \times \RR^2} \mu \vee \gamma + \int^\Omega_{\CC^4 \times \RR^2} \nu \beta .
\]
This pairing is compatible with the odd Poisson bracket present in the original theory on $\CC^4 \times \RR^2$.

The type IIA action completely determines the action of this theory with potentials. 
One simply takes the \eqref{eqn:IIAaction} and replaces all appearances of $\Pi$ with $\div \gamma$ and all appearances of $\sigma$ with $\div \theta$. 
This yields the interaction of the potential theory
\beqn\label{eqn:IIAactionpot}
\til I_{IIA} = \frac12 \int^\Omega_{\CC^4 \times \RR^2} \frac{1}{1-\nu} \mu^2 \vee \del \gamma + \int^\Omega_{\CC^4 \times \RR^2} \frac{1}{1-\nu} (\eta \wedge \mu) \vee \del \theta + \frac12 \int_{\CC^4 \times \RR^2} \frac{1}{1-\nu} \eta \wedge \del \gamma \wedge \del \gamma 
\eeqn
Notice that the terms involving higher descendants vanishes since these fields are set to zero in the potential theory.


\parsec[-]

We turn to the proof of the main result of this section that the dimensional reduction of our eleven-dimensional theory agrees with the twist of IIA supergravity just introduced. 

We recall the notion of dimensional along a holomorphic direction following \cite{ESW}. 
Suppose that $V_\RR$ is a real vector space and denote by $V$ its complexification. 
We consider a field theory defined on $M \times V$, which is holomorphic along $V$ (in particular, this means that the theory is translation invariant along $V$).  
We consider the dimensional reduction along the projection 
\beqn\label{eqn:dimred}
M \times V \to M \times V_\RR
\eeqn
induced by ${\rm Re} \colon V \to V_\RR$.
Most relevant for us is the case when $V = \CC$ and $M$ is $\CC^4 \times \RR$, but the explicit form of the theory along $M$ is not important at the moment.

For illustrative purposes, let us first assume that $M$ is a point and that the space of fields is of the form $\Omega^{0,\bu}(V) \otimes W$ for $W$ some graded vector space. 
As properly formulated in \cite{ESW}, it is shown that the dimensional reduction along $V \to V_\RR$ is equivalent to the theory whose fields are $\Omega^\bu(V_\RR) \otimes W$. 
In other words, the dimensional reduction of the holomorphic theory on $V$ is a topological theory on $V_\RR$. 

If we put $M$ back in, the result is similar. 
Suppose the original theory is of the form $\cE(M) \otimes \Omega^{0,\bu}(V) \otimes W$.
Then, the dimensional reduction along \eqref{eqn:dimred} is the theory whose space of fields is $\cE(M) \otimes \Omega^\bu(V_\RR) \otimes W$.

An explicit model for this reduction can be described as follows. 
Suppose $V \cong \CC^n$ and place the theory on $(\CC^\times)^{\times n} \subset \CC^n$. 
The dimensional reduction along $\CC^n \to \RR^n$ agrees with the compactification of the theory along $S^1 \times \cdots \times S^1$ where one throws away all nonzero winding modes around each circle.

\begin{prop}\label{prop:dimred}
The $SU(4)$ invariant twist of type IIA on $\CC^4 \times \RR^2$ is the dimensional reduction of the eleven-dimensional theory along  
\[
\CC^4 \times \CC \times \RR_t \to \CC^4 \times \RR_x \times \RR_t \cong \CC^4 \times \RR^2 .
\]
\end{prop}
\begin{proof}
Let us denote the holomorphic coordinate we are reducing along by $z_5 = x + \im y$. 
We first read off the dimensional reduction of each component field of the eleven-dimensional theory. 
Per the above discussion, this is obtained by taking all fields to be independent of $y$ and replacing $\d \zbar_5$ by $\d x$. 
To not confuse the notations of fields in ten and eleven dimensions, we use the notation $\alpha_{11d}$ to denote an eleven-dimensional field.

The reductions of the eleven-dimensional fields $\nu_{11d}, \beta_{11d}$ are easy to describe. 
Recall that 
\[
\nu_{11d} \in \PV^{0,\bu}(\CC^5) \otimes \Omega^\bu(\RR) .
\]
The reduction of this field is a ten-dimensional $\nu$ field
\[
\nu (z_i,x,t) = \nu_{11d} (z_i, x, y=0, t) |_{\d \zbar_5 = \d x}  .
\]
Similarly, the reduction of $\beta_{11d}$ is a ten-dimensional $\beta$ field
\[
\beta (z_i,x,t) = \beta_{11d} (z_i, x, y=0, t) |_{\d \zbar_5 = \d x}  .
\]

The reduction of the eleven-dimensional fields $\mu_{11d}$ and $\gamma_{11d}$ require a bit of massaging. 
We break the $SU(5)$ symmetry to $SU(4)$ to write
\[
\mu_{11d} = \mu^0_{11d} + \theta_{11d} \partial_{z_5} 
\]
where
\begin{align*}
\mu^0_{11d} & \in \PV^{1,\bu}(\CC^4) \otimes \Omega^{0,\bu}(\CC_{z_5}) \otimes \Omega^\bu(\RR_t) \\
\theta_{11d} & \in \Omega^{0,\bu}(\CC^4) \otimes \Omega^{0,\bu}(\CC_{z_5}) \otimes \Omega^\bu(\RR_t) .
\end{align*}
The dimensional reduction of $\mu^0_{11d}$ is a ten-dimensional $\mu$ field
\[
\mu(z_i,x,t) = \mu_{11d}^0 (z_i, x,y=0,t)|_{\d \zbar_5 = \d x} .
\]
The dimensional reduction of $\theta_{11d}$ is a $\theta$ field
\[
\theta(z_i,x,t) = \theta_{11d} (z_i, x,y=0,t)|_{\d \zbar_5 = \d x} .
\]

Finally, write the eleven-dimensional field $\gamma_{11d}$ as
\[
\gamma_{11d} = \gamma_{11d}^0 + \eta_{11d} \d z_5
\]
where
\begin{align*}
\gamma^0_{11d} & \in \Omega^{1,\bu}(\CC^4) \otimes \Omega^{0,\bu}(\CC_{z_5}) \otimes \Omega^\bu(\RR_t) \\
\eta_{11d} & \in \PV^{0,\bu}(\CC^4) \otimes \Omega^{0,\bu}(\CC_{z_5}) \otimes \Omega^\bu(\RR_t) .
\end{align*}
The dimensional reduction of $\gamma^0_{11d}$ is a ten-dimensional $\gamma$ field
\[
\gamma(z_i,x,t) = \gamma_{11d}^0 (z_i, x,y=0,t)|_{\d \zbar_5 = \d x} .
\]
The dimensional reduction of $\eta_{11d}$ is an $\eta$ field
\[
\eta(z_i,x,t) = \eta_{11d} (z_i, x,y=0,t)|_{\d \zbar_5 = \d x} .
\]

Next, we read off the dimensional reduction of the eleven-dimensional action. 
Let us first focus on the term present in BF theory which is
$\int^\Omega \frac{1}{1-\nu_{11d}} \mu_{11d}^2 \vee \del \gamma_{11d}$.
Upon reduction, this becomes 
\beqn\label{eqn:bfred}
\int^{\Omega_{\CC^4}}_{\CC^4 \times \RR^2} \frac{1}{1-\nu} \mu^2 \vee \del \gamma + \int^{\Omega_{\CC^4}}_{\CC^4 \times \RR^2} \frac{1}{1-\nu} (\theta \wedge \mu) \vee  \del \eta 
\eeqn

Next, consider the cubic term in the eleven-dimensional action $J = \frac16 \int \gamma_{11d} \wedge \del \gamma_{11d} \wedge \del \gamma_{11d}$. 
Upon reduction, this becomes 
\beqn\label{eqn:jred}
\int_{\CC^4 \times \RR^2} \eta \wedge \del \gamma \wedge \del \gamma .
\eeqn

The sum of the action functionals \eqref{eqn:bfred} and \eqref{eqn:jred} does not precisely agree with the IIA action $\til I_{IIA}$. 
To relate the two actions we must make the following field redefinition:
\[
\til \theta = \frac{1}{1-\nu} \theta, \quad \til \eta = (1- \nu) \eta, \quad \til \beta = \beta + \frac{1}{1-\nu} \eta \wedge \theta .
\]
Notice that this change of coordinates is compatible with the odd symplectic pairing on the fields. 
Under this field redefinition the total dimensionally reduced action can be written as
\begin{multline}
\int^{\Omega_{\CC^4}}_{\CC^4 \times \RR^2} \frac{1}{1-\nu} \mu^2 \vee \del \gamma + \int^{\Omega_{\CC^4}}_{\CC^4 \times \RR^2} \frac{1}{1-\nu} \til\eta \wedge \del \gamma \wedge \del \gamma + \int^{\Omega_{\CC^4}}_{\CC^4 \times \RR^2} (\til\theta \wedge \mu) \vee  \del \left(\frac{1}{1-\nu} \til\eta\right) \\ + \int_{\CC^4 \times \RR^2}^{\Omega_{\CC^4}} \frac{1}{1-\nu} (\til \eta \wedge \til \theta) \div \mu  .
\end{multline}
The first line comes from plugging in the new fields into the interactions \eqref{eqn:bfred} and \eqref{eqn:jred}.
The second line comes from plugging in the new fields into the kinetic term $\int \beta \div \mu$, which because of the non-linear change of coordinates now contributes to the interaction. 
We observe that the first two terms agree with the first and third terms in \eqref{eqn:IIAactionpot}. 

After integrating by parts, the remaining terms can be written as 
\[
- \int^{\Omega_{\CC^4}}_{\CC^4 \times \RR^2} \left(\frac{1}{1-\nu} \til\eta\right) \div (\til\theta \mu) + \int_{\CC^4 \times \RR^2}^{\Omega_{\CC^4}} \left(\frac{1}{1-\nu} \til \eta\right) \til \theta \div \mu .
\]
Applying the identity $\div (\til \theta \mu) = \til \theta \div \mu + \del (\til \theta) \vee \mu$, we see that this agrees exactly with the second term in \eqref{eqn:IIAactionpot}.
%
%
\end{proof}

\subsection{The twist of type I supergravity}

We now turn to a different type of redution of the eleven-dimensional theory, this time involving type I supergravity. 
We begin by briefly recalling the description of type I supergravity following \cite{CLtypeI} which was motivated by the unoriented $B$-model. 
In \cite{SWspinor}, the second two authors verified the conjectural description of the space of fields recalled below using the pure spinor formalism. 
Unlike type IIA supergravity, there only exists an $SU(5)$ invariant twist of type I supergravity and it is holomorphic in the maximal number of dimensions.

Concretely, the space of fields of the $SU(5)$ twist of type I supergravity is a subspace of minimal Kodaira--Spencer theory on $\CC^5$. 
The $\ZZ/2$ graded space of fields equipped with its linear BRST operator is 
\beqn\label{eqn:IIApot}
\begin{tikzcd}
- & + & - & +  \\[-1.7em] \hline
{\PV^{1,\bu}}(\CC^5) \arrow[r, "u\div"]    & {u\PV^{0,\bu}}(\CC^5) \\
{\PV^{3,\bu}} (\CC^5)\arrow[r, "u\div"] & {u\PV^{2,\bu}} (\CC^5)\arrow[r, "u\div"] & {u^2\PV^{1,\bu}}(\CC^5) \arrow[r, "u\div"] & {u^3\PV^{0,\bu}}(\CC^5)
\end{tikzcd}.
\eeqn

Let us give a description of the classical action. 
Introduce notations for the fields of this type I model:
\beqn\label{eqn:Ifields}
\mu + u \nu \in \PV^{1, \bu}(\CC^5) \oplus u \PV^{0,\bu} (\CC^5), \quad \sigma \in \PV^{3,\bu}(\CC^5) .
\eeqn
We will not need an explicit notation for the remaining descendant fields. 

\begin{conj}
The twist of type I supergravity on $\CC^5$ is the $\Z/2$-graded theory with fields $\mu+u\nu, \sigma$ as above and with classical action
\beqn\label{eqn:typeIaction}
I_{{\rm type\, I}} = {\rm Tr}_{\CC^5} \frac{1}{1-\nu} \mu^2 \vee \sigma + \cdots
\eeqn
where the $\cdots$ stands for terms involving the higher descendant fields. 
\end{conj}

\parsec[s:typeIpot]

Like in the type IIA discussion, there is a slight modification of the type I model above which is most directly related to eleven-dimensional supergravity. 

This modification involves replacing the field $\sigma$ above by a potential $\til \gamma \in \Omega^{1,\bu}(\CC^5)$ which satisfies $\Omega \vee \sigma = \del \til \gamma$. 
This condition does not fix $\til \gamma$ uniquely, there is a gauge symmetry of the form $\til \gamma \to \til \gamma + \del \til \beta$. 

In detail, this potential theory we are considering has underlying cochain complex of fields
\begin{equation}
  \label{eq:Ipot} 
  \begin{tikzcd}[row sep = 1 ex]
    - & + & -\\ \hline
    \PV^{1,\bu}(\CC^5)_\mu \ar[r, "\div"] & \PV^{0,\bu}(\CC^5)_\nu  \\
         & \Omega^{0,\bu}(\CC^5)_{\til\beta} \ar[r, "\del"] & \Omega^{1,\bu}(\CC^5)_{\til\gamma} .
\end{tikzcd}
\end{equation} 
This space of fields is equipped with an odd nondegenerate pairing.
Like the eleven-dimensional theory, it is a classical BV theory in the $\ZZ/2$-graded sense. 

The type I action \eqref{eqn:typeIaction} completely determines the action of this theory with potentials. 
One simply takes the action and replaces all appearances of $\sigma$ with $\Omega^{-1} \vee \del \til\gamma$. 
This yields the interaction of the potential theory
\beqn\label{eqn:Iactionpot}
\til I_{\text{type I}} = \frac12 \int^\Omega_{\CC^5} \frac{1}{1-\nu} \mu^2 \vee \del \til\gamma .
\eeqn
Notice that the terms involving higher descendants vanishes since these fields are set to zero in the potential theory.

\subsection{Slab compactification}\label{s:Ired}

We consider placing twisted eleven-dimensional supergravity on the manifold $\CC^5 \times [0,1]$. 
In order to do this, we must choose appropriate boundary conditions at $t=0$ and $t=1$.
Our eleven-dimensional theory on such manifolds fits nicely into the formalism of \cite{BY,Eugene} in that it is topological in the direction transverse to the boundary.  

The phase space of the theory at $t=0$ or $t=1$ is 
\begin{equation}
  \label{eq:lin1} 
  \begin{tikzcd}[row sep = 1 ex]
    - & + \\ \hline
    \PV^{1,\bu}(\CC^5)_\mu \ar[r, "\div"] & \PV^{0,\bu}(\CC^5)_\nu \\ 
     \Omega^{0,\bu}(\CC^5)_\beta \ar[r, "\del"] & \Omega^{1,\bu}(\CC^5)_\gamma.
\end{tikzcd}
\end{equation}
The wedge an integrate pairing between the top and bottom lines induces an {\em even} symplectic structure on the phase space. 
Denote this phase space by $\cE_{\del}$ for the moment.

The phase space is equipped with the restriction of the linear BRST operator of the full eleven-dimensional theory. 
There is also a non linear BRST operator, just like in the bulk theory.
The BV action induces a $L_\infty$ structure on the parity shift $\Pi\cE_{\del}$ whose cohomology is still a trivial central extension of $E(5,10)$. 

A boundary condition is given by a Lagrangian subspace of $\cE_{\del}$ with respect to this even symplectic structure. 
To make sense of the theory on $\CC^5 \times [0,1]$ we must make the choice of two separate boundary conditions 
\[
\cM_{t=0} , \cM_{t=1} \subset \cE_{\del} .
\]
Moreover, these boundary conditions carry non linear BRST operators endowing their parity shifts $\Pi \cM_{t=0} , \Pi\cM_{t=1}$ with the structures of $L_\infty$ algebras. 
These $L_\infty$ structures must be compatible with the one on the phase space.
In fact, in our context these boundary conditions are abstractly isomorphic. 
We will explain the explicit boundary conditions momentarily. 

An important thing to note is that the fields of the theory compactified on the slab is computed by the {\em derived} intersection of the two Lagrangians:
\[
\cM_{t=0} \overset{\LL}{\underset{\cE_{\del}}{\times}}\cM_{t=1}.
\]
To compute this derived intersection we must suitably resolve the boundary conditions.

\parsec[s:boundary]
At $t=0$, the boundary condition of the eleven-dimensional theory is determined by declaring 
\[
\cM_{t=0}: \quad \gamma|_{t=0} = \beta|_{t=0} = 0 .
\]
We will place the theory on $\CC^5 \times [0,1]$ by imposing the same boundary condition at $t=1$:
\[
\cM_{t=1}: \quad \gamma|_{t=1} = \beta|_{t=1} = 0 .
\]

\begin{prop}
With these boundary conditions for the classical eleven-dimensional theory on $\CC^5 \times [0,1]$, the dimensional reduction along 
\[
\CC^5 \times [0,1] \to \CC^5
\]
is equivalent to the twist of type I supergravity on $\CC^5$. 
\end{prop}
\begin{proof}
Notice that both $\cM_{t=0}$ and $\cM_{t=0}$ are abstractly isomorphic to the complex resolving divergence-free vector fields
\begin{equation}
  \label{eq:lin2} 
  \begin{tikzcd}[row sep = 1 ex]
    - & + \\ \hline
    \PV^{1,\bu}(\CC^5)_\mu \ar[r, "\div"] & \PV^{0,\bu}(\CC^5)_\nu  .
\end{tikzcd}
\end{equation}

To compute the derived intersection between the two Lagrangians at $t=0$ and $t=1$ we replace the Lagrangian morphism $\cM_{t=0} \hookrightarrow \cE_{\del}$. 
Consider the cochain complex $\til \cM_{t=0}$ defined by
\begin{equation}
  \label{eq:lin3} 
  \begin{tikzcd}[row sep = 1 ex]
    - & + & - \\ \hline
    \PV^{1,\bu}(\CC^5)_\mu \ar[r, "\div"] & \PV^{0,\bu}(\CC^5)_\nu \\ 
     \Omega^{0,\bu}(\CC^5)_\beta \ar[dr,dotted,"\id"]\ar[r, "\del"] & \Omega^{1,\bu}(\CC^5)_\gamma \ar[dr,dotted,"\id"] \\
     & \Omega^{0,\bu}(\CC^5)_{\til\beta} \ar[r, "\del"] & \Omega^{1,\bu}(\CC^5)_{\til\gamma} .
\end{tikzcd}
\end{equation}
Notice that as a graded vector space, this complex is of the form $\cE_{\del} \oplus (\Omega^{0,\bu} \oplus \Pi \Omega^{1,\bu})$. 
The $L_\infty$ structure on $\Pi \til \cM_{t=0}$ extends the one on $\cE_{\del}$ coming from the bulk BV action. 
Notice that the obvious embedding $\cM_{t=0} \hookrightarrow \til \cM_{t=0}$ is a quasi-isomorphism.

The projection map $\til \cM_{t=0} \twoheadrightarrow \cE_{\del}$ factors the original Lagrangian inclusion as
\[
\cM_{t=0} \hookrightarrow \til \cM_{t=0} \twoheadrightarrow \cE_{\del} .
\]
To compute the derived intersection of $\cM_{t=0}$ and $\cM_{t=1}$ we can compute the ordinary intersection of $\til \cM_{t=0}$ and $\cM_{t=1}$. 

Let $\mu_{t=1}$ and $\nu_{t=1}$ denote the fields present in the other boundary condition $\cM_{t=1}$. 
The intersection $\til \cM_{t=0} \times_{\cE_{\del}} \cM_{t=1}$ is computed by setting the fields $\beta, \gamma$ to zero and $\mu=\mu_{t=1}$, $\nu = \nu_{t=1}$. 
Thus, we are left with
\begin{equation}
  \label{eq:lin3} 
  \begin{tikzcd}[row sep = 1 ex]
    - & + & -\\ \hline
    \PV^{1,\bu}(\CC^5)_\mu \ar[r, "\div"] & \PV^{0,\bu}(\CC^5)_\nu  \\
         & \Omega^{0,\bu}(\CC^5)_{\til\beta} \ar[r, "\del"] & \Omega^{1,\bu}(\CC^5)_{\til\gamma} 
\end{tikzcd}
\end{equation}
This is precisely the underlying cochain complex of fields for the type I model with potentials. 
The odd nondegenerate pairing on this complex agrees with the one on this particular potential theory for the twist of type I supergravity. 
The $L_\infty$ structure on the parity shift of this complex is compatible with the one induced from the BV action in \eqref{eqn:Iactionpot}.
\end{proof}

\subsection{The $SU(5)$ twist of type IIA supergravity}
\label{s:su5IIA}


Thus, given that our eleven-dimensional theory correctly describes the $SU(5)$-invariant twist of supergravity on $\CC^5 \times \RR$, to obtain the $SU(5)$ twist of type IIA supergravity we should reduce along the topological $\RR$ direction. 
This results in a $SU(5)$ invariant, holomorphic, theory on $\CC^5$. 

Let us briefly spell out the fields present in this dimensional reduction. 
The reduction is obtained by replacing $\Omega^\bu(\RR)$ with its translation invariant subalgebra $\CC[\ep] = \CC[\d t]$. 
Here, $\ep$ is an odd parameter playing the role of the translation invariant one-form $\d t \in \Omega^1(\RR)$. 
Equivalently, we are compactifying the theory along 
\[
\CC^5 \times S^1 \to \CC^5 .
\]

The field $\mu_{11d}$ is replaced by the field 
\[
\mu + \ep \mu' \in \Pi \PV^{1,\bu}(\CC^5) [\ep] .
\]
Notice that the lowest component of $\mu$ is odd (just like $\mu_{11d})$, but the lowest component of $\mu'$ is now even. 
Completely similarly, the remaining fields reduce as $\nu + \ep \nu'$, $\gamma + \ep \gamma'$, and $\beta + \ep \beta'$. 

In summary, the linear complex of fields of the dimensionally reduced theory on $\CC^5$ is
\begin{equation}
  \label{eqn:IIAsu5} 
  \begin{tikzcd}[row sep = 1 ex]
    {\rm odd} & {\rm even} & {\rm even} & {\rm odd} \\ \hline
    \PV^{1,\bu}(\CC^5)_{\mu} \ar[r, "\del"] & \PV^{0,\bu}(\CC^5)_\nu & \\ 
     & \ep\Omega^{0,\bu}(\CC^5)_{\beta'} \ar[r, "\div"] & \ep\Omega^{1,\bu}(\CC^5)_{\gamma'} . \\
     &  \ep \PV^{1,\bu}(\CC^5)_{\mu'} \ar[r, "\div"] & \ep \PV^{0,\bu}(\CC^5)_{\nu'} \\
     & & \Omega^{0,\bu}(\CC^5)_\beta \ar[r, "\div"] & \Omega^{1,\bu}(\CC^5)_\gamma.
\end{tikzcd}
\end{equation}

We can compute the dimensional reduction of the eleven-dimensional action $S_{BF,\infty} + J$ in a similar way to how we have done in the past few sections. 
We arrive at the action functional described below. 

\begin{conj}
\label{conj:IIAsu5}
The $SU(5)$ twist of type IIA supergravity on $\CC^5$ is equivalent to the theory whose linear BRST complex of fields is displayed in \eqref{eqn:IIAsu5}. 
The full action functional is 
\begin{multline}
\label{eqn:su5action}
\int^\Omega_{\CC^5}\bigg(\beta' \wedge \dbar \nu + \beta \wedge \dbar \nu' + \gamma' \wedge \dbar \mu + \gamma \wedge \dbar \mu' +  \beta' \wedge \div \mu + \beta \wedge \div \mu' \bigg) \\
+ \int^\Omega_{\CC^5} \bigg( \frac12 \frac{1}{1-\nu} \mu^2 \vee \del \gamma' +  \frac{1}{1-\nu} (\mu \wedge \mu') \vee \del \gamma' + \frac12 \frac{\nu'}{(1-\nu)^2} \mu^2 \vee \del \gamma \bigg) \\
+ \frac12 \int_{\CC^5} \gamma' \wedge \del \gamma \wedge \del \gamma .
\end{multline} 
\end{conj} 

The first two lines in \eqref{eqn:su5action} arise from the reduction of the BF action $S_{BF,\infty}$. 
The final line arises from the reduction of $J = \frac16 \int \gamma_{11d} \del \gamma_{11d} \del \gamma_{11d}$. 

%
%
%
%

\parsec[]\label{s:orbifold}

The slab compactification of the previous section was one way to implement the $S^1 / \ZZ/2$ reduction of the eleven-dimensional theory. 
We offer another point of view of this $S^1 / \ZZ/2$ reduction. 

First off, there is the following $\ZZ/2$ action on the eleven-dimensional theory on $\CC^5 \times S^1$ before compactifying. 
We obtain it by the following tensor product of $\ZZ/2$ actions. 
First, $\ZZ/2$ acts on $\Omega^\bu(S^1)$ by orientation reversing diffeomorphisms. 
Second, we declare that the eigenvalue of the $\ZZ/2$ action on $\PV^{k,\bu}(\CC^5)$, for $k=0,1$ is $+1$ and the eigenvalue of the $\ZZ/2$ action on $\Omega^{k,\bu}(\CC^5)$ for $k=0,1$ is $-1$. 
This determines a $\ZZ/2$ action on the full space of fields of the eleven-dimensional theory. 

Upon $S^1$ compactification the $\ZZ/2$ action is easy to describe: $\mu,\nu$ both have eigenvalue $+1$, $\mu',\nu'$ both have eigenvalue $-1$, $\gamma,\beta$ both have eigenvalue $-1$, and $\gamma',\beta'$ both have eigenvalue $+1$. 
In particular, we see that the $\ZZ/2$ fixed points simply pick out the $\mu, \nu, \gamma', \beta'$ fields; this comprises the first two lines of \eqref{eqn:IIAsu5}. 

The fields match precisely with the fields in the twist of type I supergravity that we recalled in \S \ref{s:typeIpot} (Under the relabeling $\gamma' \leftrightarrow \til \gamma, \beta' \leftrightarrow \til \beta$). 
Furthermore, restricting the action in the above conjecture agrees precisely with the action of this twisted type I model. 

\subsection{Compactification along a CY3}\label{s:CY3}

In the first section we saw that the eleven-dimensional theory can be defined on any manifold that is a product of a Calabi--Yau five-fold with a smooth oriented one-manifold. 
In this section, we investigate an important compactification of the eleven-dimensional theory which involves the Calabi--Yau manifold $X \times \CC^2$ where $X$ is a simply connected compact Calabi--Yau three-fold.

The compactification of the theory along the three-fold $X$ 
\[
X \times \CC^2 \times \RR \to \CC^2 \times \RR
\]
yields an effective five-dimensional theory which is holomorphic along $\CC^2$ and topological along $\RR$. 
Upon compactification, we will find a match with a description of the twist of five-dimensional minimally supersymmetric supergravity. 

\begin{prop}
\label{prop:5dsugra}
The compactification of the eleven-dimensional theory along a Calabi--Yau three-fold $X$ is equivalent to the twist of five-dimensional $\cN=1$ supergravity with $h^{1,1}(X)-1$ vector multiplets and $h^{1,2}(X) + 1$ hypermultiplets. 
\end{prop}

\parsec[s:5dsugra]

We give a conjectural capitulation of the twist of five-dimensional $\cN=1$ supergravity. 
Before twisting, a general five-dimensional $\cN=1$ supergravity contains a gravity multiplet coupled to some number of vector and hypermultiplets. 
The twist of the vector and hypermultiplet has been computed in \cite{ESW}, and we recall it below. 
The twist of the gravity multiplet is less clear. 
A thorough computation of the twist has yet to appear, though some checks have been established by Elliott and the last author in \cite{EWpoisson}. 
We give a description of the twist now, but leave a detailed computation from first principles to future work. 

The gravity multiplet, see \cite{CCDF} for instance, consists of a graviton $e$, a gravitino $\psi$, and a one-form gauge field $\cA_{grav}$.
After twisting, the graviton and components of the gravitino decompose into two Dolbeault-de Rham valued fields 
\[
\alpha, \eta \in \Pi \Omega^{0,\bu}(\CC^2) \otimes \Omega^\bu(\RR) ,
\]
whose lowest components both carry odd parity. 
The one-form gauge field $\cA_{grav}$ and the remaining components of the gravitino decompose into two more Dolbeault-de Rham valued fields
\[
A_{grav} , B_{grav} \in \Pi \Omega^{0,\bu}(\CC^2) \otimes \Omega^\bu(\RR) ,
\]
whose lowest components also both carry odd parity. 

\begin{conj}
\label{conj:5dsugra}
The twist of five-dimensional supergravity (with nonzero Chern--Simons term) with vector multiplets valued in a Lie algebra $\fg$ and hypermultiplets valued in a representation $V$ has BV fields
\begin{itemize}
\item $\alpha, A_{grav} \in \Pi \Omega^{0,\bu}(\CC^2) \otimes \Omega^\bu(\RR)$ with conjugate BV fields $\eta, B_{grav}$,
\item $A \in \Pi \Omega^{0,\bu}(\CC^2) \otimes \Omega^\bu(\RR) \otimes \fg$ with conjugate BV field $B$,
\item $\chi \in \Omega^{0,\bu}(\CC^2) \otimes \Omega^\bu(\RR) \otimes V$ with conjugate BV field $\psi$.
\end{itemize}

The action is
\begin{multline}
\label{eqn:5daction} 
\int^\Omega_{\CC^2 \times \RR} \left(\eta \dbar \alpha + B_{grav} \dbar A_{grav} + B \dbar A + \psi \dbar \chi \right) \\
  + \int^\Omega_{\CC^2 \times \RR} \left( \frac12\eta \{\alpha, \alpha\} +  B_{grav} \{\alpha, A_{grav}\}+ B \{\alpha, A\} +  \psi \{\alpha, \chi \}\right) \\ 
+ \frac16 \int_{\CC^2 \times \RR} B_{grav} \del B_{grav} \del B_{grav} .
\end{multline}
\end{conj}


\parsec[-]

With this description of the twist of five-dimensional supergravity, we turn to the proof of Proposition \ref{prop:5dsugra}.

First, we set up some notation. 
Let $\Omega_X$ be the holomorphic volume form on $X$. 
To define the eleven-dimensional theory on $X \times \CC^2 \times \RR$ we use the Calabi--Yau form $\Omega_X \wedge \d z_1 \wedge \d z_2$, where $\{z_i\}$ is a holomorphic coordinate on $\CC^2$. 
Let $\omega \in \Omega^{1,1}(X)$ be a fixed K\"ahler form on $X$.
For any $k$, let $H^k(X, \Omega^k_X)_\perp$ denote the cohomology of the primitive elements. 

\begin{proof}
Consider the eleven-dimensional field $\nu_{11d}$. 
Under the equivalence
\begin{align*}
\PV^{0,\bu}(X \times \CC^2) \otimes \Omega^\bu(\RR) & \simeq H^{\bu}(X, \cO) \otimes \PV^{0,\bu}(\CC^2) \otimes \Omega^\bu(\RR) \\ & = \PV^{0,\bu}(\CC^2) \otimes \Omega^\bu(\RR) \oplus \Pi \Bar{\Omega}_X \PV^{0,\bu}(\CC^2) \otimes \Omega^\bu(\RR) 
\end{align*}
the $\nu_{11d}$ field decomposes as 
\[ 
\nu_{11d} = \nu + \Bar{\Omega}_X \til \nu .
\]
Here $\Bar{\Omega}_X$ is the complex conjugate to the holomorphic volume form on $X$. 
Notice that the zero form component of $\til \nu$ is a field with even parity. 

Next, consider the eleven-dimensional field $\mu_{11d}$. 
Under the equivalence 
\begin{align*}
\Pi \PV^{1,\bu}(X \times \CC^2) \otimes \Omega^\bu(\RR) & \simeq \Pi H^{\bu}(X, \cO) \otimes \PV^{1,\bu}(\CC^2) \otimes \Omega^\bu(\RR) \\ & \oplus \Pi H^\bu(X, \T_X) \otimes \PV^{0,\bu}(\CC^2) \otimes \Omega^\bu(\RR) \\ & = \Pi \PV^{1,\bu}(\CC^2) \otimes \Omega^\bu(\RR) \oplus \Bar{\Omega}_X \PV^{1,\bu}(\CC^2) \otimes \Omega^\bu(\RR) \\ & \oplus H^{1}(X, \T_X) \otimes \PV^{0,\bu}(\CC^2) \otimes \Omega^\bu(\RR) \oplus \Pi H^{2}(X, \T_X) \otimes \PV^{0,\bu}(\CC^2) \otimes \Omega^\bu(\RR)
\end{align*}
the field $\mu_{11d}$ decomposes as 
\begin{align*}
\mu_{11d} & = \mu + \Bar{\Omega}_X \til \mu \\ 
& + e^i \chi_i + f^a A_a +  (\Omega_X^{-1} \vee \omega^2)A_{grav} .  
\end{align*} 
Here, $\{e^i\}_{i=1,\ldots, h^{2,1}}$ is a basis for $H^{1}(X, \T_X)$ and $\{f^a\}_{a=1,\ldots, h^{1,1}-1}$ is a basis for 
\[
H^2 (X, \Omega^2_X)_\perp \subset H^2(X, \Omega^2_X) \cong H^2(X, \T_X) . 
\]
Notice that the zero form part of $\til{\mu}$ is an even field, the zero form part of $\chi_i$ is an even field, the zero form part of $A_a$ is an odd field, and the zero form part of $\mu_\omega$ is an odd field. 

The decomposition for the eleven-dimensional fields $\gamma_{11d}$ and $\beta_{11d}$ is similar. 
We record it here:
\begin{align*}
\beta_{11d} & = \beta + \Bar{\Omega}_X \til \beta \\ 
\gamma_{11d} & = \gamma + \Bar{\Omega}_X \til \gamma + e_i \psi^i + f_a B^a + \omega \wedge B_{grav}  . 
\end{align*}
Here, $\{e_i\}_{i=1,\ldots,h^{2,1}}$ is a basis for $H^2 (X, \Omega^1_X)$ dual to the basis $\{e^i\}$ under the Serre pairing.
Also, $\{f_a\}_{a=1,\ldots,h^{1,1}-1}$ is a basis for $H^{1}(X, \Omega^1_X)_\perp$ dual to the basis $\{f^a\}$. 

To compare most directly to the description of the twist of five-dimensional $\cN=1$ supergravity we modestly modify the fields. 
Let $\del$ be the holomorphic de Rham differential along $\CC^2$. 
First, we introduce a potential for the fields $\mu$ and $\til \mu$. 
Let 
\[
\alpha, \chi \in \Omega^{0,\bu}(\CC^2) \otimes \Omega^\bu(\RR)
\]
be differential forms satisfying $\del \alpha = \mu \vee \Omega_{\CC^2}$ and $\del \chi = \til \mu \vee \Omega_{\CC^2}$. 
The fields $\nu, \til \nu$ are set to zero. 
Dually, we replace the fields $\gamma, \til \gamma$ their `field strengths', suitably renormalized with respect to the volume form
\[
\eta = (\d^2 z)^{-1} \vee \del \til \gamma , \quad \psi = (\d^2 z)^{-1} \vee \del \gamma \in \Omega^{0,\bu}(\CC^2) \otimes \Omega^\bu(\RR) .
\]
The roles of $\beta, \til \beta$ were as gauge symmetries implementing $\gamma \to \gamma + \del \beta$ and $\til \gamma \to \til \gamma + \del \til \beta$. 
Since we are replacing $\gamma, \til \gamma$ by their images under the operator $\del$, these gauge symmetries are set to zero. 

In summary, we are left with the following fields 
\[
\begin{array}{cccccccccc}
\alpha,A_{grav} & \in & \Pi \Omega^{0,\bu}(\CC^2) \otimes \Omega^\bu(\RR), & \eta, B_{grav} & \in & \Pi \Omega^{0,\bu}(\CC^2) \otimes \Omega^\bu(\RR) \\
\chi, \chi_i & \in & \Omega^{0,\bu}(\CC^2) \otimes \Omega^\bu(\RR),  & \psi, \psi^i & \in & \Omega^{0,\bu}(\CC^2) \otimes \Omega^\bu(\RR), & i=1,\ldots, h^{2,1}  \\
A_a & \in & \Pi \Omega^{0,\bu}(\CC^2) \otimes \Omega^\bu(\RR), & B^a & \in & \Pi \Omega^{0,\bu}(\CC^2) \otimes \Omega^\bu(\RR) , & a = 1, \ldots, h^{1,1}-1.
\end{array}
\]

Let us plug these fields in to the eleven-dimensional action.
First, consider the BF term $\frac12 \int^{\Omega} \frac{1}{1-\nu_{11d}} \mu_{11d}^2 \gamma_{11d}$. 
With the field redefinitions above, this decomposes as
\begin{multline}\label{eqn:5dsugra1}
 \int_{\CC^2 \times \RR}^\Omega \left( \frac12 \del \alpha \wedge \del \alpha \wedge \eta + \del A_{grav} \wedge \del A_{grav} \wedge B_{grav} \right) \\
 + \int_{\CC^2 \times \RR}^\Omega \left( \del \alpha \wedge \del \chi \wedge \psi + \del \alpha \wedge \del \chi_i \wedge \psi^i + \del \alpha \wedge \del A_a \wedge B^a \right) .
\end{multline}
This term agrees with the second line in the five-dimensional action \eqref{eqn:5daction}. 

Finally, consider the term in the eleven-dimensional action $J(\gamma_{11d}) = \frac16 \int \gamma_{11d} \wedge \del \gamma_{11d} \wedge \del \gamma_{11d}$. 
This induces the five-dimensional Chern--Simons term 
\beqn\label{eqn:5dsugra2}
\frac16 \int_{\CC^2 \times \RR} B_{grav} \del B_{grav} \del B_{grav} .
\eeqn
This completes the proof. 
\end{proof}

\parsec[s:5dglobal]

In \S \ref{sec:global} we computed the global symmetry algebra of the eleven-dimensional theory on $\CC^5 \times \RR$ and found a close relationship to the exceptional super Lie algebra $E(5,10)$. 
In this section we deduce the form of the global symmetry algebra of the five-dimensional compactified theory on $\CC^2 \times \RR$. 

Consider the full de Rham cohomology of $X$ by
\[
H^\bu (X, \Omega^\bu) = \oplus_{i,j} H^i (X, \Omega^j_X)  .
\]
This is a graded commutative algebra using the wedge product of differential forms. 
Next, consider the space of holomorphic functions $\cO(\CC^2)$ on $\CC^2$. 
The Poisson bracket $\{-,-,\}$ associated to the standard holomorphic symplectic structure on $\CC^2$ endows $\cO(\CC^2)$ with the structure of a Lie algebra. 
In particular, we can tensor $\cO(\CC^2)$ with $H^\bu (X, \Omega^\bu)$ to obtain the structure of a graded Lie algebra on 
\[
H^\bu (X, \Omega^\bu) \otimes \cO(\CC^2) .
\]
Let $[\omega] \in H^1(X, \Omega^1_X)$ be the class of the K\"ahler form on $X$.

The global symmetry algebra of the compactified theory along the Calabi--Yau three-fold $X$ is equivalent to a deformation of this graded Lie algebra. 
The deformation introduces the following Lie bracket 
\[
\big[ [\omega] \otimes f, [\omega] \otimes g\big] = [\omega^2] \otimes \{f,g\} \in H^{2,2}(X, \Omega^2) \otimes \cO(\CC^2) . 
\]

\section{Twisted supergravity on AdS space}
\label{sec:ads}

So far, we have mostly given evidence for the eleven-dimensional theory as a twist of supergravity in a flat background. 
We now turn to twisted versions of AdS backgrounds of eleven-dimensional supergravity. 

In M-theory, AdS backgrounds arise from backreacting some number $N$ of branes. 
For M2 branes, the backreacted geometry is ${\rm AdS}_4 \times S^7$.
For the M5 branes, the backreacted geometry is ${\rm AdS}_7 \times S^4$. 

According to the AdS/CFT correspondence, supergravity on such backgrounds should be dual to the relevant worldvolume theory in the large-$N$ limit. 
In this section, we do not directly refer to the worldvolume theories on the holomorphic twists of the M2 and M5 branes.
Rather, we identify the fields sourcing the branes at the level of the twisted eleven-dimensional theory.
In turn, we give a proposal for the twisted AdS background. 
We will show that the twist of the superconformal algebra is a global symmetry of this twisted background. 

\subsection{Superconformal algebras}

The complex form of the algebra of isometries for supergravity in both the ${\rm AdS}_4$ and ${\rm AdS}_7$ backgrounds is $\lie{osp}(8|2)$ (though, their real forms differ). 
This agrees with the complex form of the 6d $\cN=(2,0)$ superconformal algebra and the 3d $\cN=8$ superconformal algebra. 
The bosonic part of this algebra is isomorphic to $\lie{so}(8) \oplus \lie{sp}(2) \cong \lie{so}(8) \oplus \lie{so}(5)$. 

The minimal supercharge $Q$ acting on eleven-dimensional supersymmetry algebra is an element of this superconformal algebra. 
Its $Q$-cohomology is isomorphic to $\lie{osp}(6|1)$. (Twisted superconformal symmetry in six dimensions is studied in detail by the second  two authors  in~\cite{SWsuco2}.)
This super Lie algebra will play the role of the isometries in the twisted AdS background. 

\subsection{The ${\rm AdS}_4 \times S^7$ background}

In this section we introduce the analog of the ${\rm AdS}_4 \times S^7$ background in our conjectural description of the minimal twist of eleven-dimensional supergravity. 

\parsec[]

Decompose the eleven-dimensional manifold $\CC^5 \times \RR$ as
\[
 \CC^4_w\times \CC_z \times \RR .
\]

Analogous to before, the ${\rm AdS}_4 \times S^7$ background arises from backreacting M2 branes. Consider a stack of $N$ M2 branes wrapping $\R\times \C_z$. A natural interaction to consider is 
\[
I_{M2}(\gamma) = N\int_{\C_z} \gamma + \cdots
\] 
which is nonzero only on the component of $\gamma$ in $\Omega^1(\R)\otimes \Omega^{1,1}(\C^5)$. Unlike the case of M5 branes, the coupling does not involve choosing a primitive for a field strength---it is an electric coupling.
We have only indicated the lowest order coupling, the $\cdots$ indicate higher-order couplings which will be higher order in the fields of the eleven-dimensional theory and explicitly involve the fields in the worldvolume theory. 

This coupling is justified by comparison with the physical theory and by dimensional reduction. 
Indeed, as discussed in~\S\ref{s:components}, the component of $\gamma$ which participates in the above coupling is a component of the $C$-field of eleven dimensional supergravity. Thus, the proposal mirrors electric couplings of M2 branes in the physical theory, which simply involves integrating the $C$-field over the worldvolume of the brane. 

Moreover, reducing on a circle transverse to the M2 brane yields the $SU(4)$ twist of type IIA supergravity on $\R^2\times \C_z\times \C^3$ with $N$ $D2$ branes wrapping $\R\times \C_z$. As is shown in \cite{CLsugra}, an electric coupling of D2 branes to the $SU(4)$ twist of type IIA supergravity is given by 
\[
I_{D2}(\gamma) = N \int_{\R\times\C_z} \gamma + \cdots
\] 
where $\gamma$ now denotes the 1-form field of the $SU(4)$ twist of type IIA supergravity. It is immediate that the pullback of $I_{M2}$ along the map in the proof of proposition \ref{prop:dimred} recovers $I_{D2}$. 

\parsec[sec:m2backreact]

The backreacted geometry will be given by a solution to the equations of motion upon deforming the eleven-dimensional action by the interaction $I_{M2}(\gamma)$. 
Varying the deformed action with respect to $\gamma$,
we obtain the equation of motion
\beqn\label{eqn:ads4eom1}
\dbar \mu + \frac12 [\mu, \mu] + \partial\gamma\partial\gamma = N \Omega^{-1} \delta_{w=0} .
\eeqn
Here $[-,-]$ is the Schouten bracket. 
Varying $\beta$, we obtain the equation of motion
\beqn\label{eqn:adseom2}
\div \mu = 0 .
\eeqn

\begin{lem}
Let
\[
 F_{M2} = \frac{6}{(2\pi i)^4} \frac{\sum_{a=1}^4 \wbar_a \d \wbar_1 \cdots \Hat{\d \wbar_a} \cdots \d \wbar_4}{\|w\|^{8}} \partial_z .
\]
Then the background where $\mu = N F_{M2}$ and $\gamma = 0$
satisfies the above equations of motion in the presence of a stack of $N$ M2 branes:
\begin{align*}
\dbar (N F_{M2}) + \frac12 [N F_{M2}, N F_{M2}] & = N \Omega^{-1} \delta_{w=0} \\
\div (N F_{M2}) & = 0  .
\end{align*}
Here we set all components of the field $\gamma$ equal to zero (as well as the fields $\nu,\beta$). 
\end{lem}

\begin{proof}
Upon specializing $\gamma = 0$, the last term in the first equation above vanishes. The equation $\dbar F_{M2} = \Omega^{-1} \delta_{w=0}$ characterizes the Bochner--Martinelli kernel representing the residue class on $\CC^4 \, \setminus \, 0$. 
It is clear that $\div F_{M2} = 0$ and 
\[
[F_{M2}, F_{M2}] = 0
\] 
by simple type reasons. 
\end{proof}

\parsec[]

To provide evidence for the claim that this is the twisted analog of the AdS geometry, we will show that the twist of the symmetries present in the physical theory are witnessed in the twisted theory in this background. 

We have recalled that the $Q$-cohomology of $\lie{osp}(8|2)$ is isomorphic to the super Lie algebra $\lie{osp}(6|1)$. 
We will define an embedding of $\lie{osp}(6|1)$ into the eleven-dimensional theory on $\CC^5 \times \RR \setminus \{w=0\}$ which corresponds to the twist of the 3d superconformal algebra.
We first focus on the case where the flux $N=0$, for which the embedding can be extended to all of $\CC^5 \times \RR$. 

\parsec[] 

The bosonic part of $\lie{osp}(6|1)$ is the direct sum Lie algebra $\lie{sl}(4) \oplus \lie{sl}(2)$. 
The Lie algebra $\lie{sl}(2)$ represents (holomorphic) conformal transformations in $\CC_z$, which are inherited from the natural M\"obius group action on~$P^1(\C)$; the vector fields representing these transformations are not all divergence-free, and as such must be slightly adjusted. 
The Lie algebra $\lie{sl}(4)$ represents rotations along the plane $\CC^4_w$.   

\begin{itemize}[leftmargin=\parindent]
\item The bosonic summand $\lie{sl}(2)$ is mapped to the vector fields:
\[
\frac{\del}{\del z} ,\quad z \frac{\del}{\del z} - \frac14 \sum_{a=1}^4 w_a \frac{\del}{\del w_a} , \quad z \left(z \frac{\del}{\del z} - \frac12 \sum_{a=1}^4 w_a \frac{\del}{\del w_a} \right) \in \PV^{1,0}(\CC^5) \otimes \Omega^0(\RR) .
\]
Notice that these vector fields are divergence-free and reduce to the usual holomorphic conformal transformations along $w=0$.
\item The bosonic summand $\lie{sl}(4)$ is mapped to four-dimensional rotations: 
\[
\sum_{a,b=1}^4 B_{ab} w_a \frac{\del}{\del w_b} \in \PV^{1,0}(\CC^5) \otimes \Omega^0(\RR) , \quad (B_{ab}) \in \lie{sl}(4) .
\]
\end{itemize}

The odd part of the algebra $\lie{osp}(6|1)$ is $\wedge^4 W \otimes R$ where $W$ is the fundamental $\lie{sl}(4)$ representation and $R$ is the fundamental $\lie{sl}(2)$ representation. 
It is natural to split $R = \CC_{+1} \oplus \CC_{-1}$, so that the odd part decomposes as
\[
(\wedge^2 \CC^4)_{+1} \oplus (\wedge^2 \CC^4)_{-1} .
\]

\begin{itemize}[leftmargin=\parindent]
\item 
The fermionic summand $(\wedge^2 \CC^4)_{+1}$ consists of the supertranslations. 
It is mapped to the fields: 
\[
\frac{1}{2} (w_a \d w_b - w_b \d w_a) \in \Omega^{1,0}(\CC^5) \otimes \Omega^0(\RR) , \quad a,b=1,2,3,4 .
\] 
\item The fermionic summand $(\wedge^2 \CC^4)_{-1}$ consists of the remaining superconformal transformations. 
It is mapped to the fields: 
\[
\frac{1}{2} z (w_a \d w_b - w_b \d w_a) \in \Omega^{1,0}(\CC^5) \otimes \Omega^0(\RR) , \quad a,b=1,2,3,4. 
\] 
\end{itemize}

\begin{lem}\label{lem:m2emb}
These assignments define an embedding of $\lie{osp}(6|1)$ into the linearized BRST cohomology of the fields of the eleven-dimensional theory on $\CC^5 \times \RR$. 
Equivalently, it defines an embedding
\[
i_{M2} \colon \lie{osp}(6|1) \hookrightarrow E(5,10) .
\]
\end{lem} 
\begin{proof}
The second assertion follows from Theorem \ref{thm:global}, which shows that, as a super Lie algebra, the linearized BRST cohomology of the global symmetry algebra of the eleven-dimensional theory on $\CC^5 \times \RR$ is the trivial central extension of $E(5,10)$. 
Recall that the odd part of $E(5,10)$ is precisely the module of closed two-forms on $\CC^5$. 
To explicitly describe the embedding into $E(5,10)$ we simply apply the de Rham differential to the last two formulas above.
Recall, we are using the holomorphic coordinates $(z,w_1,\ldots,w_4)$ on $\CC^5$ where $z$ is the holomorphic coordinate along the M2 brane. 
\begin{itemize}[leftmargin=\parindent]
\item 
The fermionic summand $(\wedge^2 \CC^4)_{+1}$ embeds into closed two-forms as
\[
\d w_a \wedge \d w_b, \quad a,b=1,2,3,4. 
\] 
\item The fermionic summand $(\wedge^2 \CC^4)_{-1}$ embeds into closed two-forms as
\[
z \d w_a  \wedge \d w_b + \frac12 \d z \wedge (w_a \d w_b - w_b \d w_a) , \quad a,b=1,2,3,4. 
\] 
\end{itemize}
\end{proof}
\parsec[]

Next, we turn on $N \ne 0$ units of nontrivial flux. 
Since not all of the fields we wrote down above commute with the flux $N F_{M2}$, they are not compatible with the total differential $\delta^{(1)} + [N F_{M2}, -]$ acting on the fields supported on $\CC^5 \times \RR \setminus \{w=0\}$. 
Nevertheless, we have the following. 

\begin{prop}
\label{prop:brads4}
There exist $N$-dependent corrections to the fields defining the embedding of $\lie{osp}(6|1)$ summarized above which are closed for the modified BRST differential $\delta^{(1)} + [N F_{M2},-]$. 
Furthermore, these order $N$ corrections define an embedding of $\lie{osp}(6|1)$ inside the cohomology of the fields of eleven-dimensional theory on $\CC^5 \times \RR \setminus \CC \times \RR$ with respect to the differential $\delta^{(1)} + [N F_{M2},-]$.
\end{prop}

\begin{proof}
Let $\cL(\CC^5 \times \RR \setminus \{w=0\})$ denote the super $L_\infty$ algebra obtained by parity shifting the fields of the eleven-dimensional theory. 
We make the identification 
\[
(\CC^5 \times \RR) \setminus \{w=0\} \cong (\CC_w^4 \setminus 0) \times \CC_z \times \RR .
\]

Set $F = F_{M2}$ for notational convenience. Recall that we are viewing $F$ as an element of $\PV^{1,3}(\CC_w^4 \setminus 0) \otimes \Omega^{0,0}(\CC_z) \otimes \Omega^0(\RR)$. 
The operator $[F,-]$ acts on the fields according to two types of maps:
\begin{align*}
[F ,-] & \colon \PV^{i,\bu}(\CC^4_w \setminus 0) \otimes \PV^{j,\bu} (\CC_z) \otimes \Omega^\bu (\RR) \to \PV^{i,\bu+3}(\CC^4_w \setminus 0) \otimes \PV^{j,\bu} (\CC_z) \otimes \Omega^\bu (\RR) \\
[F,-] & \colon \Omega^{i,\bu}(\CC^4_w \setminus 0) \otimes \Omega^{j,\bu} (\CC_z) \otimes \Omega^\bu (\RR) \to \Omega^{i,\bu+3}(\CC^4_w \setminus 0) \otimes \Omega^{j,\bu} (\CC_z) \otimes \Omega^\bu (\RR).
\end{align*}


The first page of the spectral sequence is the cohomology with respect to the original linearized BRST differential $\delta^{(1)}$. 
Recall that the linearized BRST differential decomposes as
\[
\delta^{(1)} = \dbar + \d_{\RR} + \div |_{\mu \to \nu} + \del |_{\beta \to \gamma}  .
\]
To compute this page, we use an auxiliary spectral sequence which simply filters by the holomorphic form and polyvector field type. 
This first page of this auxiliary spectral sequence is simply given by the cohomology with respect to $\dbar + \d_{\RR}$. 
This cohomology is given by
\begin{equation}
  \label{eqn:ads4ss} 
  \begin{tikzcd}[row sep = 1 ex]
    + & - \\ \hline
H^\bu(\CC^4\setminus 0, \T) \otimes H^\bu(\CC, \cO) & H^\bu(\CC^4 \setminus 0, \cO) \otimes H^\bu(\CC, \cO) \\
H^\bu(\CC^4\setminus 0, \cO) \otimes H^\bu(\CC, \T) \\
H^\bu(\CC^4\setminus 0, \cO) \otimes H^\bu(\CC, \cO) & H^\bu(\CC^4\setminus 0, \cO) \otimes H^\bu(\CC, \Omega^1) \\ & H^\bu(\CC^4\setminus 0, \Omega^1) \otimes H^\bu(\CC, \cO)  
\end{tikzcd}
\end{equation}
where $\T$ denotes the holomorphic tangent sheaf, $\Omega^1$ denotes the sheaf of holomorphic one-forms, and $\cO$ is the sheaf of holomorphic functions.

The cohomology of $\CC$ is concentrated in degree zero and there is a dense embedding
\[
\CC[z] \hookrightarrow H^\bu(\CC, \cF) 
\]
for $\cF = \cO, \T$, or $\Omega^1$. 

For $\cF = \cO, \T$, or $\Omega^1$, the cohomology $H^\bu(\CC^4 \setminus 0, \cF)$ is concentrated in degrees $0$ and $3$. 
There are the following dense embeddings 
\begin{align*}
\CC[w_1,\ldots, w_4] & \hookrightarrow H^0(\CC^4 \setminus 0, \cO) \\ 
\CC[w_1,\ldots, w_4] \{\partial_{w_i}\} & \hookrightarrow H^0(\CC^4 \setminus 0, \T) \\
\CC[w_1,\ldots, w_4] \{\d w_i\} & \hookrightarrow H^0(\CC^4 \setminus 0, \Omega^1) 
\end{align*}
and
\begin{align*}
(w_1\cdots w_4)^{-1} \CC[w_1^{-1},\ldots, w_4^{-1}] & \hookrightarrow H^3(\CC^4 \setminus 0, \cO) \\ 
(w_1\cdots w_4)^{-1} \CC[w_1^{-1},\ldots, w_4^{-1}] \{\partial_{w_i}\} & \hookrightarrow H^3(\CC^4 \setminus 0, \T) \\
(w_1\cdots w_4)^{-1} \CC[w_1^{-1},\ldots, w_4^{-1}] \{\d w_i\} & \hookrightarrow H^3(\CC^4 \setminus 0, \Omega^1) .
\end{align*}

It follows that (up to completion) the cohomology 
\[
H^\bu(\cL(\CC^5 \times \RR \setminus \{w=0\}) , \dbar)
\]
is the direct sum of $H^\bu(\cL(\CC^5 \times \RR), \dbar)$ with 
\begin{equation}
  \label{eqn:ads4ss2} 
  \begin{tikzcd}[row sep = 1 ex]
    - & + \\ \hline
H^3(\CC^4\setminus 0, \cO)[z] \{\partial_{w_i}\}  \ar[r, dotted, "\div"] & H^3(\CC^4 \setminus 0, \cO) [z] \\
H^3(\CC^4\setminus 0, \cO) [z] \partial_z \ar[ur, dotted, "\div"'] \\
H^3(\CC^4\setminus 0, \cO) [z] \ar[r, dotted, "\del"] \ar[dr, dotted, "\del"'] & H^3(\CC^4\setminus 0, \cO)[z] \d z \\ & H^3(\CC^4\setminus 0, \Omega^1)[z] \{\d w_i\} .
\end{tikzcd}
\end{equation}
The remaining piece of the original BRST operator is drawn in dotted lines. 
The first page of the spectral sequence converging to the cohomology with respect to $\delta^{(1)} + [N F, -]$ is given by the cohomology of the global symmetry algebra on $\CC^5 \times \RR$, which we computed in \S \ref{sec:global}, plus the cohomology of the above complex with respect to the dotted-line operators. 
In this description, the image of the flux $F$ at this page in the spectral sequence corresponds to the class 
\[
[F] = (w_1 \cdots w_4)^{-1} \partial_z \in H^3(\CC^4\setminus 0, \cO) [z] \partial_z .
\]

The next page of the spectral sequence is given by computing the cohomology with respect to the operator $[N F,-]$. 
As observed above, this operator maps Dolbeault degree zero elements to Dolbeault degree three elements. 
For degree reasons, there are no further differentials and the spectral sequence collapses after the second page. 

The embedding of $\lie{osp}(6|1)$ we wrote down in lemma \ref{lem:m2emb} lands in the kernel of the original BRST operator $\delta^{(1)}$. 
To see that it this embedding can be lifted to the full cohomology we need to check that any element in the image of the original embedding is annihilated by $\big[ N [F] , - \big]$. 
This is a direct calculation. 
For instance, recall that an element in the image of the odd summand $(\wedge^2 \CC^2)_{-1}$ (which corresponds to a superconformal transformation) is of the form $z w_a \wedge \d w_b = z(w_a \d w_b - w_b \d w_a)$. 
We have
\[
\big[[F] , z(w_a \d w_b - w_b \d w_a) \big] = (w_1\cdots w_4)^{-1} (w_a \d w_b - w_b \d w_a) = 0
\]
since the class $(w_1\cdots w_4)^{-1}$ is in the kernel of the operator given by multiplication by $w_a$ for any $a = 1,\ldots 4$. 
\end{proof}

\subsection{The ${\rm AdS}_7 \times S^4$ background}

In this section we introduce the analog of the ${\rm AdS}_7 \times S^4$ background in our description of the minimal twist of eleven-dimensional supergravity. Decompose the eleven dimensional spacetime as $\C^3_z\times \C^2_w\times \R$.

\parsec[sec:m5coupling]

Analogous to the physical theory, the ${\rm AdS}_7 \times S^4$ background in the holomorphic twist will arise by backreacting M5 branes. To this effect, we begin by discussing how the eleven-dimensional theory couples to M5 branes. 
Consider a stack of $N$ M5 branes wrapping 
\[
\{w_1=w_2=t=0\} \subset \C^3_z\times \C^2_w\times \R.
\] 
It is natural to consider the nonlocal interaction 
\[
I_{M5} = N\int_{\C^3_z} \div^{-1}\mu \vee \Omega +\cdots 
\]
Note that this expression is only nonzero on the component of $\mu$ in $\PV^{1,3}$. 
We argue that this coupling is consistent with expectations from the physical theory and from dimensional reduction. 

The twisted field $\mu^{1,3}$ is a component of the Hodge star of the $G$-flux in the physical theory (\S\ref{s:components}). 
In the physical theory, M5 branes magnetically couple to the $C$-field; the coupling involves choosing a primitive for the Hodge star of the $G$-flux and integrating it over the M5 worldvolume. Our twist contains no fields corresponding to components of such a primitive; hence such a magnetic coupling is reflected in the appearance of $\div^{-1}$. 

\parsec[]

We obtain a deeper justification for this coupling through dimensional reduction to type IIA supergravity. 
Reducing on the circle along the directions the M5 branes wrap yields the $SU(4)$ invariant twist of type IIA supergravity on $\CC^4 \times \RR^2$ with $N$ $D4$ branes wrapping $\CC^2 \times \RR$. 

In \cite{CLsugra}, it is shown that the magnetic coupling of $D4$ branes to the $SU(4)$ twist of IIA is of the form
\[
N \int _{\C^2 \times \RR} \div^{-1} \mu \vee \Omega_{\C^4} + \cdots .
\]
Again, we have only explicitly indicated the first-order piece of the coupling. 

\parsec[s:m5backreact]

The backreacted geometry will be given by a solution to the equations of motion upon deforming the eleven-dimensional action by the interaction $I_{M5}(\mu)$. 

Varying the potential $\div^{-1} \mu$, we obtain the following equation of motion involving the field $\gamma$:
\beqn\label{eqn:m5eom1}
\dbar \del \gamma + \div \left(\frac{1}{1-\nu} \mu\right) \wedge \del \gamma = N \delta_{w_1=w_2=t=0} .
\eeqn
Notice that there is an extra derivative compared to the equation of motion arising from varying the field $\mu$. 
This equation only depends on $\gamma$ through its field strength $\del \gamma$. 

Varying $\gamma$ we obtain the equation of motion 
\beqn\label{eqn:m5eom2}
(\dbar + \d_\RR) \mu + \del \gamma \del \gamma = 0 .
\eeqn 
Again, this only depends on $\gamma$ through its field strength $\del \gamma$.

\begin{lem}
\label{lem:ads7flux}
Let
\[
F_{M5} = \frac{1}{(2\pi i)^3} \frac{\wbar_1 \d \wbar_2 \wedge \d t - \wbar_2 \d \wbar_1 \wedge \d t + t \d \wbar_1 \wedge \d \wbar_2}{(\|w\|^2 + t^2)^{5/2}} \wedge \d w_1 \wedge \d w_2
\]
Then, $\del\gamma = N F_{M5}$, $\mu = 0$, and $\nu = 0$ satisfies the equations of motion in the presence of a stack of $N$ M5 branes sourced by the term $N \delta_{w_1=w_2=t=0}$:
\begin{align*}
\dbar (NF_{M5}) + \d_{\RR} (NF_{M5}) & = N \delta_{w_1=w_2=t=0}  \\ 
(NF_{M5}) \wedge (NF_{M5}) & = 0 .
\end{align*}
Here, we set all components of the field $\mu$ equal to zero (as well as the fields $\nu,\beta$). 
\end{lem}

\begin{proof}
The first equation,
\[
\dbar F + \d_{\RR} F = N \delta_{w_1=w_2=t=0},
\]
characterizes the kernel representing $N$ times the residue class for a four-sphere in 
\[
(\CC^2 \times \RR) \setminus 0 \simeq S^4 \times \RR .
\] 
That is
\[
\oint_{S^4} N F = N 
\]
for any four-sphere centered at $0 \in \CC^2 \times \RR$.

The second equation $F \wedge F = 0$ follows by simple type reasons. 
\end{proof}

\parsec[s:m5embedding]

To provide evidence for the claim that this is the twisted analog of the AdS geometry, we will again show that the twist of the symmetries present in the physical theory on ${\rm AdS}_7 \times S^4$ appear in the twisted theory on this background. 

We have recalled that the $Q$-cohomology of $\lie{osp}(8|2)$ is isomorphic to the super Lie algebra $\lie{osp}(6|1)$. 
We will define an embedding of $\lie{osp}(6|1)$ into the eleven-dimensional theory on $\CC^5 \times \RR \setminus \{w_1=w_2=t=0\}$ which corresponds to the twist of the 6d superconformal algebra.

We first focus on the case where the flux $N=0$.
In this case the embedding can be extended to all of $\CC^5 \times \RR$.

Recall that we have chosen coordinates of the form
\[
\CC^5 \times \RR = \CC_z^3 \times \CC_w^2 \times \RR_t
\]
with $z_i, i=1,2,3$ and $w_a, a=1,2$.
The stack of M5 branes wrap $w_1=w_2=t=0$. 

The embedding of the bosonic piece of $\lie{osp}(6|1)$ can be described as follows. Recall that the bosonic part of $\lie{osp}(6|1)$ is the direct sum Lie algebra
\[
\lie{sl}(4) \oplus \lie{sl}(2) .
\]
which we write as $\lie{sl}(W) \oplus \lie{sl}(R)$ with $W,R$ complex four, two dimensional complex vector spaces. The roles of the $\lie{sl}(4)$ and $\lie{sl}(2)$ summands are interchanged compared to the case of the M2 brane. 
The Lie algebra $\lie{sl}(4)$ represents (holomorphic) conformal transformations along the plane $\CC^3_z$, again coming from the natural action on $P^3(\C)$.
Since not all such infinitesimal transformations are divergence-free, the precise vector fields must be adjusted.   
The Lie algebra $\lie{sl}(2)$ represents rotations in $\CC^2_w$; the vector fields representing these transformations are automatically divergence-free.
In more detail, the embedding of the bosonic piece can be given by the following explicit formulas. 

\begin{itemize}[leftmargin=\parindent]
\item
The bosonic abelian subalgebra $\CC^3 \subset \lie{sl}(4)$ of translations is mapped to the obvious vector fields 
\[
\frac{\del}{\del z_i} \in \PV^{1,0}(\CC^5) \otimes \Omega^0(\RR) , \quad i=1,2,3.
\]

\item
The bosonic subalgebra $\lie{sl}(3) \subset \lie{sl}(4)$ is mapped to the 
rotations
\[
A_{ij} z_i \frac{\del}{\del z_j} \in \PV^{1,0}(\CC^5)\otimes \Omega^0(\RR) , \quad (A_{ij}) \in \lie{sl}(3) .
\]

\item
The bosonic subalgebra $\CC \subset \lie{sl}(4)$ corresponding to rescaling $\C^3$ is mapped to the element
\[
\sum_{i=1}^3 z_i \frac{\del}{\del z_i} - \frac32 \sum_{a=1}^2 w_a \frac{\del}{\del w_a} \in \PV^{1,0}(\CC^5) \otimes \Omega^0(\RR)  .
\] 
Notice that this vector field are divergence-free and restricts to the ordinary dilation (Euler vector field) along $w=0$. 
\item 
The bosonic subalgebra of $\lie{sl}(4)$ describing special conformal transformations on $\CC^3$ is mapped to the elements 
\[
z_j \left(\sum_{i=1}^3 z_i \frac{\del}{\del z_i} - 2 \sum_{a=1}^2 w_a \frac{\del}{\del w_a} \right) \in \PV^{1,0}(\CC^5) \otimes \Omega^0(\RR) .
\] 
Notice that these vector fields are divergence-free and restrict to the ordinary special conformal transformations along $w=0$. 
\item 
The bosonic summand $\lie{sl}(2)$ ($R$-symmetry) is mapped to the triple
\[
w_1 \frac{\del}{\del w_2}, w_2 \frac{\del}{\del w_1}, \frac12 \left(w_1 \frac{\del}{\del w_1} - w_2 \frac{\del}{\del w_2}\right) \in \PV^{1,0}(\CC^5) \otimes \Omega^0(\RR) .
\]
\end{itemize}

The odd part of the algebra $\lie{osp}(6|1)$ is $\wedge^4 W \otimes R$ where $W$ is the fundamental $\lie{sl}(4)$ representation and $R$ is the fundamental $\lie{sl}(2)$ representation. 
It is natural to split $W = L \oplus \CC$ with $L = \CC^3$ the fundamental $\lie{sl}(3) \subset \lie{sl}(4)$ representation. 
Then the odd part decomposes as
\[
L \otimes R \oplus \wedge^2 L \otimes R \cong \CC^3 \otimes \CC^2 \oplus \wedge^2 \CC^3 \otimes \CC .
\]

\begin{itemize} 
\item The summand $L \otimes R$ consists of the remaining 6d supertranslations. 
It is mapped to the fields 
\[
z_i \d w_a \in \Omega^{1,0}(\CC^5) \otimes \Omega^0(\RR) ,\quad a=1,2, \quad i =1,2,3.
\] 
\item The summand $\wedge^2 L \otimes R$ consists of the remaining 6d superconformal transformations. 
It is mapped to the fields
\[
\frac12 w_a (z_i \d z_j - z_j \d z_i) \in \Omega^{1,0}(\CC^5)\otimes \Omega^0(\RR) , \quad a = 1,2, \quad k = 1,2,3. 
\]
\end{itemize}

\begin{lem}
These assignments define an embedding of $\lie{osp}(6|1)$ into the linearized BRST cohomology of the fields of the eleven-dimensional theory on $\CC^5 \times \RR$. 
Equivalently, it defines an embedding
\[
i_{M5} \colon \lie{osp}(6|1) \hookrightarrow E(5,10) .
\]
\end{lem} 

\begin{proof}
To explicitly describe the embedding into $E(5,10)$ we simply apply the de Rham differential to the last two formulas above.
Recall, we are using the holomorphic coordinates $(z_1,z_2,z_3, w_1,w_2)$ on $\CC^5$ where $z_i$ are the holomorphic coordinates along the M5 brane. 
\begin{itemize}
\item 
The fermionic summand $L \otimes R$ embeds into closed two-forms as
\[
\d z_i \wedge \d w_a, \quad i=1,2,3, \quad a=1,2. 
\] 
\item 
The fermionic summand $\wedge^2 L \otimes R$ embeds into closed two-forms as
\[
w_a \d z_i  \wedge \d z_j + \frac12 \d w_a \wedge (z_i \d z_j - z_j \d z_i) , \quad i,j=1,2,3, \quad a=1,2. 
\] 
\end{itemize}
\end{proof}
\parsec[]

Next, we turn on $N \ne 0$  units of nontrivial flux. 
Since not all of the fields we wrote down above commute with the flux $N F$, they are not compatible with the total differential $\delta^{(1)} + [N F, -]$ acting on the fields supported on $\CC^5 \times \RR \setminus \{w_1=w_2=t=0\}$. 
Nevertheless, we have the following.

\begin{prop}
\label{prop:brads7}
There exist $N$-dependent corrections to the embedding $i_{M5}$ which are compatible with the modified BRST differential $\delta^{(1)} + [N F_{M5},-]$. 
Furthermore, these $N$-dependent corrections define an embedding of $\lie{osp}(6|1)$ inside the cohomology of the fields of eleven-dimensional theory on $\CC^5 \times \RR \setminus \CC \times \RR$ with respect to the differential $\delta^{(1)} + [N  F_{M5},-]$.
\end{prop}

\parsec[s:thfcohomology]

The proof of the above proposition follows from another indirect cohomological argument. 
Before getting to the proof, we introduce the relevant cohomology. 

The eleven-dimensional theory is built from fields which live in the tensor product of complexes 
\[
\Omega^{0,\bu}(\CC^5) \otimes \Omega^\bu(\RR).
\]
More precisely, this is where the  fields $\beta$ and~$\nu$ live. 
The $\mu$ and~$\gamma$ fields live in versions of this complex where we take Dolbeault forms with coefficients in the holomorphic tangent and cotangent bundles, respectively. 

Another way to think about this complex is to first consider the full de Rham complex $\Omega^\bu(\CC^5 \times \RR)$, which includes both holomorphic and anti-holomorphic forms in the $\CC^5$ direction. 
The dg algebra of all de Rham forms on $\CC^5 \times \RR$ has an ideal generated by the holomorphic one forms $\{\d z_i\}_{i=1,\ldots,5}$.
There is an isomorphism of dg algebras
\[
\Omega^{0,\bu}(\CC^5) \otimes \Omega^\bu(\RR) \cong \Omega^\bu(\CC^5 \times \RR) \, / \, (\d z_1,\ldots, \d z_5) .
\]
The advantage of this presentation is that we can define a complex associated to more general manifolds that are not products of a complex manifold with a smooth manifold.\footnote{More generally, we are describing the cohomology of a manifold equipped with a topological holomorphic foliation.}

For the M5 brane, it was convenient to rename the holomorphic coordinates on $\CC^5$ to $z_1,z_2,z_3,w_1,w_2$. 
At the twisted level, the geometry arising from backreacting M5 branes is based on the manifold 
\[
\CC^5 \times \RR \setminus \CC^3 \cong \CC_z^3 \times (\CC^2_w \times \RR \setminus 0) .
\]
The $\beta$ and~$\nu$ fields of the eleven-dimensional theory on this submanifold of $\CC^5 \times \RR$ live in the complex 
\[
\Omega^\bu\bigg(\CC^5 \times \RR \setminus \CC^3\bigg) \, / \, (\d z_1,\d z_2,\d z_3, \d w_1, \d w_2)  .
\]
The $\mu$ and~$\gamma$ fields live in similar complexes, where we introduce a (trivial) vector bundle on $\CC^5 \times \RR \setminus \CC^3$. 

Since the $\CC^3$ wraps $w_1=w_2=t=0$ we can apply a version of the K\"unneth formula to identify this complex with 
\[
\Omega^{0,\bu}(\CC^3_z) \otimes \bigg( \Omega^\bu\left(\CC^2_w \times \RR \setminus 0 \right) \, / \, (\d w_1, \d w_2) \bigg).
\]

The cohomology of the Dolbeault complex of $\CC^3_z$ is easy to compute. 
The cohomology of the bit in parentheses is concentrated in degrees zero and two. 
In degree zero, there is a dense embedding
\[
\CC[w_1,w_2] \hookrightarrow H^0 \bigg( \Omega^\bu\left(\CC^2_w \times \RR \setminus 0 \right) \, / \, (\d w_1, \d w_2) \bigg)
\]
In degree two, there is a dense embedding
\[
w_{1}^{-1} w_2^{-1} \CC[w_1,w_2] \hookrightarrow H^2 \bigg( \Omega^\bu\left(\CC^2_w \times \RR \setminus 0 \right) \, / \, (\d w_1, \d w_2) \bigg).
\]

It will be useful to explain this last embedding in more detail. 
Consider the homogenous element $w_1^{-1} w_2^{-1}$. 
This represents the class of the Dolbeault-de Rham two-form
\[
\frac{\wbar_1 \d \wbar_2 \wedge \d t - \wbar_2 \d \wbar_1 \wedge \d t + t \d \wbar_1 \wedge \d \wbar_2}{(\|w\|^2 + t^2)^{5/2}} .
\]
Notice that, if we wedge with the volume form $\d w_1 \d w_2$, this is the unit  flux ($N=1$) introduced in Lemma \ref{lem:ads7flux}. 
The homogenous element $w_1^{-n-1} w_2^{-m-1}$ represents the class of the holomorphic derivatives $\partial_{w_1}^n \partial_{w_2}^{m}$ applied to this two-form. 

Observe that, when restricted to $\CC^5 \times \RR \setminus \CC^3$, the holomorphic tangent bundle along $\CC^5$ is still trivializable. 

\parsec[]

Let's turn to the proof of Proposition~\ref{prop:brads7}.
We proceed completely analogously to the case of backreacted M2 branes as in the proof of Proposition \ref{prop:brads4}. 

\begin{proof}[Proof of Proposition \ref{prop:brads7}]
Let $\cL (\CC^5 \times \RR \setminus \{w_1=w_2=t=0\})$ denote the super $L_\infty$ algebra obtained by parity shifting the fields of the eleven-dimensional theory on $\CC^5 \times \RR \setminus \{w_1=w_2=t=0\}$. 

There is a spectral sequence which converges to the cohomology of the fields with respect to the deformed linear BRST differential $\delta^{(1)} + [N F_{M5},-]$ whose first page
is the cohomology with respect to the original linearized BRST differential $\delta^{(1)}$. 
Recall that the linearized BRST differential decomposes as
\[
\delta^{(1)} = \dbar + \d_{\RR} + \div |_{\mu \to \nu} + \del |_{\beta \to \gamma}  .
\]
To compute this page, we use an auxiliary spectral sequence which simply filters by the holomorphic form and polyvector field type. 
This first page of this auxiliary spectral sequence is simply given by the cohomology of the fields supported on 
\[
\CC^5 \times \RR \setminus \{w_1=w_2=t=0\} \cong \CC_z^3 \times (\CC^2_w \times \RR \setminus 0)
\]
with respect to $\dbar + \d_{\RR}$. 

To compute this cohomology we follow the discussion in \S \ref{s:thfcohomology}.
Just as in the case of the M2 brane, we see that the $\dbar + \d_{\RR}$ cohomology is (up to completions) is the direct sum of the cohomology on flat space $H^\bu(\cL(\CC^5 \times \RR), \dbar)$ with
\begin{equation}
  \label{eqn:ads7ss2} 
  \begin{tikzcd}[row sep = 1 ex]
    + & - \\ \hline
w_1^{-1} w_2^{-1} \CC[w_1^{-1}, w_2^{-1}][z_1,z_2,z_3] \{\partial_{w_i}\}  \ar[r, dotted, "\div"] & w_1^{-1} w_2^{-1} \CC[w_1^{-1}, w_2^{-1}] [z_1,z_2,z_3] \\
w_1^{-1} w_2^{-1} \CC[w_1^{-1}, w_2^{-1}] [z_1,z_2,z_3] \{\del_{z_i}\} \ar[ur, dotted, "\div"'] \\
w_1^{-1} w_2^{-1} \CC[w_1^{-1}, w_2^{-1}] [z_1,z_2,z_3] \ar[r, dotted, "\del"] \ar[dr, dotted, "\del"'] & w_1^{-1} w_2^{-1} \CC[w_1^{-1}, w_2^{-1}][z_1,z_2,z_3] \{\d z_i\} \\ & w_1^{-1} w_2^{-1} \CC[w_1^{-1}, w_2^{-1}][z_1,z_2,z_3] \{\d w_i\} .
\end{tikzcd}
\end{equation}

Recall that the flux $F$ was defined as the image under $\del$ of some $\gamma$-type field. 
Therefore, the class $[F]$ does not live inside this page of the spectral sequence, but the operator $[[F], -]$ does act on this page nevertheless. 
For instance, if $f^i(z,w) \d z_i$ is a one-form living in $H^0(\CC^5, \Omega^1) \otimes H^0(\RR)$, then
\[
[ [F] , f^i (z,w) \d z_i ] = \ep_{ijk} w_1^{-1} w_2^{-1} \partial_{z_j} f^i(z,w) \del_{z_k} 
\]
which is an element in 
\[
\CC[w_1^{-1}, w_2^{-1}][z_1,z_2,z_3] \{\del_{z_i}\} \subset H^0(\CC^3, \T) \otimes H^2 \big(\Omega^\bu(\CC^2 \times \RR \setminus 0) / (\d w_1 , \d w_2) \big) .
\]

The first page of the spectral sequence converging to the cohomology with respect to $\delta^{(1)} + [N F, -]$ is given by the cohomology of the global symmetry algebra on $\CC^5 \times \RR$, which we computed in \S \ref{sec:global}, plus the cohomology with respect to the dotted-line operators in~\eqref{eqn:ads7ss2}. 

The next page of the spectral sequence is given by computing the cohomology with respect to the operator $[N F,-]$. 
This operator maps Dolbeault-de Rham degree zero elements to Dolbeault-de Rham degree two elements. 
For degree reasons, there are no further differentials and the spectral sequence collapses after the second page. 

The embedding of $\lie{osp}(6|1)$ for $N=0$ lives in the kernel of the original BRST operator $\delta^{(1)}$. 
To see that it this embedding can be lifted to the full cohomology we need to check that any element in the image of the original embedding is annihilated by $\big[ N [F] , - \big]$. 
This is a direct calculation. 
For instance, recall that an element in the image of the odd summand $\wedge^2 L \otimes R = \wedge^2 \CC^3 \otimes \CC^2$ (which corresponds to a superconformal transformation) is of the form $w_a (z_i \d z_j - z_j \d z_i)$, $a=1,2, i,j=1,2,3$. 
We have
\[
\big[[F] , w_a (z_i \d z_j - z_j \d z_i)\big] = 2 \ep_{ijk} (w_1^{-1} w_2^{-1}) \cdot w_a \del_{z_k} = 0
\]
since the class $w_1^{-1} w_2^{-1}$ is in the kernel of the operator given by multiplication by $w_a$ for $a=1,2$.
Verifying that the remaining elements in the image of $i_{M5}$ are in the kernel of $\big[ [F], -\big]$ is similar.
This completes the proof.
\end{proof}

\printbibliography

\end{document}